\documentclass[12pt]{article}
\usepackage[utf8]{inputenc}
\usepackage[a4paper,top=3.2cm,width=17cm, bottom=3.3cm, left=2.3cm]{geometry}
\usepackage{array}
\usepackage{physics}
\usepackage{longtable}
\usepackage{rotating}
\usepackage{booktabs}
\usepackage{subcaption}
\usepackage{float}
\usepackage{textcomp}
\usepackage{multirow}
\usepackage{amsmath}
\usepackage{amssymb}
\usepackage{amsthm}
\usepackage{relsize}
\usepackage{graphicx}
\usepackage{dsfont}
\usepackage{comment}
\usepackage{adjustbox}
\usepackage{multirow}
\usepackage{extarrows}\usepackage{pdflscape}
\usepackage{tikz,wrapfig,pgfplots,xcolor}
\usepackage{soul}
\usetikzlibrary{arrows.meta, positioning,calc}
\font\myfont=cmr12 at 11pt 
\font\myfonttwo=cmr12 at 10.5pt 
\usetikzlibrary{positioning, arrows.meta, shapes}

\usepackage{color}
\definecolor{purple}{rgb}{0.7,0,0.7}
\definecolor{grijs1}{cmyk}{0.73,.65,.6,.16}

\usepackage{hyperref}
\hypersetup{
    colorlinks=true,
    linkcolor=blue,
    filecolor=magenta,      
    urlcolor=cyan,
    citecolor=blue,
    hypertexnames=false, 
destlabel=true
}

\numberwithin{equation}{section}

\newtheorem{thm}{Theorem}
\newtheorem{prop}{Proposition}
\newtheorem{cor}{Corollary}
\newtheorem{lem}{Lemma}
\newtheorem{rmk}{Remark}
\newtheorem{dfn}{Definition}
\newtheorem{con}{Conjecture}

\newcommand\sn{\mathsf{n}}

\newcommand\snu{\boldsymbol{\nu}}
\newcommand\srho{\boldsymbol{\rho}}

\newcommand\sx{\mathsf{x}}
\newcommand\sy{\mathsf{y}}
\newcommand\sa{\mathsf{a}}

\newcommand\ssb{\mathsf{b}}
\newcommand\ssc{\mathsf{c}}
\newcommand\ssd{\mathsf{d}}
\newcommand\ssl{\mathsf{l}}
\newcommand\sz{\mathsf{z}}

\newcommand\sM{{\mathsf M}}

\newcommand\ex{\rm{e}}
\newcommand\ha{\frac{1}{2}}

\newcommand{\bnu}{\boldsymbol \nu}
\newcommand\NN{\mathbb{N}}
\newcommand\CC{\mathbb{C}}
\newcommand\HH{\mathbb H}
\newcommand\RR{\mathbb R}
\newcommand\QQ{\mathbb{Q}}
\newcommand\ZZ{\mathbb{Z}}
\newcommand{\sign}{\operatorname{sgn}}
\newcommand{\amk}{a_{-m}(K;q)}
\newcommand{\asub}{\mathsf{a}_{\boldsymbol \nu,\hat \epsilon}}
\newcommand{\rsub}{r_{\boldsymbol \nu,\hat \epsilon}}
\newcommand{\Xsub}{X_{\ell,\hat \epsilon}}
\newcommand{\Asub}{A_{\ell,\hat \epsilon}}

\newcommand{\normsq}[1]{\left|#1\right|^{2}}
\newcommand{\inner}[2]{B(#1,#2)}

\newcommand{\Mod}[1]{\ (\mathrm{mod}\ #1)}

\def\be{\begin{equation}}
\def\ee{\end{equation}}
\def\bem{\begin{matrix}}
\def\eem{\end{matrix}}

\usepackage{authblk}

\title{3d Modularity Revisited}

\author[1,2,3]{\myfont Miranda C. N. Cheng\thanks{\myfonttwo  
On leave from CNRS, France.}}

\author[4,5]{\myfont Ioana Coman 
}
\author[6]{\myfont Piotr Kucharski 
}

\author[1]{\myfont Davide Passaro 
}

\author[1]{\myfont Gabriele Sgroi }

\affil[1]{\myfont Institute of Physics, University of Amsterdam, the Netherlands}

\affil[2]{\myfont Institute for Mathematics, Academia Sinica, Taipei, Taiwan}

\affil[3]{\myfont Korteweg-de Vries Institute for Mathematics, University of Amsterdam, the Netherlands}

\affil[4]{\myfont Kavli Institute for the Physics and Mathematics of the Universe, University of Tokyo, Kashiwa, Japan}

\affil[5]{\myfont School of Mathematics, University of Edinburgh, Edinburgh EH9 3FD, UK}

\affil[6]{\myfont Institute of Mathematics, University of Warsaw, ul. Banacha 2, 02-097 Warsaw, Poland}

\date{}
\begin{document}

\maketitle

  \abstract{\myfonttwo{
The three-manifold topological invariants $\hat Z$ capture the half-index of the three-dimensional theory with ${\cal N}=2$ supersymmetry obtained by compactifying the M5 brane theory on the closed three-manifold. In 2019, surprising general relations between the  $\hat Z$-invariants, quantum modular forms, and vertex algebras have been proposed. In the meanwhile, an extensive array of examples have been studied, but several important general structural questions remain.  
{First, for many three-manifolds we have seen hints of concrete  $\widetilde {\rm SL}_2(\ZZ)$ representations underlying the different $\hat Z$-invariants of the given manifolds. At the same time, these invariants appear to only span a subspace of the representation, and the role of the latter remains mysterious. 
We elucidate the meaning of the modular group representation, realized as vector-valued quantum modular forms, by first proposing the analogue $\hat Z$-invariants with supersymmetric defects, and subsequently showing that the full vector-valued quantum modular form for $\widetilde {\rm SL}_2(\ZZ)$ is precisely the object capturing all the $\hat Z$-invariants of a given three-manifold, when the newly defined defects $\hat Z$-invariants are included.} Second, it was expected that matching radial limits is a key feature of $\hat Z$-invariants when changing the orientation of the plumbed three-manifold, suggesting the relevance of mock modularity. {We substantiate the conjecture by providing explicit proposals for such $\hat Z$-invariants for three three-manifolds and verify their mock modularity and limits. } Third, we initiate the study of the vertex algebra structure of the mock type invariants by showcasing a systematic way to construct cone vertex operator algebras associated to these mock invariants, which can be viewed as the partner of logarithmic vertex operator algebras in this context.  \\

\newpage
\tableofcontents

\newpage
\section{Introduction, Background, and Summary}\label{sec:intro}

In this section we briefly summarize the intricate connections between quantum modular forms, vertex operator algebras (VOAs), and three-manifold topological  $\hat Z$- and $F_K$- invariants that have been studied in recent years, discussing their context in physics, number theory, and topology, and highlighting some of the aspects that are at present still mysterious. 
Subsequently, we summarize the main results of this paper in less technical terms and describe the structure of the paper. 
 
\subsubsection*{(Mock) Modular Forms}

Modular forms feature prominently in mathematics and branches of theoretical physics. 
See for instance \cite{Zagier2008EllipticMF, DHoker:2022dxx} for an overview. 
These functions on the upper-half plane $\HH$ are distinguished by their {\em modular} symmetry property, which reflects the discrete symmetries of $\HH$. It is interesting to ask how this symmetry can be broken in natural and meaningful ways. 
Mock modular forms embody such a natural generalization of modular forms. 
Here, the modular symmetry only emerges when a non-holomorphic contribution determined by a modular form --  the {\em shadow} -- is added to the mock modular form. 
Since the development of their modern theory two decades ago \cite{zwegers,Bringmann-Ono,2002math12286H,zagier_mock},  through a series of rapid developments it has been established that mock modular forms have a similarly prominent role in combinatorics, moonshine, conformal field theory, string theory and more, extending the applications of modular forms in a fascinating way. 
A partial list of examples can be found in 
\cite{zagier_mock,Ono_unearthing,Folsom_what,Duk_almost,UM,MUM,book,CFS1912,Dabholkar:2021lzt,Alexandrov:2025sig}. 
Specifically, it is known that mock modular forms also appear as characters of certain super vertex algebras \cite{kac2013representations, kac2014representations, kac2016representations,alfes2012mock,CFS1912}.

Theta functions of lattices with indefinite signatures played an important role in the development of the theory of mock modular forms  \cite{zwegers}. 
Building on previous work by Vign\'eras \cite{Vigneras}, Zwegers \cite{zwegers} showed that a regularisation for theta functions of signature $(1,n)$ (one negative direction) leads to a theta function with mock modular properties, while subsequent work showed that regularised theta functions of general signatures lead to higher-depth mock modular forms \cite{lovejoy2013q, males2021higher}. This specific realization of mock modular forms has featured in the study of string theory \cite{BHHD,Korpas:2019cwg,Korpas:2017qdo,Alexandrov:2024wla,Alexandrov:2019rth} and umbral moonshine \cite{cheng2022cone,duncan2017umbral}. Generally, 
  indefinite theta functions give rise to mixed-mock modular forms, which are mock-type forms whose completion involves a finite sum of products between usual modular forms and non-holomorphic contributions from the shadows (cf. \S\ref{subsec:Regularised Indefinite Theta Functions}).

The results of this work further establish the role of mock modular forms as topological invariants of three-manifolds, which we will introduce shortly. 
A subset of preliminary results has been reported in \cite{CFS1912}.

Apart from its modified modular symmetry property (cf. \eqref{dfn:Eichler_integral-Nhol}), another earmark of mock theta functions,  first pointed out by Ramanujan \cite{Ram00,MR3065809}, is their behaviour near the cusps $\hat\QQ:=\QQ\cup \{i\infty\}$ of the upper-half plane. In  the modern language, this leads to the related statement that mock modular forms give rise to {\em quantum modular forms} (QMFs)
\cite{Zag10}. Quantum modular forms, in essence, are functions whose differences with their images under the modular group ${\rm SL}_2(\ZZ)$ enjoy better analytic behaviour when considered near the rationals, compared to the original function. 
Below we will discuss the relevance of quantum modular forms in the context of three-manifold invariants.

In the discussion of the quantum modularity of three-manifold $\hat Z$-invariants,  
a special role will be played by certain {Weil representations} of the metaplectic group  $\widetilde {\rm SL}_2(\ZZ)$ \cite{3d,Cheng:2023row}.  
In particular, we will encounter the  Weil representations $\Theta^{m+K}$ (\ref{dfn:projB},\,\ref{irred_weil}),  which are subrepresentations of the $2m$-dimensional representation $\Theta_m$ spanned by the column vector $\theta_m=(\theta_{m,r})_{r\, \mathrm{mod}\, 2m}$  with a positive integer $m$ and theta function components 
\begin{equation}\label{eq:thetafunc}
    \theta_{m,r}(\tau,z) := \sum_{\ell \equiv r\, \mathrm{mod}\, 2m}
    q^{\frac{\ell^2}{4m}} e^{2\pi i z\ell} ~, \qquad q=e^{2\pi i \tau} ~,
\end{equation}
 labelled by a subgroup $K$ of the group of exact divisors ${\rm Ex}_m$ satisfying $m\not\in K$. The precise definition can be found in \S\ref{sec:Modularity}. 

\subsubsection*{\texorpdfstring{$\hat{Z}$}{Z}-invariants for Three-Manifolds}
Arguably, one of the most prominent open problem in topology is the smooth four-dimensional Poincar\'e conjecture, which states that there are no exotic spherical smooth structures in four dimensions. 
Motivated by this, decades ago 
 Crane and Frenkel envisioned that a \emph{categorification} of numerical three-manifold invariants \cite{crane_fourdimensional_1994}, replacing them with more sophisticated structures of (spectral sequences of) vector spaces, could potentially hold the key to defining structures distinguishing exotic smooth structures in four dimensions. 
A quantum topological invariant for three-manifolds $M_3$ \cite{GPV1602, GPPV} and the closely related quantum knot invariant $F_K$ \cite{GM} were proposed in physical terms recently, leading to an interesting new approach to such a categorification program. Specifically, the origin of these invariants in physical M-theory lends weight to a possible connection to four-dimensional topology. More precisely, it is proposed to consider the 3d half-index, namely the supersymmetric partition function on the cigar background times the temporal circle, of the three-dimensional quantum field theory obtained by compactifying the M5 brane theory, or the 6d ${\cal N}=(2,0)$ ADE superconformal field theory to be more precise,  on $M_3$ as the topological invariant of $M_3$. Note that the three-dimensional spacetime can also be thought of as a solid torus, with the complex structure  $\tau$ of the boundary torus identified with the argument of the $\hat Z$-invariants. 

However, the physical proposal does not translate into a computational algorithm to compute it for general three-manifolds, due to our insufficient 
detailed knowledge of M-theory. 
To move forward, hints can be obtained from their relation to the Witten-Reshetikhin-Turaev invariants they seek to categorify.  
Following these hints, conjectural expressions for \texorpdfstring{$\hat{Z}$}{Z}-invariants have been proposed in \cite{GPPV} for a particularly simple infinite family of three-manifolds: the {\em weakly negative plumbed three-manifolds}. 
To explain what they are, first recall 
that plumbed three-manifolds are three-manifolds that can be constructed by taking the boundary of a four-manifold obtained by gluing together disk bundles over $S^2$. 
The data can be encoded in terms of a weighted graph $(V,E,a)$ called a plumbing graph, obtained by identifying the set $V$ of vertices with the set of disk bundles, equipped with a weight function $a: V\to \ZZ$ whose values are the Euler numbers $a_v$ of the disk bundles, and connecting the two vertices $v$ and $v'$ with an edge, $(v,v')\in E$, if the corresponding disk bundles are glued. 
Alternatively, the same data can be captured using  
the ``plumbing matrix'' $M$ with the weights $a_v$ on the diagonal; for the off-diagonals entries corresponding to a pair of nodes $(v,v')$, we set the value to be 1 if $(v,v')\in E$ and zero otherwise.
A vertex is said to be a high-degree vertex if it is connected to at least three other vertices,  ${\rm deg}(v)\geq 3$. 
Finally, the plumbed three-manifold is said to be weakly negative if the inverse plumbing matrix $M^{-1}$ is negative-definite in the subspace generated by the high-degree vertices.  
In this paper, we mainly focus on Seifert manifolds with three exceptional fibres, which correspond to plumbing graphs with one degree three vertex connected with three rays. We will therefore refer to them as the {\em negative} and {\em positive} $M_3$, given by the signature of $M^{-1}$ along the one-dimensional subspace spanned by the central vertex.

Consider such a weakly negative plumbed three-manifold  $M_3$ with plumbing matrix $\sM$. 
Given a choice of $b\in (\delta+{\rm Coker}(\sM))$ (cf. \eqref{set:spinc}), corresponding to a choice of {Spin$^c$}-structure on $M_3$, 
it was conjectured that the corresponding 3d  half-index is given by the contour integral \cite{GPPV}
\begin{equation}\label{weak_neg_def}
\hat Z_b(M_{3};\tau) :=  q^\Delta \oint \prod_{v\in V}\frac{dz_v}{2\pi i z_v}~ \left(z_v-\frac{1}{z_v}\right)^{2-{\rm deg}(v)}~\Theta_b^{\sM}(\tau;{\bf z})
\end{equation}
for some  $\Delta\in \QQ$, 
where the theta function is given by
\be\label{def:theta}
\Theta^\sM_b (\tau;{\bf z}):=\sum_{\ell \in  2\sM\ZZ^{|V|}\pm b}q^{-\ell^T \sM^{-1} \ell/4} \, {\bf z}^\ell.
\ee
When necessary, the integral is defined as the principal value integration. 
It was shown that the WRT invariant, which is the Chern-Simons partition function suitably normalized, can indeed be recovered from  $\hat Z_b(M_{3};\tau)$ defined above,  by combining $\hat Z_b(M_{3};\tau)$ for different $b$ and by taking the limit $\tau\to {1/k}$, often referred to as the radial limit, from within the upper-half plane \cite{GPPV}. 
Schematically, we have 
\begin{equation}\label{Zhat_WRT}
\hat Z_b(M_{3};\tau)
\xlongrightarrow[\text{combining }b]{\tau\to \frac{1}{ k} + i 0^+}  
{\rm WRT}(M_3;k)
\end{equation}
where $k\in \ZZ$ is the (renormalized) Chern-Simons level. 

In this paper, we focus on the $\hat Z$-invariants of gauge group $G=\rm{SU}(2)$ for simplicity, while we expect that the results can be generalized analogously to other ADE gauge group. 
  See also \cite{CCFFGHP2201, Chung:2022ypb,Park1909} for prior considerations. Moreover, we focus on Seifert manifolds with three exceptional fibres, and links associated with the end nodes of the plumbing graph.

The above consideration can be extended from closed three-manifolds to knot complements. 
Consider a graph with a distinguished degree one node. It corresponds to a three-manifold $Y$ which can be identified as the complement of a knot $K$ associated with the distinguished node, in a closed three-manifold $\hat Y$ which is the plumbed manifold with the plumbing graph given by the graph that we started with, but with the distinguished node replaced by a regular node. From this construction one can write down the two-variable topological invariant $F_K(x,\tau)$ associated to the knot complement \cite{GM}.  
Moreover, another closed three-manifold $Y_{p/r}$ can be obtained via a $p/r$-Dehn surgery along the knot. The $\hat Z$-invariant of  $Y_{p/r}$ can be obtained via a Laplace-like transformation of the invariant $F_K(x,\tau)$ (\ref{eq:Laplace transform2}). 
 For some knots $K$, it was conjectured that there $F_K$ can be expressed in terms of so-called inverted Habiro series of $K$ \cite{Park2106}. This together with a set of conjectural formulas for Dehn surgeries provides further conjectural methods to compute $\hat Z$-invariants, sometimes applicable also to closed three-manifolds that are not weakly negative plumbed three-manifolds. 
All the above (conjectural) ways of obtaining the $\hat Z$-invariants will be used in our article later.

\subsubsection*{Defect Operators} 
As mentioned above, $\hat Z$-invariants give information about a three-manifold $M_3$ by computing the 3d half-index, or the supersymmetric partition function on the background of a cigar times the temporal circle, 
of the 3d ${\cal N}=2$ theory arising from compactifying $M5$ branes on $M_3$. 
It is informative to incorporate half-BPS line operators $W$ in the 3d ${\cal N}=2$ theory, arising from M2 branes located at the centre of the cigar and which wrap around a link $K\subset M_3$ (\cite{GPPV}, \S4). One can then compute the half-index $\hat Z_b(M_3;W;\tau)$ in the presence of these half-BPS line operators $W$. As before,  the parameter $\tau\in {\mathbb H}$ can be identified with the complex structure of the boundary torus of the 3-dimensional space $D^2\times {\mathbb R}_t$ on which the 3d  ${\cal N}=2$ theory lives.

We conjecture an explicit expression for $\hat Z_b(M_3;W_{\boldsymbol \nu};\tau)$, where ${\boldsymbol \nu} : V \to  \ZZ_{\geq 0}$ specifies the representations of the associated defects, as well as their formulation in terms of knot surgeries. To be precise, and $\nu_v$ times the fundamental weight is the highest weight of the highest weight $sl_2$-representation 
associated to the knot corresponding to the  node $v$ of the plumbing graph. We will find that these defect half-indices are indispensable in the full understanding of the quantum modular properties, to be discussed below, of the $\hat{Z}$-invariants. This is consistent with our expectation that including supersymmetric defects provides important physical insights into the system. 

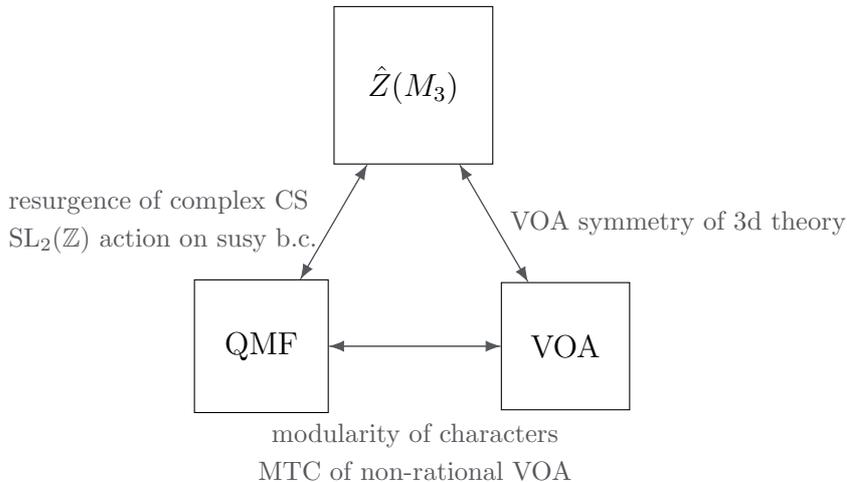
\begin{figure}[h]\begin{center}
\begin{tikzpicture}[square/.style={draw, regular polygon, regular polygon sides=4, minimum size=1.5cm}]
\node[square] (A) at (0,0) {QMF};
\node[square] (B) at (4,0) {VOA};
\node[square] (C) at (2,3.46) {$\hat Z(M_3)$};

\draw[-{Latex[length=2mm]},grijs1] (A) -- (B) node[midway, below=25pt, align=center] {\footnotesize{modularity of characters} \\  \footnotesize{MTC of non-rational VOA}};
\draw[-{Latex[length=2mm]},grijs1] (B) -- (A);

\draw[-{Latex[length=2mm]},grijs1] (B) -- (C) node[midway, right=2pt] {\footnotesize{VOA symmetry of 3d theory}};
\draw[-{Latex[length=2mm]},grijs1] (C) -- (B);

\draw[-{Latex[length=2mm]},grijs1] (C) -- (A) node[midway, left=5pt, align=left] {\footnotesize{resurgence of complex CS}  \\  \footnotesize{${\rm{SL}_2(\ZZ)}$ action on susy b.c.}};
\draw[-{Latex[length=2mm]},grijs1] (A) -- (C);
\end{tikzpicture}
\caption{\small{The topics of this paper and their relations.}}
\label{fig:intro}
\end{center}
\end{figure}

\subsubsection*{\texorpdfstring{$\hat{Z}$}{Z}-invariants, Quantum Modular Forms, and Vertex Operator Algebras} 

The study of the properties of the \texorpdfstring{$\hat{Z}$}{Z}-invariants was initiated in \cite{3d}. In particular, the authors of \cite{3d} proposed the following
\begin{enumerate}
    \item {\em \texorpdfstring{$\hat{Z}$}{Z}-invariants are closely related to quantum modular forms of some  types;}
    \item {\em\texorpdfstring{$\hat{Z}$}{Z}-invariants are related to characters of vertex operator algebras.   } 
\end{enumerate}
The relation to quantum modular objects built on an array of interesting earlier results and observations in a similar context
\cite{GPV1602, GPPV,LawZag,GMP1605,Hik0506,Hik0405,Hik0305}. 

Since then, it has been an active area of study to extend and verify these broad conjectures
\cite{CFS1912, Cheng:2023row, CCFFGHP2201, BMM1810, Chae:2022iuu, Chung1811, Kuch1906, Chung1906, Chung1912, Park1909, Costantino:2021yfd, GPP2009, Mori:2021ost, Matsusaka:2021vbw, Murakami:2022cjk, Murakami:2022bhn, Murakami:2023oam,cheng20223, costin2023going, Adams:2025aad, Harichurn:2025suf}.  
Both of the conjectural properties point to a large system of unexpected symmetries in the underlying physical system and the corresponding topological problem. 
Consequently,  understanding the relation to VOA and quantum modular forms will constitute interesting progress in the understanding of low-dimensional topology and the symmetries of M-theory compactifications. 

In the case of VOAs, it is believed that the conjectured property {\em 2.} is a consequence of a VOA symmetry acting on the BPS states of the relevant 3d ${\cal N}=2$  theory. 
See also \cite{gaiotto2018vertex,gaiotto2020miura} for related developments. See Figure \ref{fig:intro}. 

In the case of quantum modular forms, the study can be viewed as a part of the broader effort to understand the mysterious quantum modular behaviour of related invariants in low-dimensional topology \cite{LawZag,zagier2001vassiliev,qmf,Hik0506,Hik0405,Hik0305,Garoufalidis:2021lcp,Garoufalidis:2023fkt,An:2023ded}. Physically, this type of modified symmetry under the action of the modular group ${\rm SL}_2(\ZZ)$ is expected to stem from the  $\widetilde{\rm SL}_2(\ZZ)$ action on the torus boundary of the 3d spacetime relevant for the half-index, which can be viewed as a solid torus in Eucliean signature. Interestingly, we have seen hints of concrete ${\rm SL}_2(\ZZ)$ representations underlying the different $\hat Z$-invariants of the given manifolds in all the cases analyzed in  \cite{3d}. At the same time, these invariants appear to only span a subspace of the representation. More concretely, in \cite{3d} it was observed that the $\hat Z$-invariants often correspond to some but not all of the components of a $\widetilde{\rm SL}_2(\ZZ)$ vector-valued quantum modular form. 
A natural question, one that we will address in this work, is then
\begin{description}
\centering
\item[\emph{Question 1:}] \emph{What is the role of the other components of the $\widetilde{\rm SL}_2(\ZZ)$ vector-valued quantum modular forms? What does the modular group act on? } 
\end{description}
We will answer this question by providing an interpretation for the full ${\rm SL}_2(\ZZ)$ representations in this work, hence elucidating the action of the modular group in the system. The fact that $\hat Z$-invariants with different defects form a vector-valued quantum modular form suggests an interesting modular group action on the set of possible supersymmetry-preserving boundary conditions, or equivalently supersymmetric M2 brane configurations. 

Note also that the two conjectures are connected in an interesting and non-trivial way. The relevant question bridging the two is:  what is the modular-like property of the characters of a general vertex operator algebra? Such properties, if they exist, could be viewed as generalizations of the celebrated Zhu's modularity theorem beyond the realm of rational, $C_2$-cofinite VOAs  \cite{zhu1996modular,dong2000modular}. See \cite{Creutzig_2017,c24,c28,c92,c93,c30} for some of the related works.

Many questions still remain open for general weakly negative plumbed manifolds and for general ADE gauge groups, including the construction of the vertex operator algebras and the proof of the  (higher-depth) quantum modularity properties. 
At the same time, another fascinating direction is to explore  plumbed manifolds that are not weakly negative, as the mathematical definition
\eqref{weak_neg_def} given in terms of the plumbing data (and the analogue expression for general ADE gauge groups \cite{Park1909}) is not even available in these cases.
As pointed out in \cite{CFS1912} and will be refined in \S\ref{subsec:indef_zhat}, the origin of the difficulty is obvious: the contour integral leads to an infinite sum along the directions of high-degree vertices, and we hence need the bilinear form given by $-M^{-1}$ to be positive definite when restricted to those directions, in order to lead to a $q$-series with powers of $q$ bounded from below. In \cite{3d}, the authors refer to the question of how to execute the $\tau\mapsto -\tau$ transformation as the question of ``going to the other side". 
Here, the crucial question is 
\begin{description}
\centering
\item[\emph{Question 2:}] \emph{How can we take $\tau\mapsto -\tau$ to turn $\hat Z(M_3)$ into $\hat Z(-M_3)$?}
\end{description}
Despite the difficulty,  in this work we will take a step towards a {new approach}, exploiting insights from quantum modular forms\footnote{While this work is being completed, the paper \cite{costin2023going} appeared which investigates the problem from an interesting resurgence point of view, complementary to the modular point of view taken in this work.}.

After identifying the invariants $\hat Z(-M_3)$,  the second part of the 3d modularity proposal leads to the following question. 
\begin{description}
    \centering
 \item[\emph{Question 3:}] \emph{Which VOAs correspond to the invariants $\hat Z(-M_3)$ on the other side?}
\end{description}
In this paper we will address the above three questions. 

\subsubsection*{The False and the Mock }
To understand the phenomenon of ``going to the other side",
first 
recall the chirality property of the Chern-Simons theory, which leads to the relation 
\begin{equation}\label{eqn:flipWRT}
{\rm WRT}(M_3;k) = {\rm WRT}(-M_3;-k)
\end{equation}
between WRT invariants of orientation-reversed three-manifolds. 
Since the orientation reversal operation does not preserve the weakly negative property of the plumbing graph, the above chirality property, the relation \eqref{Zhat_WRT} between  $\hat{Z}$-invariants and WRT invariants, and the expression \eqref{weak_neg_def} for the weakly negative $\hat{Z}$-invariants, provide important hints about $\hat{Z}$-invariants for plumbed manifolds that are not weakly negative. Naively, one expects $$\hat Z(-M_3;\tau)= \hat Z(M_3;-\tau).$$
This hint was exploited in  \cite{3d} and resulted in the so-called {False-Mock Conjecture}. 

Focusing for now on the negative Seifert manifolds with three exceptional fibres. 
It has been established \cite{3d,BMM1810}  that their $\hat Z$-invariants are, up to an overall $q$-power and the addition of a finite polynomial, false theta functions. By false theta function, we mean
 any function of the form $\sum_{{ k\in {\mathbb N}}} a_k q^{k^2/4m}$, where $a_k = a_{k+2m} = -a_{-k}$ for all $k\in \ZZ$, for some positive integer $m$.
In other words, they are
 linear combination of 
\begin{equation}\label{eq:intro:falsetheta}
  \tilde\theta_{m,r} (\tau) :=   \sum_{\substack{k\in \ZZ\\k= r~(\mathrm{mod}\, 2m)}} 
    {\rm sgn}\left(k
    \right)q^{\frac{k^2}{4m}}~,\ r\in \ZZ/2m, 
\end{equation}
for a fixed $m\in \NN$, where the nomenclature stems from the additional sign factor in the one-dimensional lattice sum.  

False theta functions can be viewed as Eichler integrals, and this fact makes their quantum modular properties obvious \cite{Zag10}. 
More precisely, $  \tilde\theta_{m,r}$ is, up to a numerical multiplicative factor, the holomorphic Eichler integral of the unary theta function $\theta^1_{m,r} (\tau) :=   \sum_{k= r~(\mathrm{mod}\, 2m)} k\,q^{\frac{k^2}{4m}}$. 
The holomorphic Eichler integral of a weight $w\in \frac{1}{2}{\mathbb Z}$ cusp form $g(\tau)=\sum_{n>0}a_g(n) q^n$ is given by
\begin{equation}\label{dfn:Eichler_integral}
    \tilde g(\tau) := \sum_{n>0}a_g(n)\,n^{1-w} q^n ~, 
\end{equation}
or equivalently 
\begin{equation}\label{dfn:Eichler_integral-hol}
    \tilde g(\tau) ={(2\pi i)^{w-1}\over\Gamma(w-1)} \int_{\tau}^{i\infty} g(\tau')(\tau'-\tau)^{w-2} d\tau'~ 
\end{equation}
with a carefully chosen integration path. 
One similarly defines the non-holomorphic Eichler integral as  
\begin{equation}\label{dfn:Eichler_integral-Nhol}
     {g}^\ast (\tau) = {(2\pi i)^{w-1}\over\Gamma(w-1)}\int_{-\bar\tau}^{i\infty} \overline{g(-\overline{\tau'})}(\tau'+\tau)^{w-2} d\tau'. 
\end{equation}
In terms of a Fourier expansion, we have\footnote{There is in general an additional term $a_g(0) \frac{(4\pi \tau_2)^{w-1}}{1-w}$ if $g$ is not a cusp form.  } 
\begin{equation}
\label{dfn:Eichler_integral-Nhol2}
 {\ex}(\tfrac{1-w}{2}) \Gamma(w-1) {g}^\ast (\tau) = \sum_{n>0} n^{1-w} \bar a_g(n)\Gamma(w-1,4\pi n \tau_2) q^{-n}
\end{equation}
where we have written $\tau_2 = \Im \tau$, ${\ex}(x):=e^{2\pi i x}$, and the incomplete $\Gamma$ function is given by $\Gamma(1-k,x)=\int_x^\infty t^{-k}e^{-t} dt. $ Note that the summand of the right hand side vanishes as $e^{-2\pi n \tau_2}$ in the limit $\tau_2\to \infty$.

A weight $w$ (weakly holomorphic) mock modular form $f:\HH\to\CC$ with shadow given by a weight $2-w$ cusp form $g$  is a (weakly) holomorphic function such that the non-holomorphic function $\hat f:= f-  {g}^\ast$ transforms as a modular form of weight $w$.  From now on we focus on the case of weight $w=1/2$ mock modular forms for $\Gamma\subset {\rm SL}_2(\mathbb Z)$ relevant for the context in this paper.  Near a cusp $u\in\QQ$, $f$ might have an exponential singularity. In particular, there exists a finite set $\{f_u\}$ of weakly holomorphic modular forms such that $f-f_u$ is finite as $\tau$ approaches $u$ from within the upper-half plane \cite{MR3065809}. Crucially,  the cardinality of this set is necessarily larger than one for any mock modular form $f$ of weight $1/2$ with a non-vanishing shadow. 
Equipped with $\{f_u\}$, we can subtract the possible singularities of the mock form $f$ by defining $f_{u,{\rm fin}}:= f-f_u$. 
Then their asymptotic expansions near a cusp and those of the 
Eichler integral of its shadow agree to infinite order \cite{3d,CLR}: 
\begin{equation}\label{eqn:leaking}
f_{u,{\rm fin}}(u+iv)  \sim \sum_{n\geq 0} \alpha_u(n) v^n~~ , ~~ {g}^\ast (-u+ iv) \sim \sum_{n\geq 0} \alpha_u(n) (-v)^n ~,
\end{equation}
Note that the role of the $f_u$ in the subtraction is just to remove the singularity. So the above relation also holds when $f-f_u$ is replaced by $f$ with just the exponential singularities terms subtracted, instead of the whole modular form $f_u$ (the difference vanishes exponentially near the cusp $u$).  
Note that in the case $u={1/k}$, the above relation is consistent with the reversal symmetry \eqref{eqn:flipWRT} of the WRT invariants. Combined with \eqref{Zhat_WRT}, and other evidence, this consideration on the radial limits leads to the following {\em false-mock conjecture},\footnote{Note that this version of the false-mock conjecture is slightly more general than the version in \cite{CFS1912,3d}: here we allow for mixed mock modular forms, a generalization of mock modular forms. At present we are not certain whether such a generalization is actually necessary.} which constitutes an important inspiration of the present paper. 

\begin{con}\label{con:False-Mockv2}
Let $M_3$ be a three-manifold for which the $\hat{Z}$-invariants take the form
    \begin{equation}
        \hat{Z}_b(M_3;\tau) = q^c \left(\tilde{\vartheta}(\tau)+p(\tau)\right),
    \end{equation}
where $\tau\in\mathbb{H}$, $c\in \mathbb{Q}$, $\tilde{\vartheta}(\tau)$ is the Eichler integral of a unary theta function $\vartheta(\tau)$ of weight $3/2$ and $p(\tau)$ is a polynomial in $q=e^{2\pi i \tau}$. Then the $\hat{Z}$-invariant of the manifold with reversed orientation is
    \begin{equation}
        \hat{Z}_b(-M_3;\tau) = q^{-c} \left(f(\tau)+p(-\tau)\right),
    \end{equation}
where $f(\tau)$ is a weight $1/2$ weakly holomorphic (mixed) mock modular form. 
Moreover, 
its completion $\hat f$, transforming as a modular form of weight $1/2$ under a certain congruence subgroup of ${\rm SL}_2(\ZZ)$, has the following properties. It is given by 
\begin{equation}
\hat f = f  - \vartheta^\ast - \sum_{i\in I} g_i \vartheta_i^\ast
\end{equation}
where $I$ is a finite set, $\vartheta_i$ is a theta function, and $g_i$ with any $i\in I$ is a modular function for some discrete subgroup of ${\rm SL}_2(\ZZ)$ that either vanishes or has an exponential singularity at any cusp. 
\end{con}

Crucially, $\tilde \vartheta$ and the finite  part of the shadow $\vartheta^\ast$ are the holomorphic resp. non-holomorphic Eichler integrals of the {\em{same}} cusp form $\vartheta$. 

\subsubsection*{Summary of Results} 
In this paper we aim to address the three questions mentioned above. We will do so in the following steps. 

\begin{itemize}
\item We propose two expressions, Definition \ref{conj:plumbed_defect} in terms of plumbing graphs and Definition \ref{conj:surgery_defect}-\ref{dfn:surgery_defectFK} in terms of knot surgeries, as additional $\hat Z$ topological invariants. We conjecture that they give the defect half-indices, capturing the effects of inserting supersymmetric
 defect operators in the 3d theory ${\cal T}(M_3)$. 

\item With the above definition, we propose that, when an appropriately defined set of defects operators are included, the $\hat Z$-invariants of a given three-manifold are related to a (higher-depth) vector-valued quantum modular form of some type for the modular group ${\rm SL}_2(\ZZ)$. 

Focusing on Seifert manifolds with three singular fibres, we can be more concrete about the above expectation which we phrase in terms of the  
following Modularity Conjecture:  
\begin{center}{\emph{ 
For a given  Seifert manifold $M_3$ with three singular fibres, the vector space spanned by $\hat Z(M_3;W_{\bnu};\tau)$ is, up to an equivalence, isomorphic to a Weil representation $\Theta^{m+K}$, labelled by a positive integer $m$ and a subgroup $K$ of the group of exact divisors of $m$. }}
\end{center}
More explicitly, we conjectured that, up to the addition of finite polynomials and multiplicative pre-factors, the vector space spanned by $\hat Z(M_3;W_{\bnu};\tau)$ is also spanned by the components of a vector-valued quantum modular form for 
${\rm SL}_2(\ZZ)$, transforming under the (dual of the) Weil representation $\Theta^{m+K}$.
See Conjecture \ref{con:modularity} for the precise version of this Modularity Conjecture. 
 Moreover, we prove that this conjecture holds when $M_3$ is any weakly negative Brieskorn sphere $\Sigma(p,q,r)$ (see Theorem \ref{thm:modularityBrieskorn}). 

This answer the {\em{Question 1}} listed above.

\item { Analysing the regularization of the integral and the knot surgery formulations for $\hat Z$-invariants, in \S\ref{subsec:habiroexamples} and in Conjecture \ref{sonj:23torus}  we {study} concrete mock modular forms {proposed} as the defect $\hat Z$-invariants for the inverted Brieskorn spheres $-\Sigma(2,3,7)$ and $-\Sigma(2,3,5)$, and the Seifert manifold $M(-1;{1\over 2},{1\over 3},{1\over 8})$ in Conjecture \ref{sonj:H24}. 
Subsequently, we prove that the Modular Conjecture  \ref{con:modularity} in the previous item holds for the conjectured invariants of the inverted Brieskorn spheres {$\Sigma(2,3,5)$, $\Sigma(2,3,7)$}, and $M(-1;{1\over 2},{1\over 3},{1\over 8})$. See Theorem \ref{thm:mod_23surgery}, Theorem \ref{thm:mod_248}, and Corollary \ref{cor:mock_mod}. This provides an answer to {\em{Question 2}} listed above for these cases.}

\item {Equipped in the afore-mentioned cases with an expression for $\hat Z(M_3)$ in terms of theta functions of indefinite lattices, we propose a way to  attach vertex operator algebras corresponding to the positive Seifert manifolds. 
This provides an answer to the {\em{Question 3}} listed above for these cases.}

\item Using our proposals (Definition \ref{conj:plumbed_defect} and \ref{conj:surgery_defect}) for the defect invariants, in \S\ref{sec:ex} we provide an extended list of concrete examples,  furnishing evidence for the Modularity Conjecture \ref{con:modularity} as well as the validity of the proposals for the defect invariants themselves. 

\item Finally, some of the codes with which some of the calculations have been done can be found on this \href{https://github.com/d-passaro/3D-Modularity-Revisited}{page at github}. 

\end{itemize} 

\section{\texorpdfstring{Defects $\hat{Z}$}{Z}-invariants and Modularity}\label{sec:examples1S}

In this section we will first introduce our proposal for defect $\hat Z$-invariants in terms of the contour integrals and knot surgeries. 
Equipped with this larger family of invariants, we will take a fresh look at the quantum modularity of the $\hat Z$-invariants.
In particular, we will shed light on their structure as vector-valued quantum modular forms for the modular group $\widetilde {\rm SL}_2(\ZZ)$\footnote{or equivalently something that we will call quantum Jacobi theta functions. See the comment at the end of the section.}, and the important role of the defects in $\hat Z$-modularity. The consideration of quantum modularity should hold equally for three-manifolds of both orientations, though as we see the actual methods for computing them will differ on the different sides. 

\subsection{Plumbed Manifolds}\label{subsec:contourintegral}

First we give more information on the contour integral expression (\ref{weak_neg_def}-\ref{def:theta}) for $\hat Z$-invariants for weakly negative plumbed three-manifolds introduced in \S\ref{sec:intro}. 
Let $M_3$ be a plumbed three-manifold; for instance,  any
Seifert manifold with $n$ singular fibres $M_3(a;\{q_i/p_i\}_{i=1}^{n})$ is an example  of plumbed three-manifolds. Its plumbing graph contains a central vertex $v_0$ connected to $n\geq 3$ legs and the weights $\alpha^{(i)}_j$ for the nodes on the $i^{th}$ leg are determined by 
	\begin{equation}
    \label{eq:cont'edfrac}
	\frac{q_i}{p_i}=\frac{-1}{\alpha^{(i)}_1-\frac{1}{\alpha^{(i)}_2-\frac{1}{\alpha^{(i)}_3-\cdots}}} ~.
	\end{equation}

For $M_3$ a weakly negative plumbed three-manifold, using the above notation, the topological invariants $\hat{Z}_b(M_3;\tau )$ are defined via the principal value $|V|$-dimensional integral (\ref{weak_neg_def}) \cite{GPPV},
\begin{equation}\label{def:Zhatintegral}
\hat{Z}_b(M_3;\tau):= (-1)^{\pi}q^{\frac{3\sigma -\sum_{v\in V} \alpha(v)}{4}} \, \mathrm{vp}  \oint \prod_{v\in V}\frac{dz_v}{2\pi i z_v}\biggl(z_v-\frac{1}{z_v}\biggr)^{2-\text{deg}(v)}\Theta^{\sM}_b(\tau, {\sz})~,
\end{equation}
where $\sigma$ is the signature of the plumbing matrix $M$, $\pi$ denotes the number of positive eigenvalues of $M$, and the contours of integration are to be the set $|z_v|=1$. 
The label $b$ can be identified with the elements of the set
\begin{equation}\label{set:spinc}
{\mathrm{Spin}}^c(M_3)\cong (2\mathbb{Z}^{|V|}+\delta)/(2M\mathbb{Z}^{|V|})~,
\end{equation}
where $\delta \in \mathbb{Z}^{|V|}/2\mathbb{Z}^{|V|}$ is defined by $\delta_v={\rm deg}(v)\mod 2$. Labelled by ${{b}}\in{\mathrm{Spin}}^c(M_3)$ (up to the Weyl group action), the theta function in equation \eqref{def:Zhatintegral} is given by \eqref{def:theta}.

Next we propose a generalisation of the $\hat{Z}$-invariants, intended to account for the inclusion of supersymmetric line defects in the  $3$d $\mathcal{N}=2$ quantum field theory on a cigar background.
For the purposes of this paper, we only give the explicit expression for $\rm{SU}(2)$ gauge groups, though the generalization to ADE gauge groups should be straightforward (cf. \cite{CCFFGHP2201}). 
Denote by $\vec{\omega}$ the fundamental weight of the Lie algebra $\mathfrak{sl}_2$. We consider defects in the 6-dimensional parent theory on $M_3\times D^2\times_\tau S^1$ with support on links in $M_3$ associated to a collection of nodes, $v\in V_{\rm W}$, in the plumbing graph, with the corresponding highest weight representation with highest weights $\nu_v \vec\omega$. 

\begin{dfn}\label{conj:plumbed_defect}
Consider a weakly negative plumbed manifold $M_3$, and 
defects associated to a collection of nodes  $V_{\rm W}$ in the plumbing graph, with the highest weight representation with highest weight $\nu_v \vec\omega$. 
We define the defect $\hat{Z}$-invariant 
as given by the contour integral 
\begin{gather}\label{def:ZhatintegralWilson}
\begin{split}
&\hat{Z}_b(M_3;W_{\{\nu_v\}_{v\in V_W}};\tau ):= \\&(-1)^{\pi}q^{\frac{3\sigma -\sum_{v\in V} \alpha(v)}{4}} \, \mathrm{vp}  \oint \prod_{v\in V} \frac{dz_v}{2\pi i z_v}\biggl(z_v-\frac{1}{z_v}\biggr)^{2-\text{deg}(v)}\left(\prod_{v\in V_W}  \chi_{\nu_v}(z_v)  \right) \Theta^{\sM}_{b+{\{\nu_v\}_{v\in V_W}}}(\tau, {\sz})~,
\end{split}
\end{gather}
where $\chi_{\nu_v}(z_v)$ denotes the  $\mathfrak{sl}_2$ character 
\begin{equation}\label{eq:characterdefB}
  \chi_{\nu}(z) = \frac{z^{1+\nu}-z^{-1-\nu}}{z-z^{-1}} = \sum_{k=0}^{\nu} z^{\nu-2k} ~,
\end{equation}
and the modified theta function $ \Theta^{\sM}_{b+{\{\nu_v\}_{v\in V_W}}}$  is given by modifying ${{b}}\in (2\mathbb{Z}^{|V|}+\delta)/(2M\mathbb{Z}^{|V|})$ in \eqref{def:theta} by replacing $b_v\mapsto b_v+\nu_v$ for $v\in V_W$. 
\end{dfn}

\begin{prop}
The defect $\hat{Z}$-invariant proposed in Definition \ref{conj:plumbed_defect} is a topological invariant of $M_3$, as it is invariant under all the 3d Kirby moves on the plumbing graph of $M_3$ preserving the nodes with $\nu_v\neq 0$.
\end{prop}
\begin{proof}
    The proof is analogous to that of Proposition 4.6 of \cite{GM}.
\end{proof}

Moreover, we conjecture that the proposed modification of $\hat{Z}$ also has the following connection to the physical theory. 

\begin{con}\label{conj:Wilson_interpretation}
The defect $\hat{Z}$-invariant is equal to the half-index of the same 3d $\mathcal{N}=2$ supersymmetric quantum field theory with the corresponding defect operators included.
\end{con}
 
\begin{rmk}
\begin{itemize}\item
The above conjecture generalizes the expressions in previous work  \cite{GPPV,CCFFGHP2201} by modifying the theta function $\Theta^{\sM}_{b+{\{\nu_v\}_{v\in V_W}}}$ in the integrand of the contour integral. 
\item {
In \S\ref{sec:regindefTheta} we will modify and generalize the above integral expression for some plumbed three-manifolds which are not weakly negative.} 
\item
In what follows, we will often focus on the case where $M_3$ is a Seifert manifold with three exceptional fibres, and where Wilson lines are attached to links associated to the end nodes of the three legs in the plumbing graph. We will label the corresponding highest weight representations by ${\bnu} =(\nu_1,\nu_2,\nu_3) \in \NN^3$ and denote the corresponding defect $\hat Z$-invariant by $\hat{Z}_b(M_3;W_{\bnu})$. 
\item 
Note that the expression \eqref{def:ZhatintegralWilson} has the artefact of breaking original the Weyl symmetry, in this case the Weyl symmetry $b\mapsto -b$, of the choice of the (generalized) ${\rm Spin}^c$ structure. Instead, the Weyl group acts like $b+\nu\mapsto -(b+\nu)$ now.   
\end{itemize}
\end{rmk}

\subsection{Knot Surgeries}
\label{subsec:knot}
For certain closed manifolds $M_3$, the invariants $\hat{Z}_b(M_3)$ can be computed in a number of ways that are expected to be equivalent. 
After discussing the contour integral definition using the combinatorial data of the plumbed manifolds, we will now review the proposal for $\hat{Z}$-invariants of closed three-manifolds constructed through surgeries of knots.

Consider a three-manifold $Y$, obtained as the complement of a tubular neighbourhood of a knot $K$ in a closed three-manifold $\hat{Y}\cong\, S^3$ and which therefore has a parametrised torus boundary $\partial Y$. We can construct another closed manifold $S^3_{p/r}(K)$ through surgery by gluing along the boundary $\partial Y$ a solid torus $S^1\times D^2$
\begin{equation}\label{def:mfdfromsurgeryonknot}
    S^3_{p/r}(K) = Y \cup_{\partial Y} (S^1\times D^2) ~,
\end{equation}
with $p/r\in\mathbb{Q}$ prescribing the identification of cycles.  
When the manifold $Y$ admits a plumbing description, we refer to the knot $K$ as a {\em{plumbed knot}}. In this case, 
one can associate to it a two-variable series $F_K(x;\tau)$ through a contour integral similar to that in \eqref{def:Zhatintegral} for closed plumbed manifolds, with the only difference being that the plumbing graph now has a distinguished vertex, which we denote by $v_K$, corresponding to the boundary $\partial Y$~\cite{GM}.

The $\hat Z$-invariant of the closed plumbed manifolds $ S^3_{p/r}(K) $  is related to $F_K$ through the surgery formula
\begin{equation}\label{eq:surgeryS2}
\hat{Z}_{b}(S_{ p/r}^{3}(K);\tau )=\epsilon q^{d}\cdot\mathcal{L}_{ p/r}^{(b)}\left[\left(x^{\frac{1}{2r}}-x^{-\frac{1}{2r}}\right)F_{K}(x,\tau)\right] ~,
\end{equation}
where $\epsilon$ can be $\pm 1$ and $d\in\mathbb{Q}$ is a number fixed by the manifold $S_{ p/r}^{3}(K)$, as defined in Theorem 1.2 of \cite{GM}. The map $\mathcal{L}_{ p/r}^{(b)}$ acts as\footnote{A relative factor of $1/2$ is introduced here when compared to \cite{GM}, which can be understood as coming from the fact that we use the weight instead of the root basis to write down the highest weight of the Wilson lines as well as in the plumbing prescription \eqref{def:ZhatintegralWilson}. Related to this, one can think about $x$ as $z^2$ when comparing \eqref{eq:surgeryS2} with \eqref{def:ZhatintegralWilson}.}
\begin{equation}
\label{eq:Laplace transform2}
\mathcal{L}_{ p/r}^{(b)}(x^u) = q^{- u^{2}r/p} \delta_{ru-  {b/2}\Mod{p}}
\end{equation}
where for $x, ~p\in \ZZ$ we define $\delta_{x\Mod{p}}=1$  when $x\equiv 0\Mod{p}$ and 0 otherwise.

Next we will generalize the above to include Wilson line operators. 
As before, we can consider defects supported on links associated to the nodes of the plumbing graphs. 
In what follows, we focus on Wilson lines along the knot $K$ corresponding to the distinguished node.

\begin{dfn}\label{conj:surgery_defect}
Using the same notation as above, consider the $F_K(x,\tau)$ series associated to the plumbed knot $K$ with a defect operator along $K$ in the highest weight representation of $\mathfrak{sl}_2$ with highest weight $\nu \vec\omega$. 
We define the corresponding defect invariant for the closed manifold  $ S^3_{p/r}(K) $ as 
\begin{equation}\label{eq:surgeryS2_defect}
\hat{Z}_{b}(S_{ p/r}^{3}(K);W_{\nu};\tau)=\epsilon q^{d}\cdot\mathcal{L}_{p/r}^{(b+\nu)}\left[\left(x^{\frac{1}{2r}}-x^{-\frac{1}{2r}}\right)F_{K}(x,\tau)  \chi_{\nu}({x}^\frac{1}{ 2r}) \right] ~,
\end{equation}
where the $\mathfrak{sl}_2$ character $\chi_{\nu}$ is as in \eqref{eq:characterdefB}, and $\epsilon$, $d$ are as in \eqref{eq:surgeryS2} where there are no Wilson lines. 
\end{dfn}

{
Similarly, the effect of adding Wilson lines associated to other nodes can be computed in a way analogous to Definition \ref{conj:plumbed_defect}, leading to a defect $F_K$-series, which we denote by $F_{K,\bnu}$. 
 Explicitly, we are led to the following definition. }

\begin{dfn}\label{dfn:surgery_defectFK}
Using the same notation as above, consider a plumbed knot complement $Y=S^3\backslash K$, defined by the weakly negative definite weighted graph $(V,E,a)$ with the distinguished vertex $v_0\in V$. 
 
Given an integer $n$, we define the defect invariant of $Y$, with defects associated to a collection of nodes $V_{\rm W}$ in the plumbing graph, with the highest weight representation with highest weights $\nu_v \vec\omega$ which are captured by the vector $\bnu:=\{\nu_v\}_{v\in V_{\rm W}}$, to be  
\begin{gather}\label{def:FKintegralWilson}
\begin{split}
    &\hat{Z}_b (Y;W_{\bnu};z,n,\tau):= (-1)^{\pi}q^{\frac{3\sigma -\sum_{v\in V} \alpha(v)}{4}} 
    \\&
    \left(z-\frac{1}{z} \right)^{1-\text{deg}(v_0)} \, \mathrm{vp}  \oint \prod_{\substack{v\in V \\ v\neq v_0}} \frac{dz_v}{2\pi i z_v}\biggl(z_v-\frac{1}{z_v}\biggr)^{2-\text{deg}(v)}\left(\prod_{v\in V_W}  \chi_{\nu_v}(z_v)  \right) \Theta^{(\sM,n)}_{b+{\{\nu_v\}_{v\in V_W}}}(\tau, {\sz})~,
\end{split}
\end{gather}
where $z$ plays the role of $z_{v_0}$. 

In the above, the lattice theta function $\Theta^{(\sM,n)}_a(\tau, {\sz})$ is defined as in \eqref{def:theta},  but now with the sum over the lattice vector restricted to those of the form 
$2\sM\vec{n} \pm a \in  2\sM\ZZ^{|V|}\pm a$ with the entry $n_{v_0}=n$ fixed by the input $n$. 

Using the above, analogous to \cite{GM} we define
\begin{equation}
F_{K,\bnu}(x,\tau):=\hat{Z}_0 (Y;W_{\bnu};\sqrt{x},0,\tau) ~.
\end{equation}
and the corresponding $\hat Z$-invariants, with defects along $K$ (labelled by $\nu$) and attached to other nodes in the plumbed knot complement (labelled by $\bnu$),  to be\footnote{{Note that here and in \eqref{eq:surgeryS2_defect}, we have adopted a notation that uses $\nu$, separate from the tuple $\bnu$, to signify the representation of the defect attached to the knot. One could equally have included it in the tuple, as the knot here corresponds to a distinguished node in the pluming graph. We chose the notation that treats them separately to highlight the distinguished role of the knot in the surgery construction. } }
\begin{equation}\label{eq:surgery3_defect}
\hat{Z}_{b}(S_{ p/r}^{3}(K);W_{\nu,{\bnu}};\tau)=\epsilon q^{d}\cdot\mathcal{L}_{p/r}^{(b+\nu)}\left[\left(x^{\frac{1}{2r}}-x^{-\frac{1}{2r}}\right) {\chi_{\nu}({x}^\frac{1}{ 2r})} F_{K,\bnu} (x,\tau)  \right] ~. 
\end{equation}

\end{dfn}
}

Note that performing the Laplace transform is equivalent to integrating over the distinguished node (see the proof of Theorem 1.2 of \cite{GM}). 
From this point of view, if Conjecture \ref{conj:Wilson_interpretation} on the defect half-index interpretation for the topological invariants $\hat Z(M_3;W_{\bnu})$  holds true, the above invariant $\hat{Z}_{b}(S_{ p/r}^{3}(K);W_{\nu})$ is also expected to have a similar physical definition. As a result we do not list it as a separate conjecture.

\begin{rmk}
    Note that in \eqref{def:ZhatintegralWilson} we have shifted the argument $b$  in the theta function, and in \eqref{eq:surgeryS2_defect}  the argument in the Laplace transformation. 
In general, the shift is non-trivial and not just a relabeling since $b+\nu$ might not belong to the set \eqref{set:spinc} of permissible ${\rm Spin}^c$-structures, when $\nu $ is odd. The non-trivial spin of the defect operator is responsible for this modification of  ${\rm Spin}^c$-structures\footnote{We thank Mrunmay Jagadale for suggesting the interpretation.}.
\end{rmk}

Various other methods have been developed for the computation of $F_K$-series \cite{GM, Park2106, GGKPS2005, EGGKPSS2110}, with each having its advantages.  
To appreciate the necessity of having other methods,  note that the Laplace transforms \eqref{eq:Laplace transform2} generate $q$-series with positive powers of the expansion parameter $q$ only when $p/r<0$. 
In what follows we will focus on a method based on the so-called ``inverted" Habiro series of the knot~$K$, which unlike \eqref{eq:Laplace transform2} can in some situations be applied to the $p/r>0$ surgeries. This will be especially important when we compute $\hat Z$-invariants of the mock type in the next section. In the remainder of the subsection, we will  focus on integral surgeries where $r=1$.  
\subsubsection*{Inverted Habiro Series and Defects}
\label{sec:InvertedHabiro}

First recall that for any knot $K$, there exists a sequence of Laurent polynomials $a_m(K)\in\mathbb{Z}[q,q^{-1}]$, which we will refer to as the Habiro coefficients of $K$, such that the coloured Jones polynomial $J_K(V_n)$ for the $n$-dimensional irreducible representation $V_n$ of $\mathfrak{sl}_2$ can be decomposed as \cite{Habiro}
\begin{equation}
    J_K(V_n)=\sum_{m=0}^\infty a_m(K;q) \left(\prod_{j=1}^m \left(x+x^{-1} -q^j - q^{-j} \right)\right)\bigg|_{x=q^n} ~.
\end{equation}
Following this, in \cite{Park2106} it was conjectured that for any knot $K$ with Alexander polynomial $\Delta_{K}\neq1$, there exist inverted Habiro series with coefficients $\amk$, such that
\begin{equation}
F_{K}(x,\tau)=-\left(x^{\frac{1}{2}}-x^{-\frac{1}{2}}\right)\sum_{m=1}^{\infty}\frac{\amk}{\prod_{j=0}^{m-1}(x+x^{-1}-q^{j}-q^{-j})} ~.
\label{eq:inverted habiro}
\end{equation}
Furthermore, by writing $F_{K}(x;\tau)=x^{\ha}\sum_{j=0}^{\infty}f_{j}(K;q)x^{j}$, the inverted Habiro coefficients can be extracted from the two-variable series $F_K(x,q)$ using the identities \cite{Park2106}
\begin{gather}
\begin{split}f_{j}(K;q)&=  \sum_{i=0}^{j}\left[\begin{array}{c}
j+i\\
2i
\end{array}\right]a_{-i-1}(K;q) \\ \Leftrightarrow
a_{-i-1}(K;q)&=  \sum_{j=0}^{i}(-1)^{i+j}\left[\begin{array}{c}
2i\\
i-j
\end{array}\right]\frac{[2j+1]}{[i+j+1]}f_{j}(K;q) ~,
\end{split}
\end{gather}
where we employ the notation
\begin{equation}
[n]=\frac{q^{\frac{n}{2}}-q^{-\frac{n}{2}}}{q^{\frac{1}{2}}-q^{-\frac{1}{2}}} ~ , \qquad[n]!=\prod_{k=1}^{n}[k] ~, \qquad\left[\begin{array}{c}
n\\
k
\end{array}\right]=\frac{[n]!}{[k]![n-k]!}~.
\end{equation}
We furthermore introduce the notation 
\begin{equation}\label{def:Qj}
Q_j=q^j+q^{-j}, ~~~D_m = \prod_{j=1}^m (x+x^{-1}-Q_j)~~~~ {\rm for}~j,m\in \NN ~.
\end{equation}

The form \eqref{eq:inverted habiro} is particularly useful in
the context of $\pm p$ surgeries. In the remaining part of the subsection, we will let $p\in \NN$. 
First, when combined with the surgery formula \eqref{eq:surgeryS2}, this allows to write the corresponding $\hat{Z}$-invariant as
\begin{equation}
\begin{split}
    \hat{Z}_{b} (S_{- p}^{3}(K);\tau ) &= \epsilon q^d \mathcal{L}_{- p}^{(b)}\left[\left(x^{\frac{1}{2}}-x^{-\frac{1}{2}}\right)^2\sum_{m=1}^{\infty}\frac{\amk}{\prod_{j=0}^{m-1}(x+x^{-1}-q^{j}-q^{-j})}\right] \\
     &= \epsilon q^d \mathcal{L}_{- p}^{(b)}\left[a_{-1}(K;q)+\sum_{n=1}^{\infty}\frac{a_{-n-1}(K;q)}{D_n}\right]~.
\end{split}
\label{eq:+/-p surgery}
\end{equation}
The Habiro coefficients $a_{-n-1}(K;q)$ are independent of $x$, so the Laplace transform \eqref{eq:Laplace transform2} only acts on the product $1/D_n$. 
Consider the $(-1)$-surgery for instance:
the relevant expression is 
\begin{equation}\label{min1_surg}
{\mathcal L}^{(0)}_{-1}\left(\frac{1}{D_n}\right)  = \frac{1}{\prod_{j=1}^{n}(x+x^{-1}-q^{j}-q^{-j})}\bigg\lvert_{x^u\mapsto q^{u^2}} = \frac{q^{n^{2}}}{\left(q^{n+1};q\right)_{n}} 
\end{equation}
where
$
(x;q)_{n}=\prod_{i=0}^{n-1}(1-xq^{i})~
$
denotes the $q$-Pochhammer symbol. 
Now, the corresponding expression for $(-p)$-surgery can be written as
\begin{equation}
\mathcal{L}_{-p}^{(b)}\left(\frac{1}{D_n}\right) =q^{-\frac{b(2p-b)}{4p}} \frac{q^{n^{2}}}{\left(q^{n+1};q\right)_{n}}\,P_{n}^{p,b}(q^{-1}) ~
, 
\label{eq:expansions-for-surgeries-1}
\end{equation}
or equivalently
\begin{equation}
P_{n}^{p,b}(q^{-1})  := q^{\frac{b(2p-b)}{4p}}  \frac{\mathcal{L}_{- p}^{(b)}\left(\frac{1}{D_n}\right) }{{\mathcal L}^{(0)}_{-1}\left(\frac{1}{D_n}\right) } 
\,,~~ b~{\rm{even}}. 
\label{eq:expansions-for-surgeries-12}
\end{equation}

  It turns out that  $P_{n}^{p,b}$ defined above is a finite polynomial for any $n\in \ZZ_+$.  When $b$ is odd, the left-hand side of  \eqref{eq:expansions-for-surgeries-1}
vanishes identically. We record the closed form formulae and explicitly the first few polynomials $P_{n}^{p,b}$ in Appendix \ref{app:Pnpb}.  The definition for the case of odd $b$ will be given in \eqref{eq:expansions-for-surgeries-1 odd}.

Using \eqref{eq:expansions-for-surgeries-1}, expression \eqref{eq:+/-p surgery}
for $(-p)$-surgery can be written as 
\begin{equation}
\hat{Z}_{b} (S_{-p}^{3}(K);\tau )=  
\epsilon q^{d-\frac{b(2p-b)}{4p}} \sum_{n=0}^{\infty}a_{-n-1}(K;q)\frac{q^{n^{2}}}{\left(q^{n+1};q\right)_{n}} P_{n}^{p,b}(q^{-1})~.
\label{eq:surgery formulas 1}
\end{equation}
For positive surgeries, taking $F_K(x, \tau)$ as a series expansion in $x$ and $q$ and directly applying the Laplace transform \eqref{eq:Laplace transform2} does not generate $\hat{Z}$-invariants with convergent $q$-series in general, as there may be unbounded negative powers of $q$. 
The proposal of \cite{Park2106} is to obtain the result of a  $(+p)$-surgery formula by taking $q$ to $q^{-1}$ in a very specific way that we will now describe. 
First, assume that the only factor in the summand of \eqref{eq:surgery formulas 1} that leads to an infinite $q$-series when expanded, namely the factor $\frac{q^{n^{2}}}{\left(q^{n+1};q\right)_{n}}$,  
can be extended outside the unit circle via the $q$-hypergeometric expression 
\begin{equation}\label{q_hyp}
\frac{1}{(q^{-a};q^{-1})_n} = (-1)^n \, \frac{q^{\, an+\frac{n(n-1)}{2}}}{(q^{a};q)_n}, 
\end{equation} 
where we have $a=n+1$ for the case of \eqref{eq:surgery formulas 1}. The rest of the factors in \eqref{eq:surgery formulas 1} are given by a finite polynomial in $q$ or $q^{-1}$ and can be transformed under $q\mapsto q^{-1}$ in a straightforward way. In particular, we are led to extend the definition of  
$\mathcal{L}_{k}^{(b)}\left(\frac{1}{D_n}\right)$ to positive $k$:
\begin{equation}
\mathcal{L}_{+ p}^{(b)}\left(\frac{1}{D_n}\right) := q^{\frac{b(2p-b)}{4p}}  
\mathcal{L}_{1}^{(0)}\left(\frac{1}{D_n}\right)
P_{n}^{p,b}(q)= q^{\frac{b(2p-b)}{4p}}  \frac{(-1)^{n}q^{\frac{n(n+1)}{2}}}{\left(q^{n+1};q\right)_{n}}P_{n}^{p,b}(q),  \quad b~{\rm{even}},
\label{eq:expansions-for-surgeries-2}
\end{equation}
which gives\footnote{Note that here we are changing the surgery from $-p$ to $+p$ without changing the orientation of the knot. In contrast, following \cite{GM}, when flipping the orientation of the knot, it is natural to define the series of the mirror knot $m(K)$ as $F_{m(K)}(x,\tau) =F_{K}(x,-\tau)$.  }
\begin{equation}
\hat{Z}_{b}(S_{+p}^{3}(K);\tau)=  
\epsilon q^{-d+\frac{b(2p-b)}{4p}}\sum_{n=0}^{\infty}a_{-n-1}(K;q)\frac{(-1)^{n}q^{\frac{n(n+1)}{2}}}{\left(q^{n+1};q\right)_{n}}P_{n}^{p,b}(q)~
\label{eq:surgery formulas}
\end{equation}
for even $b$. 

After reviewing the conjectured expression \eqref{eq:surgery formulas} for $(+p)$-surgeries \cite{Park2106}, we discuss how to combine the inverted Habiro description for the $F_K$ series with the inclusion of defect operators, in accordance with Definition \ref{conj:surgery_defect}. 
Note that the proposed Wilson line expression \eqref{eq:surgeryS2_defect} combined with the expression \eqref{eq:inverted habiro} for the $F_K$-series leads, for the case of a plumbed knot $K$,  to 
\begin{multline}\label{eq:habiro_defect}
\hat{Z}_{b}(S_{\pm p}^{3}(K);W_{\nu};\tau)=\epsilon q^{\mp d}\left(a_{-1}(K; q)\mathcal{L}_{\pm p}^{(b+\nu)}\left(\chi_\nu(x^\ha)\right) \right. \\
+\left.\sum_{m=1}^{\infty} \mathcal{L}_{\pm p}^{(b+\nu)}\left[\chi_\nu(x^\ha)\frac{a_{-m-1}(K;q)}{D_m}\right]\right)~.
\end{multline}

Below we will work out their explicit expression. 
For that, first we note the following lemma, which will be relevant for defects with highest weights that are not roots. It can be proven in a way analogous to Proposition 4.1.5. of \cite{park-thesis},  {using the qZeil algorithm \cite{qZeilpaper}. 
\begin{lem}
    The following identity holds
\begin{equation} 
 \frac{x^{\ha}+x^{-\ha} }{\prod_{j=1}^{n} (x+x^{-1}-q^j-q^{-j})} \bigg|_{x^u\mapsto q^{u^2}} = \frac{q^{n^2-n+\frac{1}{4}}}{(q^{n};q)_n}~, 
\end{equation}
where $\big|_{x^u\mapsto q^{u^2}}$ denotes first series expand in $x$ and then substitute   $x^u$ with $q^{u^2}$.
\end{lem}
\vspace{5pt}

From the above lemma, we get 
\begin{equation} \label{eq:expansions-for-surgeries-3}
{\mathcal L}^{(1)}_{-1}\left(\frac{x^{\ha}+x^{-\ha} }{D_n}\right) 
= \frac{q^{n^2-n+\frac{1}{4}}}{(q^{n};q)_n} ~. 
\end{equation}
With the above result, and noting that the right-hand side can be defined both for $|q|<1$ and for $|q|>1$, we define the $(+1)$-surgery using \eqref{q_hyp}  as
\begin{equation}\label{L11}
{\mathcal L}^{(1)}_{1}\left(\frac{x^{\ha}+x^{-\ha} }{D_n}\right) := (-1)^n \frac{q^{\frac{n^2+n}{2}-\frac{1}{4}}}{(q^{n};q)_n} ~. 
\end{equation}
{
Similarly to \eqref{eq:expansions-for-surgeries-1}, with the notation 
$$
\delta_b = \begin{cases} 
1 & b ~{\rm odd}\\
0& b ~{\rm even}, 
\end{cases}
$$
we extend the definition of  the polynomials $P$ to odd $b$ by setting 
\begin{equation}
P_{n}^{p,b}(q^{-1}) :=  q^{\frac{b(2p-b)}{4p}{
{-}\frac{\delta_b}{4}}
}\frac{ \mathcal{L}_{- p}^{(b)}\left(\frac{(x^{\ha}+x^{-\ha})^{\delta_b}}{D_n}\right)}{\mathcal{L}_{- 1}^{({1})}\left(\frac{(x^{\ha}+x^{-\ha})^{\delta_b}}{D_n}\right) },~ b\in \mathbb Z .
\label{eq:expansions-for-surgeries-1 odd}
\end{equation}}

\bigskip

{
\noindent Analogous to \eqref{eq:expansions-for-surgeries-2}, which is relevant for the case with even highest weight $\nu_i$, from \eqref{L11} we also define the $+p$-surgery counterpart of the above through 
{
 \begin{gather}\label{eq:expansion-for-surgeries-Wils}\begin{split}
\mathcal{L}_{+p}^{(b)}\left(\frac{x^{\ha}+x^{-\ha} }{D_n}\right)  &:=
 q^{\frac{b(2p-b)}{4p}{
 {-}\frac{1}{4}}}\,P_{n}^{p,b}(q)  \,
{\mathcal L}^{(1)}_{1}\left(\frac{x^{\ha}+x^{-\ha} }{D_n}\right) \\&= 
(-1)^n \frac{q^{\frac{n^2+n}{2}}}{(q^{n};q)_n} q^{\frac{b(2p-b)}{4p} {-\frac{1}{2}}}P_{n}^{p,b}(q) ,  ~ b~{\rm{odd}},
\end{split}\end{gather}}
relevant for the case with odd highest weight $\nu_i$.
We can summarize the above as the following: 
\begin{dfn}
For $p,b,m\in {\mathbb N}$, write
\begin{gather}\begin{split}
f_{m}^{p,b}(q^{-1}) = q^{-\frac{b(2p-b)}{4p}{+}\frac{\delta_b}{4}}  
{\mathcal L}^{(\delta_b)}_{-1}\left(\frac{(x^{\ha}+x^{-\ha})^{\delta_b} }{D_m}\right) =  
q^{{\frac{\delta_b}{2}}-\frac{b(2p-b)}{4p}} \frac{q^{m(m-\delta_b)}}{(q^{m+1-\delta_b};q)_m}
\end{split}\end{gather}
or equivalently 
\begin{equation}\label{dfn:f}
f_{m}^{p,b}(q) = (-1)^m q^{\frac{b(2p-b)}{4p}{-\frac{\delta_b}{2}}} \frac{q^{\frac{m^2+m}{2}}}{(q^{m+1-\delta_b};q)_m}.
\end{equation}
Define the polynomial 
$P_{m}^{p,b}(q)$  by \eqref{eq:expansions-for-surgeries-1 odd}, or equivalently 	
\begin{equation}\label{dfn:Ppoly}
\mathcal{L}_{-p}^{(b)}\left(
\frac{(x^{\ha}+x^{-\ha})^{\delta_b}}{ D_m} \right)
= f_{m}^{p,b}(q^{-1})  P_{m}^{p,b}(q^{-1}). 
\end{equation}
Using the above ingredients we define
\begin{equation}
\mathcal{L}_{{+}p}^{(b)}\left(
\frac{(x^{\ha}+x^{-\ha})^{\delta_b}}{D_m} \right)
:=f_{m}^{p,b}(q) \, P_{m}^{p,b}(q)~~. 
\end{equation}
\end{dfn}
\noindent An explicit analysis of the polynomials $P_{m}^{p,b}$, including their closed form expressions, can be found in Appendix \ref{app:Pnpb}. 

\bigskip 

\noindent With the above definition, 
we are almost ready to write down  our proposal for the generalization of the conjectural expression of the $+p$-surgery formulas of \cite{Park2106}, for the cases where defects are present. 
In order to compute the Laplace transform \eqref{eq:surgeryS2_defect} involving the general $\mathfrak{sl_2}$ characters $\chi_\nu$ \eqref{eq:characterdefB}, we make use of the following Lemma: 
}

\begin{lem}\label{lem:LambdaPoly}
Writing $X=x+x^{-1}$,  
we have 
 \begin{gather}\label{character_invertedLemma}
\begin{split}
&(  x^{\ha} +x^{-\ha})^{2n-\nu}  \frac{\chi_{\nu}(x^\ha)}{D_m}\\
&=\sum_{j=0}^{n} c_j^{(\nu/2)} S_{\nu-n-1-m-j}(X,Q_m,\dots,Q_{1})
 + \sum_{\ell=0}^{{\rm min}(n,m-1)} \frac{1}{D_{m-\ell} } \sum_{j=0}^{ {n-\ell}} c_j^{(\nu/2)} S_{n-\ell-j}(Q_m,\dots,Q_{m-\ell})
\end{split}
\end{gather}
where $n=\lfloor \frac{\nu}{2}\rfloor$, 
\begin{gather}\label{character_nsmall_coeff}
\begin{split}
 c^{(n)}_{j} \;&=\; (-1)^{\lfloor \frac{j}{2}\rfloor}\,
\binom{\,n-\lceil \frac{j}{2}\rceil\,}{\,\lfloor \frac{j}{2}\rfloor\,}\\
c^{(n+\ha)}_{2j} \;&=\; (-1)^{j}\,
\binom{\,n-j\,}{\,j\,}, ~~c^{(n+\ha)}_{2j+1}=0\,.
\end{split}
\end{gather}
In the above, $S_n$ denotes the complete homogeneous symmetric polynomial of degree $n$: $$S_n(x_1,\dots,x_k) = \sum_{\substack{i_\ell \geq 0\\ \sum_{\ell=1}^k i_\ell =n}}  x_1^{i_1}\dots x_k^{i_k},  ~{\rm for}~ n\geq 0~, ~~S_n(x_1,\dots,x_k) =0 ~{\rm for}~n<0.$$
\end{lem}
\vspace{5pt}
\noindent The proof of the Lemma is recorded in Appendix \ref{app:prooflem2}. Putting the above together, and using 
\eqref{eq:habiro_defect}, we are led to the following conjecture. 

{
\begin{con}\label{con:habiro}
When $p\in \NN$ and knot $K$ is such that the inverted Habiro series expression \eqref{eq:surgery formulas 1} and \eqref{eq:surgery formulas} lead to the correct $\hat Z$-invariant for $S^3_{\pm p}(K)$, inserting a defect operator supported on $K$ in the highest representation of $\mathfrak{sl}_2$ with highest weight $\nu \vec\omega$ leads to the following defect invariant: 
\begin{gather}\label{conj_invertedHabiro}
\begin{split}
&\hat{Z}_{b}(S_{\pm p}^{3}(K);W_{\nu};\tau)=
{\rm poly}(q) \\ &+ \epsilon q^{\mp d} \sum_{m=1}^\infty a_{-m-1}(K;q^{\mp 1})\sum_{\ell=0}^{{\rm min}(n,m-1)}  f_{m-\ell}^{p,b+\nu} (q^{\pm1})  P_{m-\ell}^{p,b+\nu}(q^{\pm1}) \sum_{j=0}^{ {n-\ell}} c_j^{(\nu/2)} S_{n-\ell-j}(Q_m,\dots,Q_{m-\ell}) ,
\end{split}
\end{gather} 
where $n=\lfloor \frac{\nu}{2}\rfloor$ and
for some finite polynomial ${\rm poly}(q)$ of $q$ and $q^{-1}$. 
\end{con}
}

\vspace{0.5cm}
\begin{rmk}

\begin{itemize}
\item
As the reader might have noticed, the proposal on the inverted Habiro series and its extension to the cases 
of Laplace transforms for positive integer surgeries is, from the point of view of the $q$-series, 
highly experimental. In this paper we focus on how, when working with the cases where that proposal of \cite{Park2106} does lead to a correct answer, one can extend the proposal to include supersymmetric defect operators. 
We will concretely test this proposal in specific examples in \S\ref{sec:ex}. 
\item 
Note that this method can be used to compute $\hat Z$-invariants both of the ``false" as well as of the ``mock" kind, as exemplified in \S\ref{subsec:habiroexamples}. 
\item
In \S\ref{subsec:41} we present some hints that it might be possible to arrive at similar expressions for knots other than plumbed knots. 
\end{itemize}
\end{rmk}

\subsection{Modularity}
\label{sec:Modularity}

In \S\ref{sec:intro} we have discussed our general expectation that there exists an $SL_2(\ZZ)$ representation underlying the $\hat Z$ invariants of a given three-manifold.  
After proposing the defect invariants in Definitions \ref{conj:plumbed_defect} and \ref{conj:surgery_defect}, we are ready to discuss the content of the Modularity Conjecture \ref{con:modularity}, which is really a concrete special case of {our} general expectation about vector-valued quantum modularity discussed in \S\ref{sec:intro}. 
We will first start with a detailed description of the Weil representation $\Theta^{m+K}$ featured in the conjecture.

\subsubsection*{Weil representations}\label{sec:WeilReps}

The $\hat{Z}$-invariants have been observed \cite{3d} to be related to certain concrete Weil representations $\Theta^{m+K}$ of the metaplectic group  $\widetilde {\rm SL}_2(\ZZ)$, which are subrepresentations of the $2m$-dimensional representation $\Theta_m$ spanned by the column vector $\theta_m=(\theta_{m,r})_{r\, \mathrm{mod}\, 2m}$  with $m\in\mathbb{Z}_{>1}$ and theta function components \eqref{eq:thetafunc}. 
Their derivatives define unary theta functions 
\begin{equation}\label{dfn:theta}
  \theta^\ell_{m,r}(\tau) := \left(\left({\frac{1}{ 2\pi i}} \frac{\partial}{\partial z}\right)^{\ell}   \,\theta_{m,r}(\tau,z) \right)\Bigg\lvert_{z=0}~, \qquad \ell=0,1 ~,
\end{equation}
with $\widetilde{\theta}_{m,r}(\tau)$ in equation \eqref{eq:intro:falsetheta} proportional to the Eichler integral \eqref{dfn:Eichler_integral} of $\theta^1_{m,r}(\tau)$.

To define the Weil representation $\Theta^{m+K}$, one starts with $\mathrm{Ex}_m$, the group of exact divisors of $m$, where a divisor $n$ of $m$ is exact if $\left(n,m/n\right)=1$ and the group multiplication for $\mathrm{Ex}_m$ is $n\ast n':=nn'/(n,n')^2$. 
For a subgroup $K\subset\mathrm{Ex}_m$ one can construct in the following way
a subrepresentation of $\Theta_m$ denoted $\Theta^{m+K}$. 

Consider the space of matrices commuting with the $S$- and $T$-matrices on $\Theta_m$. It is generated by the  $\Omega$-matrices, $\Omega_m(n)$ with $n|m$ \cite{Cappelli:1987xt}, defined by 
\begin{equation}\label{omegamatrices}
 \Omega_m(n)_{r,r'} =\delta_{r+r'\Mod{2n}}\delta_{r-r'\Mod{2m/n}}
\end{equation}
and it represents the group $\mathrm{Ex}_m$: for $n,n'\lvert m$ and  $n\in \mathrm{Ex}_m$, one has $\Omega_m(n)\Omega_m(n')=\Omega_m(n\ast n')$.
As a result, one can define the
 projection operators 
\begin{equation}\label{dfn:proj}
    P^\pm_m(n) :=\left( \mathbf{1}_m \pm \Omega_m(n) \right) /2 ~, \qquad n\in{\rm Ex}_m ~,
\end{equation}
satisfying $ (P^\pm_m(n) )^2= P^\pm_m(n) $. 
For subgroups $K$ with $m\not\in K$, we define the projector
\begin{equation}\label{dfn:projB}
P^{m+K} :=    \left(\prod_{n\in K}P^+_m(n)\right) P^-_m(m)~.
\end{equation}
With these, $\Theta^{m+K}$ can explicitly be constructed from $\Theta_m$  
as spanned by $(\theta^{m+K}_{r})_{r\in \sigma^{m+K}}$  
\begin{equation}\label{linear_weil}
    \theta^{m+K}_{r}(\tau,z) :=  2^{\frac{|K|}{2}}\sum_{r'\in \mathbb{Z}/2m} P_{r,r'}^{m+K} \theta_{m,r'}(\tau,z)~,
\end{equation}
where $\sigma^{m+K}\subset \mathbb{Z}/2m$ is the set which labels the linearly independent functions $\theta^{m+K}_{r}(\tau,z)$. 
When $K$ is maximal, in the sense that $\mathrm{Ex}_m = K\cup (m\ast K)$, and $m$ is square-free, the Weil representation $\Theta^{m+K}$ is irreducible.
If $K$ is maximal but $m$ is not square-free, then the irreducible representation 
$\Theta^{m+K,{\rm irred}}$ \cite{SkoruppaThesis, 2007arXiv0707.0718S} can be constructed by using the projection operators 
\begin{equation}\label{irred_weil}
P^{m+K, {\rm irred}} :=    \left(\prod_{n\in K}P^+_m(n)\right)
\left( \prod_{f^2|m}(\mathbf{1}_m-\tfrac{1}{f} \Omega_m(f) ) \right) P^-_m(m)~.
\end{equation}

Similarly to equation \eqref{linear_weil}, we introduce the notation
\begin{equation}\label{linear_weilB}
    \widetilde{\theta}^{m+K}_{r}(\tau) := 2^{\frac{|K|}{2}} \sum_{r'\in \mathbb{Z}/2m} P_{r,r'}^{m+K} \, \widetilde{\theta}_{m,r'}(\tau)~ 
\end{equation}
and similarly for $\theta^{1,m+K}_r$ (cf. \eqref{dfn:theta}).
We also define $ \widetilde{\theta}^{m+K,{\rm{irr}}}_{r}$, ${\theta}^{1,m+K,{\rm{irr}}}_{r}$, and $\sigma^{m+K,{\rm{irr}}}$ analogously using \eqref{irred_weil}. 
It will be in terms of these Weil representations that we phrase a concrete conjecture (Conjecture \ref{con:modularity}) as a special realization of the general modularity expectation that we outlined in~\S\ref{sec:intro}. In expressing the data $m$ and $K$, we use the notation of writing $m+K$ and skip writing the identity element in $K$, as in \cite{Conway:1979qga}. For instance, we write $42+6,14,21$ to denote $m=42$, $K=\{1,6,14,21\}$.

{The relationship between $\hat Z$ invariants and concrete quantum modular forms is captured by an equivalence relation: we say two infinite $q$-series are equivalent in the following sense
\begin{equation}\label{dfn:equivalence}
f_1\sim f_2 ~{\rm if}~f_1=C\,q^{\Delta} f_2 + q^{\Delta'}p(q)
\end{equation}
where $C\in \CC$, $\Delta, \Delta'\in \QQ$ and $p(q)\in \CC[q,q^{-1}]$  is a finite polynomial. In the rest of the paper we will continue to use $\sim$ to denote this equivalence. We also extend the equivalence between infinite $q$-series to their spans. Namely, we say $V_1\sim V_2$ if for all $v_1\in V_1$ there is a $v_2$ such that $v_1\sim v_2$.}

\subsubsection*{The Modularity Conjecture}

\begin{con}\label{con:modularity}

Consider a Seifert manifold $M_3$ with three singular fibers. 
Define 
 \begin{equation}\label{dfn:span}
 {\rm span}(\hat Z(M_3)):= {\rm span}_\CC \{\hat Z_b(M_3,W_{\bnu};\tau), b\in {{\mathrm{Spin}}^c(M_3)}\lvert \bnu\in \NN^3\}. \end{equation} 
Then there exists a Weil representation 
\begin{equation}
\Theta^{(M_3)} = \Theta^{m+K} ~~{\rm{or}}~~ \Theta^{(M_3)} =\Theta^{m+K,{\rm irr}}
\end{equation}
for some
 positive integer $m$ and a subgroup $K\subset {\rm Ex}_m$, 
such that the following is true. 
\begin{enumerate}
\item When $M_3$ is a negative Seifert manifold, we have    
$${\rm span}(\hat Z(M_3))\sim {\rm span}_\CC \{ \widetilde{\theta}^{(M_3)}_{r}\lvert r\in \ZZ/2m \}. $$ 
\item When $M_3$ is a positive Seifert manifold, there is a $SL_2(\ZZ)$ vector-valued (mixed) mock modular form $f^{(M_3)}=(f_r^{(M_3)})$ transforming in the dual representation of  $\Theta^{(M_3)}$, such that 
 $${\rm span}(\hat Z(M_3))\sim {\rm span}_\CC \{ f^{(M_3)}_{r}\lvert ~ r\in \ZZ/2m\}\,.$$
\end{enumerate}
\end{con}
  \addtocounter{con}{-1}

We can prove the above conjecture for all Brieskorn spheres. 
\begin{thm}\label{thm:modularityBrieskorn} 
The  Conjecture \ref{con:modularity}-1. is true when $M_3$ is $\Sigma(p_1,p_2,p_3)$, with $m=p_1p_2p_3$, $K=\{1, p_1 p_2, p_2p_3, p_1 p_3\}$. 
More precisely, we have 
\begin{equation}\label{exp:Brieskorn_wilson}
\hat Z_0(\Sigma(p_1,p_2,p_3); W_{\bnu};\tau) \sim  \widetilde{\theta}^{m+K}_{r_{\bnu}} ~~,~~r_{\bnu}=m -\sum_i  (1+\nu_i )\bar p_i. 
\end{equation}

\end{thm}

\begin{proof}
When $M=\Sigma(p_1,p_2,p_3)$ is a Brieskorn sphere, we necessarily have $(p_i,p_j)=1$ when $i\neq j$.
Define $m=p_1p_2p_3$ and  $\bar p_i:= m/p_i$. 
From the Chinese remainder theorem, it is easy to see that, for any $r \in \ZZ/2m$, the equations 
\be
r+r'\equiv 0\Mod{2p_i},~r- r'\equiv 0\Mod{2\bar p_i}
\ee 
have a unique answer $r'\in \ZZ/2m$. 
Also, note that there exists $a_1,a \in \ZZ$ such that $1= a_1 \bar p_1 + a p_1$. 
Another application of the Chinese remainder theorem then implies that there exists $a_2, a_3 \in \ZZ$ such that $a=a_2 p_3+ a_3 p_2$, and hence there exists $ (a_1,a_2, a_3) \in \ZZ^3$ such that $1= \sum_{i=1}^3 a_i \bar p_i$. As a result,  every $r\in \ZZ/2m$ admits an expression in terms of $A_i\in \ZZ/p_i$ such that $r=\sum_{i=1}^3 A_i \bar p_i$.

Putting $m=p_1p_2p_3$, $K=\{1, p_1 p_2, p_2p_3, p_1 p_3\}$ in the definition \eqref{dfn:projB} and \eqref{linear_weilB} and using the properties of the $\Omega$-matrices, one obtains
\be\label{orbit}
\widetilde{\theta}^{m+K}_{r} =\sum_{r'\in \mathbb{Z}/2m} \left({\bf 1}+\sum_{i=1}^3 \Omega_m(\bar p_i)\right)_{r,r'} \widetilde{\theta}^{m+K}_{r'}.
\ee
It is easy to check that, when writing $r=m-\sum_{i=1}^3 A_i \bar p_i$ (which is always possible from the above argument), we have 
\be
\left(\Omega_m(\bar p_j)\right)_{r,r'}  =\delta_{r'-(m-\sum_i (-1)^{1+\delta_{i,j}} A_i \bar p_i)\Mod{2m}}, 
\ee
which, when combined with \eqref{orbit}, gives 
\be\label{eqn:folding}
\widetilde{\theta}^{m+K}_{r} =\sum_{\substack{(\epsilon_{1},\epsilon_{2},\epsilon_{3})\in (\ZZ/2)^{3} \\ \sum_{j}\epsilon_{j}\equiv 0\Mod{2}}}  \widetilde{\theta}_{m,m-\sum_{i=1}^3 (-1)^{\epsilon_i}A_i \bar p_i }     ,~~r=m-\sum_{i=1}^3 A_i \bar p_i ~.
\ee
Finally, applying  Definition \ref{conj:plumbed_defect} to the case where $M_3=\Sigma(p_1,p_2,p_3)$ and when the defect operators 
corresponding to the end nodes are inserted, one obtains  \eqref{exp:Brieskorn_wilson} after a routine calculation. 
We therefore conclude that by changing the highest weight $\nu_i\vec \omega$ of the defect operator, we can cover all the components of the Weil representation. Conversely, one can associate a Wilson line configuration to any component labelled by $r\in \sigma^{m+K}$. 
\end{proof}

{In \S\ref{subsec:indef_zhat}, we will also show that Conjecture \ref{con:modularity}{\it -2}. is compatible with our proposal for the mock $\hat Z$-invariants of certain positive plumbed three-manifolds $-M_3$. }

\bigskip

Given the prominent role of the Weil representation in Conjecture \ref{con:modularity}, it might be natural to consider combining the $\hat Z_b(W_{\bnu})$-invariants for different choices of $b$ and $\bnu$ into the Jacobi-like functions $\sum_r \overline{\widetilde{\theta}^{(M_3)}_r(\tau)} {\theta_{m,r}(\tau,z)}$ and $\sum_r f^{(M_3)}_r(\tau) {\theta_{m,r}(\tau,z)}$, for the case of Conjecture \ref{con:modularity}.1 and \ref{con:modularity}.2 respectively. We will call them quantum Jacobi theta functions, which can be defined in an obvious way by combining the definition of mock Jacobi theta functions and quantum modular forms. 
To be specific, let us define the following. 
\begin{dfn}\label{dfn:quantumJac}
We say that $\phi:\HH\times \CC\to \CC$ with $\phi(\tau,z) := \sum_{r\in \ZZ/2m} f_r(\tau) \theta_{m,r}(\tau,z)$ (resp. $\sum_{r\in \ZZ/2m} \overline{f_r(\tau)} \theta_{m,r}(\tau,z)$) is a weakly holomorphic (resp. skew-holomorphic) quantum Jacobi form  of index $m$ and weight $w+1/2$, if each $f_r$is a  weakly holomorphic function on $\HH$, and they form a vector-valued quantum modular form $(f_r)_{r\in \ZZ/2m}$ of weight $w$ for $\Gamma\subset SL_2(\ZZ)$, whose character is the dual of that of $(\theta_{m,r})_{r\in \ZZ/2m}$ (resp. that of $(\theta_{m,r})_{r\in \ZZ/2m}$). 
\end{dfn}

With the above definition, we can write down a natural corollary for Conjecture \ref{con:modularity}. 

{
\begingroup
 \renewcommand{\thecon}{\arabic{con}$'$}
\begin{con}\label{conj:Jac}
Use the same definition \eqref{dfn:span} of ${\rm span}(\hat Z(M_3))$ and the same equivalence relation \eqref{dfn:equivalence}.  
The following statements are true. 
\begin{enumerate}
\item When $M_3$ is a negative Seifert manifold, 
there exists  a skew-holomorphic quantum Jacobi form $$\phi(\tau,z) := \sum_{r\in \ZZ/2m} \overline{f_r(\tau)} \theta_{m,r}(\tau,z)$$ of index $m$ such that 
${\rm span}(\hat Z(M_3))={\rm span}_\CC\{f_r(\tau), r\in \ZZ/2m\}$. 
\item When $M_3$ is a positve Seifert manifold, there exists  a weakly holomorphic quantum Jacobi form $$\phi(\tau,z) := \sum_{r\in \ZZ/2m} {f_r(\tau)} \theta_{m,r}(\tau,z)$$ of index $m$ such that 
${\rm span}(\hat Z(M_3)) ={\rm span}_\CC\{f_r(\tau), r\in \ZZ/2m\} $. 
\end{enumerate}
\end{con}
\endgroup
}
\vspace{0.5cm}
{As in Theorem \ref{thm:modularityBrieskorn}, we have shown that the above is true when $M_3$ is a Brieskorn sphere. } 

\begin{rmk}
\begin{itemize}
\item
The importance of the defect $\hat Z$-invariants for the understanding of the modularity of $\hat Z$-invariants can be seen in the following. Without including the defect operators, in \cite{3d} it was noted that the space spanned by $\hat Z_b(M_3)$ correspond to subspaces of the Weil representation that do not close under the action of the modular group $\widetilde {\rm SL}_2(\ZZ)$. The space gets completed once the appropriate defect operators are included.

\item {It should also be possible to give a similar proof as that of Theorem \ref{thm:modularityBrieskorn} using the proposal (Definition \ref{conj:surgery_defect} and \ref{dfn:surgery_defectFK}) for the defect  $F_K$, when $S^3_{p/r}(K)$ is a Brieskorn sphere. We provide examples in \S\ref{subsec:examples_surgeries}. }

\item It is easy to see why 
there are, up to the addition of finite polynomials, finitely many $\hat Z_b(M_3;W_{\bnu})$ even though $\bnu\in \NN^3$ in principle. 
From the expression of the contour integration \eqref{def:ZhatintegralWilson} and its result \eqref{exp:Brieskorn_wilson}, we see that the equivalence classes (up to the addition of finite polynomials to $\hat Z(W_{\bnu})$) of Wilson lines are given by $\nu_i\in \ZZ/p_i$, with the finite polynomials in the form of
$$\sum_{\ell\equiv r\Mod{2m}} {\rm{sgn}}(\ell) q^{\ell^2/4m} - \sum_{\ell\equiv r\Mod{2m}} {\rm{sgn}}(\ell+ 2mN ) q^{\ell^2/4m}$$ for some finite integer $N$.
\item 
All quantum modular forms we encounter here are the so-called {\emph{strong}} quantum modular forms in the language of \cite{qmf}. We expect this to be true for all quantum modular forms related to $\hat Z$-invariants. 
\item 
Note that our definition of Jacobi-type quantum modular objects (Definition \ref{dfn:quantumJac}) is different from that in \cite{FolsomQuanJac} and references therein. 
\item In Conjecture Conjecture \ref{con:modularity} and 4', note that  $f_r$ for some $r\in \ZZ/2m$  vanish when the corresponding Weil representation is $\Theta^{m+K}\neq \Theta_m$. 
Equivalently, the span is by definition unchanged if instead of taking the set to be  $r\in \ZZ/2m$ one uses  $r\in \sigma^{m+K}$ resp.  $r\in \sigma^{m+K,{\rm irr}}$. The same statement also holds.    
\end{itemize}
\end{rmk}

\section{Mock \texorpdfstring{$\hat{Z}$}{Z}-invariants}
\label{sec:regindefTheta}

In this section we discuss  various considerations on constructing the mock \texorpdfstring{$\hat{Z}$}{Z}-invariants, and proposing explicit expressions for specific {instances} of three-manifolds, comparing different approaches, and analysing the modular properties in details.

\subsection{Regularised Indefinite Theta Functions}
\label{subsec:Regularised Indefinite Theta Functions}

In this subsection we recall the basic construction of indefinite theta functions and their relation to mock modular forms, based on \cite{zwegers} and mostly following the notation of \cite[Chapter 8]{bringmann2017harmonic}. {These regularized indefinite theta functions are naturally mixed mock modular forms, which we now introduce. }

{
In the introduction (below \eqref{dfn:Eichler_integral-Nhol}), we have reviewed the modern definition of mock modular forms. 
This definition allows for a few natural generalizations, including for instance higher-depth mock modular forms, which are expected to be relevant for $\hat Z$ invariants for plumbed manifolds with two or more non-negative high valency vertices, and for higher rank gauge groups. The generalization we would like to highlight here is that of mixed mock modular forms \cite{QBH}. 
To emphasize the difference between mock and mixed mock properties, mock modular forms are also sometimes referred to as ``pure" mock modular forms.
The novel characteristic of mixed mock modular forms is that they allow a more complex structure 
in their relation to the completion function (the associated non-holomorphic function that transforms like a modular form).  
For instance, a simple multiplication of a mock modular form by a modular form
leads to a function that is no longer a pure mock modular form but is still a mixed mock modular form. }

{
Formally, a mixed mock modular form is the first function of a triplet $(f,\{h_j\}_{j\in I}, \{g_j\}_{j\in I})$ of finite sets of functions  defined on the upper half plane $\HH$, such that $f$ is weakly holomorphic and 
the non-holomorphic function  $\hat{f}:=f -\sum_j h_j g^{*}_j$, called the completion of $f$, transforms as a modular form of weight $k$.
In our work, we let the functions $h_j$ be weakly holomorphic and $g_j$ are restricted to be cusp forms. Moreover, $h_i$ and $g_i$ transform modularly with weights $\ell$ and $2-k+\ell$ respectively, for some $\ell \in \tfrac{1}{2}+\ZZ$. }

Here we limit our discussion to two-dimensional indefinite lattices that are relevant for the present paper, though the construction is completely analogous for general lattices of signature $(1,n)$. See also \cite{Alexandrov:2016enp,Nazaroglu_2018} for the analysis of general indefinite signature lattices.

For the two-dimensional lattice $L\cong \ZZ^2$, we denote by $A$ the symmetric $2\times 2$ matrix with integral entries giving rise to the bilinear form of the lattice, which we will take to have signature $(1,1)$. One can extend the norm to $L\otimes_\ZZ \RR$ by defining  \begin{equation}\label{def:quadform}
    \normsq{\sx}:= \sx^T A\sx
\end{equation}
for $\sx\in\mathbb{R}^2$. We will also denote the inner product as
\begin{equation}\label{def:bilform}
    \inner{\sx}{\sy}:= \sx^T A \sy~.
\end{equation}
With these, the set of vectors $\ssc\in\mathbb{R}^2$ with negative norm $\normsq{\ssc}<0$ splits into two connected components. To specify one component, for a fixed vector $\ssc_0$ with $\normsq{\ssc_0}<0$, we denote with 
\begin{equation}
    C_{B}(\ssc_{0}):=\{ \ssc\in\mathbb{R}^2 | \normsq{\ssc}<0, B(\ssc,\ssc_0)<0\} ~
\end{equation}
the component containing $c_0$. We also denote with
\begin{equation}
    S_{B}(\ssc_0):=\{\ssc=(c_1,c_2)\in \mathbb{Z}^2 |\,  \mathrm{gcd}(c_1,c_2)=1, \normsq{\ssc}=0, \inner{\ssc}{\ssc_0}<0\}~
\end{equation}
the set of integral primitive vectors on the boundary of the component, a set of representatives of the cosets for the cusps of $C_B(\ssc_0)/ {\mathbb R}_+$  with respect to the orientation-preserving lattice orthogonal group, and 
 $\bar{C}_{B}(\ssc_0) :=C_{B}(\ssc_0)\cup S_{B}(\ssc_0)$ the corresponding compactification. Then, for any  $\ssc_1, \ssc_2\in \bar{C}_{B}(\ssc_0)\cap  \ZZ^2$, one can 
 define 
 the regularization factor as
\begin{equation}\label{eq:regulatingfactor}
    \rho^{\ssc_1,\ssc_2}(\sn):= {\rm sgn}(\inner{\ssc_1}{\sn})-{\rm sgn}(\inner{\ssc_2}{\sn})~.
\end{equation}
It was shown in \S2 of \cite{zwegers} that the norm is positive definite when restricted to the support of $ \rho^{\ssc_1,\ssc_2}$.
 As a result, for
$\sa \in \mathcal{R}(\ssc_1)\cap \mathcal{R}(\ssc_2)$, $\ssb\in\RR^2$ with
\begin{equation}
    \mathcal{R}(\ssc):= \begin{cases} \mathbb{R}^2 & \mathrm{if}~\normsq{\ssc}<0 \\
    \{\sa\in \mathbb{R}^2 | B(\ssc,\sa)\notin\mathbb{Z}\} &\mathrm{if}~\normsq{\ssc}=0
    \end{cases}
\end{equation}
one can define
 the regularised indefinite theta function  \cite{zwegers}
\begin{equation}\label{def:regularisedindefinitetheta}
    \Theta_{A,\sa,\ssb,\ssc_1,\ssc_2} (\tau) := \sum_{\sn\in \sa+\mathbb{Z}^2} \rho^{\ssc_1,\ssc_2}(\sn)\, \ex(\inner{\sn}{\ssb}) q^{\frac{\normsq{\sn}}{2}},
\end{equation}
 which is absolutely convergent on the upper-half plane. In the above we write $\ex(x):=e^{2\pi i x}$ for $x\in\CC$. 
 {For later convenience and to obtain simpler coefficients, we also define a rescaled version of the above regularised indefinite theta function
 \begin{equation}\label{def:rescaled_regularisedindefinitetheta}
   \tilde \Theta_{A,\sa,\ssb,\ssc_1,\ssc_2} (\tau) :={\ex(-\inner{\sa}{\ssb}) \over 2} \sum_{\sn\in \sa+\mathbb{Z}^2} \rho^{\ssc_1,\ssc_2}(\sn)\, \ex(\inner{\sn}{\ssb}) q^{\frac{\normsq{\sn}}{2}},
\end{equation}}

 It was shown that $\Theta_{A,\sa,\ssb,\ssc_1,\ssc_2}$ enjoys a certain mock modular property. Namely, it is a mixed mock modular form.
Concretely, $\Theta_{A,\sa,\ssb,\ssc_1,\ssc_2}$ can be viewed as the holomorphic part of
\begin{equation}\label{def:regularisedindefinitetheta-completion}
    \hat{\Theta}_{A,\sa,\ssb,\ssc_1,\ssc_2} (\tau,\bar{\tau}) := \sum_{\sn\in \sa+\mathbb{Z}^2} \hat{\rho}^{\ssc_1,\ssc_2}(\sn;\tau) e(\inner{\sn}{\ssb}) q^{\frac{\normsq{\sn}}{2}},
\end{equation} 
where
\begin{equation}\label{eq:completion-regulator}
\begin{split}
      \hat{\rho}^{\ssc_1,\ssc_2}(\sn;\tau)&:=\hat{\rho}^{\ssc_1}(\sn;\tau)-\hat{\rho}^{\ssc_2}(\sn;\tau)\\
      \hat{\rho}^{\ssc}(\sn;\tau)&=
      \begin{cases}
        E\left(\frac{B(\ssc,\sn)(\Im(\tau))^{\tfrac{1}{2}}}{\sqrt{-2\normsq{\ssc}}}\right) & \text{if }\normsq{\ssc}<0\\
        \sign\left(B(\ssc,\sn)\right) & \text{if } \normsq{\ssc}=0
      \end{cases}~,
\end{split}
\end{equation} 
and the non-holomorphic completion 
$\hat{\Theta}_{A,\sa,\ssb,\ssc_1,\ssc_2}$ transforms like a weight one modular form. In the language of \cite[Chapter 8]{bringmann2017harmonic}, $\hat{\Theta}_{A,\sa,\ssb,\ssc_1,\ssc_2}$  is a weight one mixed harmonic form. See also \eqref{prop:8-33HarmonicMaassFormsV2}. 
In particular, if we consider $\sa,\ssb \in L^\ast$ in the dual lattice, then 
 \be\label{SL2_IND}
    \hat\Theta_{A,\sa,\ssb,\ssc_1,\ssc_2} (-\tfrac{1}{\tau}) = \frac{\tau{\ex}(B(\sa,\ssb))}{\sqrt{|{\rm{det}}(A)|}}  \sum_{{\boldsymbol \mu \in{ L^\ast/L}}}  \hat \Theta_{A,{\boldsymbol \mu}+\ssb,-\sa,\ssc_1,\ssc_2}(\tau).
\ee

Its  $\bar{\tau}$-derivative, specifying the deviation from modularity of the holomorphic function $\Theta_{A,\sa,\ssb,\ssc_1,\ssc_2}$, is given by 
(see \cite[Proposition 8.33]{bringmann2017harmonic}) 
\begin{gather}\label{prop:8-33HarmonicMaassForms}\begin{split}
& \overline{\frac{\partial}{\partial\bar{\tau}} (\hat\Theta_{A,\sa,\ssb,\ssc_1,\ssc_2} (\tau))} =  \frac{1}{i \sqrt{2\Im(\tau)}} \, \times \\&  \sum_{j=1,2} (-1)^j { \sqrt{-|\ssc_j|^2 }} \sum_{\ssl\in P_{j}} g_{\frac{B(\ssc_j,\ssl+ \sa)}{|\ssc_j|^2},-B(\ssc_j,\ssb)}(-|\ssc_j|^2\tau) \overline{\sum_{\nu\in(\ssl+ \sa)^\perp_j + \langle \ssc_j \rangle^\perp_\mathbb{Z}} e(B(\nu,\ssb)) q^{\frac{\normsq{\nu}}{2}}} ~. 
\end{split}\end{gather} 
In the above, we denote by $\ssl^\perp_j:=\ssl - \frac{B(\ssc_j,\ssl)}{B(\ssc_j,\ssc_j)}\, \ssc_j$  the projection of the vector $\ssl$ onto the line that is orthogonal to $\ssc_j$ and write 
\begin{equation} 
\langle \ssc \rangle^\perp_\mathbb{Z} := \{ \lambda\in\mathbb{Z}^2 : 
B(\ssc, \lambda)=0 \}~.  
\end{equation}
Denote by $\ssc_{j, \perp}$  a generator of $\langle \ssc_j \rangle^\perp_\mathbb{Z}$, then $P_j$ can be chosen to be any representative of the quotient $\ZZ^2/   \ZZ \ssc_j \oplus  \ZZ \ssc_{j,\perp} $ and is in particular independent of $\sa$.

The function $g_{\rho,\beta} (\tau)$
\begin{equation}
    g_{\rho,\beta} (\tau) := \sum_{n\in \rho + \mathbb{Z}} n e^{2\pi i n\beta} q^{\frac{n^2}{2}},
\end{equation}
in the case of $\beta\in \ZZ$, is related to a unary theta function as defined in equation \eqref{dfn:theta} through 
\begin{equation}\label{eq:gtotheta}
    g_{\frac{r}{2m},\beta} (\tau) = \sum_{x\in \frac{r}{2m} + \mathbb{Z}}  x {\ex}({x \beta}) q^{\frac{nx^2}{2}}  = \frac{{\ex}(\frac{ r \beta}{2m})}{2m} \theta^1_{m,r} \left( \frac{\tau}{2m} \right) ~.
\end{equation}

Next we consider the special cases where the following additional properties hold:
\be\label{eqn:parameters_ind}
\ssc_j\in\ZZ^2,~B(\ssc_j, \ssb) \in \ZZ ~~,~~ -{\frac{1}{2}}|\ssc_j|^2 = \frac{n_j}{d_j} m ~~{\rm with}~  m\in \ZZ_+~~ 
{\rm for }~~ j=1,2~ , 
\ee
where $n_j,d_j$ are coprime natural numbers.
For these cases, the mixed-mock structure simplifies into the following expression

\begin{gather}\label{prop:8-33HarmonicMaassFormsV2}
\begin{split}
2 i {\sqrt{ \Im(\tau)}}\, \overline{\frac{\partial}{\partial\bar{\tau}} (\hat\Theta_{A,\sa,\ssb,\ssc_1,\ssc_2} (\tau))} & = \sum_{j=1,2} (-1)^j   {\sqrt{\frac{d_j}{m \, n_j}}} \sum_{\ssl\in P_{j}}
{\ex}\left(\frac{d_j B(\ssc_j,\ssb)B(\ssc_j,\ssl+\sa)}{2n_jm}\right)\, \\
&\times\theta^1_{n_j m,-d_jB(\ssc_j,\ssl+\sa)} \left(\frac{\tau}{d_j}\right)\overline{\sum_{\nu\in(\ssl+\sa)^\perp_j + \ZZ \ssc_{j, \perp} } e(B(\nu,\ssb)) q^{\frac{\normsq{\nu}}{2}}} ~.
\end{split}
\end{gather}

\subsection{Indefinite Theta Functions as Mock \texorpdfstring{$\hat{Z}$}{Z}-invariants}
\label{subsec:indef_zhat}

We now turn to the second method to construct false-mock pairs, which unlike the inverted Habiro series discussed in \S\ref{subsec:knot} is designed to produce topological three-manifold invariants with manifest (mixed) mock modular properties, making use of the indefinite theta functions discussed in  \S\ref{subsec:Regularised Indefinite Theta Functions}.  We will also comment on the relation between the two methods in \S\ref{subsec:relation}. 

For this we start by recalling the proposal for the ``false side" \texorpdfstring{$\hat{Z}$}{Z}-invariants for weakly negative three-manifolds in discussed in \S\ref{subsec:contourintegral}. We focus on weakly negative Seifert three-manifolds with three exceptional fibres, corresponding to plumbing graphs with one central node $v_0$ and three legs. They can be expressed in terms of the Seifert data as $M(b;\{q_i/p_i\}_{i=1,2,3})$, which moreover have 
\be
{\mathfrak e}=b+\sum_i \frac{q_i}{p_i} <0. 
\ee
Let $D$ be the smallest positive integer such that $\frac{D}{{\mathfrak e}p_i}\in \ZZ$ for all $i$, and $m=-DM_{v_0,v_0}^{-1}$ given by the inverse plumbing matrix. Consider Wilson lines associated to the end nodes of the three legs in the plumbing graph. As before, we label the corresponding highest weight representations with ${\bnu} =(\nu_1,\nu_2,\nu_3) \in \NN^3$ as in \S\ref{sec:Modularity} and denote the corresponding defect $\hat Z$-invariant with $\hat{Z}_b(M_3;W_{\bnu})$. 
 The equation \eqref{def:ZhatintegralWilson} 
gives 
\be\label{false_onemoretime}
\hat{Z}_b(M_3;W_{\bnu};\tau) =  C(q) \, \sum_{\substack{\hat\epsilon=(\epsilon_1,\epsilon_2,\epsilon_3)\\\in (\ZZ/2)^{3}}}
\tilde\chi_{\hat \epsilon}(\tau)+p_{\hat \epsilon}(\tau)
\ee
after performing the contour integration along all $z_{v}$ with $v\neq v_0$ \cite{CCFFGHP2201}, 
where $ \tilde\chi_{\hat \epsilon}$ either vanishes or is given by the false theta function 
\be\label{dfn:chi}
\tilde\chi_{\hat \epsilon}(\tau) = 
(-1)^{\hat\epsilon} \sum_{\substack{\ell \equiv mv_{\hat \epsilon}+a_{\hat \epsilon}\\\Mod{2mD}} } q^{\frac{\ell^2}{4Dm}} \, {\rm{sgn}}(\ell)
\ee
for some $v_{\hat \epsilon}, a_{\hat \epsilon} \in \ZZ$,  $p_{\hat \epsilon}(\tau)$ a finite polynomial of $q$,  and $C(q) =c q^\Delta$ for some $c\in \CC$ and $\Delta \in \QQ$ is the pre-factor that can be explicitly computed from  \eqref{def:ZhatintegralWilson} which we will mostly ignore from now on.  In particualr, we will 
continue working with the equivalence relation \eqref{dfn:equivalence}.     
We will also  use the shorthand notation  
\[
(-1)^{\hat\epsilon} := (-1)^{\sum_i\epsilon_i} .
\]

From the arguments reviewed in \S\ref{sec:intro}, when mapping $M_3$ to $-M_3$ we expect to map $\tau\to -\tau$, an action that naively causes \eqref{dfn:chi} to cease being a convergent function in $\HH$. 
Inspired by the construction of the regularised theta function of indefinite signature reviewed in \S\ref{subsec:Regularised Indefinite Theta Functions}, we propose to regularize the function obtained by $\tau\to -\tau$ by first turning the sum over a one-dimensional negative-definite lattice sum into a sum over a signature $(1,n)$ lattice, by adding an $n$-dimensional lattice $\Lambda$ with positive-definite norm. Schematically, 
we insert into \eqref{dfn:chi} the identity \be \label{positive_lattice}
1=\frac{\sum_{v\in \Lambda+\gamma} q^{ |v|_{\Lambda}^2/2} \mathrm{e}(B_\Lambda (v,\rho))  }{ \theta_{\Lambda, \gamma,\rho} }. \ee
Subsequently, we regularize the resulting  $(1,n)$ lattice theta function using the method discussed in the previous subsection,
 restricting the sum over the $(1,n)$ lattice 
\be\label{sign_fac}
\sum_{v\in \Lambda+\gamma}  \sum_{\substack{\ell \equiv mv_{\hat \epsilon}+a_{\hat \epsilon}\\\Mod{2mD}} } {\rm{sgn}}(\ell) {\ex}(B_\Lambda (v,\rho)) q^{-\frac{1}{4Dm}\ell^2+|v|_{\Lambda}^2/2}
\ee
to a positive-definite cone. 
{To be concrete, we let $n=1$ and let $\theta_{\Lambda, \gamma,\rho}$ to be the Dedekind eta function $$\eta(\tau)=\sum_{k\in\mathbb{Z}} (-1)^k q^{\frac{3}{2}\left(k-\frac{1}{6} \right)^2}~,$$ as in \cite{CFS1912}, in the regularisation procedure.} The restriction to the cone, specified by two negative-norm vectors $\ssc_1$, $\ssc_2$ as in \S\ref{subsec:Regularised Indefinite Theta Functions}, is implemented by including a factor 
\be\label{cone_support}
{1\over 2}|\rho^{\ssc_1,\ssc_2}(\sn)|= \frac{1}{2}| {\rm{sgn}}(B(\ssc_1,\sn))-{\rm{sgn}}(B(\ssc_2,\sn)) |
\ee
which is equal to one inside either of the two components of the cone and zero elsewhere, leading to the restricted lattice sum
\be\label{dfn:chi_minus}\begin{split}
\tilde\chi_{\hat \epsilon}^{\rm reg}(-\tau) &:=
\sum_{\substack{\ell \equiv mv_{\hat \epsilon}+a_{\hat \epsilon}\\\Mod{2mD}} } q^{-\frac{1}{4Dm}\ell^2} \,{\rm{sgn}}(\ell)\\&\times \left(\frac{1}{\eta(\tau)  }\sum_{k\in\mathbb{Z}} (-1)^k q^{\frac{3}{2}\left(k-\frac{1}{6} \right)^2}
\frac{1}{2}| {\rm{sgn}}(B(\ssc_1,\sn))-{\rm{sgn}}(B(\ssc_2,\sn)) |
\right) \end{split}
\ee
where $\sn^T=(\frac{1}{2mD}\ell,k-\frac{1}{6})$ denotes the vector in the real vector space underlying the  two-dimensional lattice with bilinear form \eqref{def:quadform} given by $A={\rm diag}(-2mD,3)$,  
and the factor inside the bracket is the regularizing term. 
Choosing moreover 
\begin{equation}\label{eq:propregvect1}
    \ssc_{1} = \left( \begin{matrix}
             1 \\
              0  \end{matrix} \right) ~, 
\end{equation}
we obtain the following expression in terms of the regularized indefinite theta function \eqref{def:regularisedindefinitetheta} 
\begin{align} \label{dfn:chi_minus3}\begin{split}
\tilde\chi_{\hat \epsilon}^{\rm reg}(-\tau) &=\frac{(-1)^{\hat\epsilon} {\ex}(\tfrac{1}{12})}{2\eta(\tau)} 
\sum_{\sn\in \sa_{\hat \epsilon}+\ZZ^2}\rho^{\ssc_1,\ssc_2}(\sn)\, q^{\frac{\sn^T A \sn}{ 2}} {\ex}(B(\ssb,\sn)) \\&= \frac{(-1)^{\hat\epsilon} }{\eta(\tau)}  \tilde\Theta_{A,\sa_{\hat\epsilon},\ssb,\ssc_1,\ssc_2}   (\tau)
\end{split}\end{align}
where $\Theta_{A,\sa_{\hat\epsilon},\ssb,\ssc_1,\ssc_2}$ is the regularized indefinite theta function with
\be A= \left( \begin{matrix} -2mD & 0 \\ 0& 3 \end{matrix} \right), ~
  \ssb=\left( \begin{matrix}
              0 \\
              \tfrac{1}{6}  \end{matrix} \right) ~,
\ee
 and $\sa^T_{\hat\epsilon}=({mv_{\hat \epsilon}+a_{\hat \epsilon} \over 2mD},-\frac{1}{6})$, as in \eqref{dfn:chi}.  
As a result, for those manifolds $M_3$ to which the above regularization is appropriate, we  propose 
\be\label{dfn:mockzhat1}
\hat{Z}_b(-M_3;W_{\bnu};\tau) =  C(q^{-1}) \,  
\sum_{\substack{\hat\epsilon=(\epsilon_1,\epsilon_2,\epsilon_3)\\\in (\ZZ/2)^{3}}} \chi^{\rm reg}_{\hat \epsilon}(-\tau)+p_{\hat \epsilon}(-\tau)
\ee
with $\tilde\chi_{\hat \epsilon}^{\rm reg}(-\tau) =0$ for those ${\hat \epsilon}$ with $\tilde \chi_{\hat \epsilon} =0$ and otherwise given as in \eqref{dfn:chi_minus3}: 
\be\label{indef_zhat}
\tilde\chi_{\hat \epsilon}^{\rm reg}(-\tau) := 
\frac{(-1)^{\hat\epsilon} }{\eta(\tau)} 
\tilde \Theta_{A,\sa_{\hat\epsilon},\ssb,\ssc_1,\ssc_2} (\tau).   
\ee
For instance, when $M_3=\Sigma(p_1,p_2,p_3)$ is a Brieskorn sphere, we have $D=1$, $m=p_1p_2p_3$, and $\tilde\chi_{\hat \epsilon}^{\rm reg}(-\tau)\neq 0$ for all ${\hat \epsilon}$.

From the discussion in the previous subsection, we see that the function \eqref{dfn:mockzhat1} enjoys a manifest relation to mixed mock modular forms as it is given by indefinite theta functions. To completely specify the function, it remains to specify the second negative-norm vector $\ssc_2$.

{
For instance, it is known that the $\hat Z$-invariant for the homological sphere $-\Sigma(2,3,7)$, predicted to be given by Ramanujan's order seven mock theta function in \cite{3d} and computed to the leading orders using knot surgery formula in \cite{GM}, admits expressions in terms of indefinite theta functions \cite{CFS1912,Park2106}. 
Concretely, recall that for $M_3=\Sigma(2,3,7)$  
we have for the defect $\hat Z$-invariants \eqref{false_onemoretime} 
\be
\tilde \chi_{\hat \epsilon}(\tau) = (-1)^{\sum_i {\epsilon_i}} \tilde \theta_{m,\rsub}, ~~\rsub = m-\sum_{i=1}^3 (-1)^{\epsilon_i} (1+\nu_i) \bar p_i. 
\ee
See the proof of Theorem \ref{thm:modularityBrieskorn} for details. 
For $M_3=-\Sigma(2,3,7)$, we take 
\be
   \ssc_{2}=\left( \begin{matrix}
              3 \\
              14  \end{matrix} \right)
\ee
and obtain 
\be
\hat{Z}_{b_0}(-M_3;W_{\bnu};\tau) =  C(q^{-1}) \, \frac{1}{\eta(\tau)}  \sum_{\substack{\hat\epsilon=(\epsilon_1,\epsilon_2,\epsilon_3)\\\in (\ZZ/2)^{3}}}
(-1)^{\hat\epsilon} \tilde\Theta_{A,\asub,\ssb,\ssc_1,\ssc_2}(\tau)
\ee
with $\asub^T=( {\rsub\over 2m},-\frac{1}{6})$. Putting the above proposal for $A$, $\asub$, $\ssc_1$, $\ssc_2$ and $\ssb$ together, 
one can show that the above expression with no defects (${\snu}=(0,0,0)$) is given, up to a prefactor, by the  order seven mock theta function $F_0$: 
\begin{gather}\begin{split}
\hat{Z}_{b_0}(-M_3;\tau) &
\sim  \frac{1}{\eta(\tau)}  \sum_{\substack{\hat\epsilon=(\epsilon_1,\epsilon_2,\epsilon_3)\\\in (\ZZ/2){^{3}}}}
(-1)^{\hat\epsilon} \tilde\Theta_{A,\sa_{\hat\epsilon},\ssb,\ssc_1,\ssc_2}(\tau) \\& = H^{42+6,14,21}_1(\tau) = q^{-1/168}F_0(q) =q^{-1/168} \sum_{n\geq 0}  \frac{q^{n^2}}{(q^{n+1};q)_n},\end{split}
\end{gather}
where $H^{42+6,14,21}_1(\tau)$ is the mock modular form that is a component of the corresponding optimal mock Jacobi form \cite{CD1605}, and is one of Ramanujan's order 7 mock theta function, $F_0$, up to a pre-factor $2q^{-1/168}$.  Similarly, an analogous treatment for  $M_3=\pm\Sigma(2,3,5)$ leads to 
$$
\hat{Z}_{b_0}(-\Sigma(2,3,5);\tau)   \sim H^{30+6,10,15}_1(\tau)  \sim \chi_0(\tau)
$$
where $H^{30+6,10,15}_1(\tau)$ is the mock modular form that is a component of the corresponding optimal mock Jacobi form \cite{CD1605}, and is one of Ramanujan's order 5 mock theta function, $\chi_0$, up to a pre-factor $2q^{-1/120}$.  A complete analysis for $M_3=\pm\Sigma(2,3,5)$  can be found in \S\ref{subsec:habiroexamples}. 
\vspace{5pt}

}

{
\begin{con}\label{sonj:23torus}
Consider the orientation-reversed Brieskorn spheres $-\Sigma(2,3,6\pm 1)$. The defect  $\hat Z$-invariants are given by 
\be\label{dfn:inv23surgery}
\hat{Z}_{b_0}(-\Sigma(2,3,6\pm 1);W_{\bnu};\tau) \sim  {1\over \eta(\tau)}  \sum_{\substack{\hat\epsilon=(\epsilon_1,\epsilon_2,\epsilon_3)\\\in (\ZZ/2)^{3}}}
(-1)^{\hat\epsilon} \tilde\Theta_{A,\asub,\ssb,\ssc_1,\ssc_2}(\tau)
\ee
with 
\begin{gather}\label{23r:parameters}\begin{split}
A= \left( \begin{matrix} -2m & 0 \\ 0& 3 \end{matrix} \right), ~ \ssc_{1} = \left( \begin{matrix}
             1 \\
              0  \end{matrix} \right),~\ssc_{2}=\left( \begin{matrix}
              3 \\
               2(6\pm 1) \end{matrix} \right), ~ \ssb=\left( \begin{matrix}
              0 \\
              \tfrac{1}{6}  \end{matrix}  \right), ~\asub^T=\left( {\rsub\over 2m},-\frac{1}{6}\right)
\end{split}
\end{gather}
where we have used the same notation as in Theorem \ref{thm:modularityBrieskorn} with the triplet given by $(p_1,p_2,p_3)=(2,3,6\pm 1)$. 
\end{con}
}

Moreover, when defect invariants are taken into account, one can show that they are indeed components of vector-valued mock modular forms for ${\rm SL}_2(\ZZ)$ with shadows that are given by the unary theta function, in accordance with the false-mock conjecture and providing evidence for Conjecture \ref{con:modularity} and \ref{conj:Jac}.

{
\renewcommand{\thethm}{2-1}
 \begin{thm}
     \label{thm:mod_23surgery}
The conjectural expression 
\eqref{dfn:inv23surgery} for the $\hat Z$-invariants has the property that 
{${\hat{Z}_{b_0}(-\Sigma(2,3,6\pm 1);W_{\bnu};\tau)}\sim f_{m,r_{\bnu}}^{m+K}$}, where   
$f_{m,r_{\bnu}}^{m+K}$ is a mock modular form with shadow given by $\theta^{1,m+K}_{r_{\nu}}$, where $r_{\bnu}=m-\sum_i  (1+\nu_i)\bar p_i $, in the same notation as in Theorem \ref{thm:modularityBrieskorn}. 
\end{thm}}
{
\begin{proof}
We refer to Appendix \ref{app:mock-proof} for the relevant calculation.
\end{proof}
}

{
One might wonder about the origin and the interpretation of the conjectural expression in Conjecture~\ref{sonj:23torus}, in particular the choice of the parameters \eqref{23r:parameters}. 
Though we do not offer a complete answer to this interesting and important question here, we hope to make these choices more transparent with an analysis of the structure of the sum $\tilde\Theta_{A,\asub,\ssb,\ssc_1,\ssc_2}$, how it reproduces the lattice sum $\tilde\chi_{\hat \epsilon}^{\rm reg}(-\tau)$ (cf. \eqref{dfn:chi_minus}), and how the desired shadow (cf. Theorem \ref{thm:mod_23surgery}) is captured by the lattice sum \eqref{prop:8-33HarmonicMaassForms}. We see that, in the case of $-\Sigma(2,3,6\pm 1)$, our proposed lattice sum captures the 1-dimensional  lattice sum in the negative direction (arising from changing $q\to q^{-1}$ in the false theta function) as a sum along $\ssc_1$, and the 1-dimensional  lattice sum in the positive direction (arising from the Dedeking eta function)  as a sum along $\ssc_1^\perp$. The bilinear form $A$ is chosen such that the norms of the 1-dimensional  lattices are reproduced. The lack of active partcipation of $\ssc_2$, giving the other boundary of the cone, is reflected in the fact that the $\ssc_2$ contribution to the shadow conspire to vanish (see Lemma \ref{lem:c2_contribution}).
}

This is however not the only possibility. 
For instance, one can instead have   $\ssc_1$ and $\ssc_2$ that related by a symmetry, 
and realize the
 1-dimensional  lattice sum in the negative direction as a sum along $\ssc_1$ and along $\ssc_2$, and the 1-dimensional  lattice sum in the positive direction as a sum along $\ssc_1^\perp$ and $\ssc_2^\perp$. Also the shadow will now receive contributoin from both vectors in a symmetric manner. In this way, we obtain the following conjecture. See also \S\ref{sec:mockfalseIHSM24p8} for an analysis of this case from a surgery point of view. 
{
\begin{con}\label{sonj:H24}
Consider the Seifert manifold $M_{3}=-M\left( -1; \frac{1}{2}, \frac{1}{3}, \frac{1}{8} \right)$. The defect  $\hat Z$-invariants are given by 
\be\label{dfn:invH24_1}
\hat{Z}_{b_0}(M_{3};W_{\nu};\tau) \sim \, {1\over 2 \eta(\tau)} 
 \tilde\Theta_{A,\sa_{e,\nu},\ssb_e,\ssc_1,\ssc_2}(\tau), ~~ \nu=1,3
\ee
and 
\be\label{dfn:invH24_2}
\hat{Z}_{b_\ell}(M_{3};W_{\nu};\tau)\sim \, {(-1)^{\nu/2}\over2 \eta(\tau)} \left(
 \tilde\Theta_{A,\sa_{o,\nu},\ssb_{o,1},\ssc_1,\ssc_2}(\tau)+ (-1)^{\ell}\tilde\Theta_{A,\sa_{o,\nu},\ssb_{o,2},\ssc_1,\ssc_2}(\tau)\right) ,~~  \nu=0,2, \ell=0,1. 
\ee
The definition of the indefinite lattice regulariztion  is given by  
\begin{gather}\label{248_parameters1}\begin{split}
A= \left( \begin{matrix} -1 & 0 \\ 0& 4 \end{matrix} \right), ~ \ssc_{i} = \left( \begin{matrix}
             4 \\
              \epsilon_i  \end{matrix} \right), ~   \epsilon_i= (-1)^{i+1} , i=1,2,
\end{split}
\end{gather}
and the relevant lattice vectors are
\begin{gather}\label{248_parameters2}
\begin{split}
&\ssb_e = \begin{pmatrix} 0 \\ \tfrac{1}{8} \end{pmatrix}, 
\ssb_{o,1} = \begin{pmatrix} -\tfrac{1}{2} \\ -\tfrac{1}{4} \end{pmatrix}, \quad
\ssb_{o,2} = \begin{pmatrix} 0 \\ -\tfrac{1}{8} \end{pmatrix}, \\
&\sa_{e,1} = \begin{pmatrix} \tfrac{1}{2} \\ \tfrac{1}{4} \end{pmatrix}, \quad
\sa_{e,3} = \begin{pmatrix} \tfrac{3}{2} \\ \tfrac{1}{2} \end{pmatrix}, \quad
\sa_{o,0} = \begin{pmatrix} 0 \\ \tfrac{1}{8} \end{pmatrix}, \quad
\sa_{o,2} = \begin{pmatrix} 0 \\ \tfrac{3}{8} \end{pmatrix}.
\end{split}
\end{gather}
See \S5.3.2 for the detailed description of the defects involved. 
\end{con}}

{
\renewcommand{\thethm}{2-2}
 \begin{thm}\label{thm:mod_248}
The conjectural expression 
\eqref{dfn:invH24_1}-\eqref{dfn:invH24_2} for the $\hat Z$-invariants has the property that 
{${\hat{Z}_{b_\ell}(-M_3);W_{\bnu};\tau)}\sim f_{r_{\ell,\nu}}^{m+K}$}, where 
$f_{r_{\ell,\nu}}^{m+K}$ is a mock modular form with shadow given by $\theta^{1,m+K}_{r_{\ell,\nu}}$, where $m=24$ and $K=\{1,8\}$ and $r_{\ell,\nu}\in\{1,2,5,7,8,13\}$. 
\end{thm}}

{
\begin{proof}
The mock modular properties, in particular the shadow, can be derived via a straightforward application of \eqref{prop:8-33HarmonicMaassForms} to the indefinite theta functions in  \eqref{dfn:invH24_1}-\eqref{dfn:invH24_2}. Moreover, after verifying the mock modular properties, by comparing the leading coefficients in the $q$-expansion one can prove that \eqref{dfn:invH24_1}-\eqref{dfn:invH24_2} coincide with the components of one of the optimal mock Jacobi forms classified in \cite{CD1605}: 
\begin{gather}\begin{split}
& \hat{Z}_{b_\ell}(M_{3};W_{\nu};\tau)\sim{1\over 2 \eta(\tau)} 
 \tilde\Theta_{A,\sa_{e,\nu},\ssb_e,\ssc_1,\ssc_2}(\tau)
= \begin{cases}
     H_2^{24+8}  ,~&\nu=1, \ell=0 \\ 
       H_8^{24+8}(\tau) ,~&\nu=3, \ell=0
 \end{cases}\\&\\ &\hat{Z}_{b_\ell}(M_{3};W_{\nu};\tau)\sim\\
& {-1\over2 \eta(\tau)} \left(
\tilde\Theta_{A,\sa_{o,\nu},\ssb_{o,1},\ssc_1,\ssc_2}(\tau)+ (-1)^{\ell}\tilde\Theta_{A,\sa_{o,\nu},\ssb_{o,2},\ssc_1,\ssc_2}(\tau)\right)
 = \begin{cases}
     H_1^{24+8}(\tau) ,~&\nu=0, \ell=0 \\ 
       H_7^{24+8}(\tau) ,~&\nu=0,\ell=1 \\
          H_5^{24+8}(\tau) ,~&\nu=2, \ell=0 \\ 
       H_{13}^{24+8}(\tau) ,~&\nu=2,\ell=1 \\
 \end{cases} \\ 
 \end{split}
 \end{gather}
\end{proof}
}
\renewcommand{\thethm}{\arabic{thm}}

{
\begin{cor}\label{cor:mock_mod}
The mock modular property of the functions in \eqref{dfn:inv23surgery} and \eqref{dfn:invH24_1}-\eqref{dfn:invH24_2} is compatible with Conjecture \ref{con:modularity}.2. 
\end{cor}
\begin{proof}
The statement follows from Theorem \ref{thm:mod_23surgery} and Theorem \ref{thm:mod_248}, in combination with the proof of Theorem \ref{thm:modularityBrieskorn} and the calculation in \S \ref{sec:24+8-False-Surgery}. 
\end{proof}
}

\subsection{Relation to Surgeries Along Torus Knots}
\label{subsec:relation}

In this subsection we will explain the relation between the conjectural surgery formula by Park (Conjecture 5, \cite{Park2106}) and regularized indefinite theta functions, and comment on their mock modular properties. 

Consider the torus knot $T(s,t)$, for instance the right-handed trefoil $\mathbf{3}^{r}_{1}=T(2,3)$ and its mirror $\mathbf{3}^{l}_{1}=T(2,-3)$.
First, recall that,  a $-1/r$-surgery \eqref{eq:surgeryS2} along the torus knot $T(s,t)$ leads to the Brieskorn sphere $S^3_{-1/r}(T(s,t))=
\Sigma(s,t,str+1)$.
Similarly, considering the mirror of the torus knot $m(T(s,t))=T(s,-t)$, its  $+1/r$-surgery gives the orientation-reversed manifold 
\be
S^3_{+1/r}(T(s,-t))=- 
\Sigma(s,t,str+1). 
\ee
From 
\be\label{FKTorus}
F_{T(s,t)}(x,q)=-q^{\frac{(s-1)(t-1)}{2}} \sum_{k>0}\varepsilon_k q^{{k^2-(st-s-t)^2 \over 4st}}(x^{k/2}-x^{-k/2})
\ee
where 
\(\varepsilon_k=\sum_{\epsilon\in\ZZ/2}(-1)^\epsilon  (\delta_{k+ (st+(-1)^\epsilon s+ t) \Mod{2st}}+\delta_{k- (st+(-1)^\epsilon s+ t) \Mod{2st}})\), 
for the torus knot $K=T(s,t)$, we obtain via $F_{m(K)}(x,q)=F_{K}(x,q^{-1})$ the two-variable series 
for the negative torus knot $m(K)=T(s,-t)$. 
Now, applying Conjecture 5 of \cite{Park2106} to this case leads to the prediction that $\hat Z_0^{{\rm Reg. Surg.}}(M_3) =\hat Z_0(M_3)$, with
\begin{gather}\begin{split}
&\hat Z_0^{{\rm Reg. Surg.}}(S^{3}_{1/r}(T(s,-t));\tau) = \hat Z_0^{{\rm Reg. Surg.}}(-\Sigma(s,t,str+1);\tau)   \\
&\sim \sum_{k\geq 0} \, \varepsilon_k \,q^{-{k^2-(st-s-t)^2 \over 4st}} \sum_{\epsilon_2\in \ZZ/2 }(-1)^{\epsilon_2}q^{-{r\over 4} (k-(-1)^{\epsilon_2}{1\over r})^2} \left(\frac{\sum_{\substack{\ell\in \ZZ \\ |\ell|<{k-1\over 2}}} (-1)^{\ell} q^{{2r+1\over 2}(\ell+ {1\over 2(2r+1)})^2}}{\sum_{\ell\in \ZZ} (-1)^{\ell} q^{{2r+1\over 2}(\ell+ {1\over 2(2r+1)})^2}}
\right) . 
\end{split}\end{gather}

By carefully analysing the region $|\ell|<{k-1\over 2}$ one can show that the above expression is equivalent to the following
\be\label{dfn:inv23surgeryR}
\hat{Z}^{{\rm Reg. Surg.}}_{b_0}(-\Sigma(s,t,str+1);W_{\bnu};\tau) \sim \, {1\over f_{2r+1,1}(\tau)}  \sum_{\substack{\hat\epsilon=(\epsilon_1,\epsilon_2,\epsilon_3)\\\in (\ZZ/2)^{3}}}
(-1)^{\hat\epsilon} \Theta_{A,\asub,\ssb,\ssc_1,\ssc_2}(\tau)
\ee
with the 1-dimensional theta function
\begin{equation}\label{eq:theta-function}
  f_{x,\chi}(\tau) :=\sum_{k\in\mathbb{Z}} (-1)^k q^{\frac{x}{2}\left(k-\frac{\chi}{2x} \right)^2} ~
\end{equation}
and with the input data  
\begin{gather}\label{23r:parametersR}\begin{split}
A= \left( \begin{matrix} -2m & 0 \\ 0& x \end{matrix} \right), ~ \ssc_{1} = \left( \begin{matrix}
             1 \\
              0  \end{matrix} \right),~\ssc_{2}=\left( \begin{matrix}
              x \\
               2p_3 \end{matrix} \right), ~ \ssb=\left( \begin{matrix}
              0 \\
              \tfrac{1}{2x}  \end{matrix}  \right), ~\asub^T=\left( {\rsub\over 2m},-\frac{\chi}{2x}\right)~,~x=2r+1
\end{split}
\end{gather}
 and we have used the same notation as in the proof of Theorem \ref{thm:modularityBrieskorn} with the triplet given by $(p_1,p_2,p_3)=(s,t,str+1)$. In \eqref{dfn:inv23surgeryR}, we have generalized the analysis to the case with Wilson lines.

{
After establishing their relation to regularised indefinite theta functions, we can now easily analyse the mock modular properties of $\hat{Z}^{{\rm Reg. Surg.}}$. 
First, note that it is given by a mixed mock modular form $\Theta_{A,\asub,\ssb,\ssc_1,\ssc_2}$ divided by a component of a $r$-dimensional vector-valued modular form for $SL_2(\mathbb Z)$. We therefore conclude that $\hat{Z}^{{\rm Reg. Surg.}}$ is a mixed mock modular form only when $r=1$, in which case $f_{2r+1,1}=\eta$ is the Dedekind eta function. On the other hand, as explicit calculations in Appendix \ref{app:mock-proof} shows, the sum of products of a  holomorphic and an anti-holomorphic functions in the  expression for the shadow of the regularized indefinite theta function \eqref{prop:8-33HarmonicMaassForms} simplifies into a product a  holomorphic and an anti-holomorphic functions only in the case $s=2, t=3$. The above analysis leads us to the conclusion that in this family, only $\hat{Z}^{{\rm Reg. Surg.}}_{b_0}(-\Sigma(2,3,7);W_{\bnu};\tau)$ are mock modular forms, and indeed coincide with the expressions in Conjecture \ref{sonj:23torus} and coincide with order 7 mock theta functions of Ramanujan. 
Nevertheless, we note that other $\hat{Z}^{{\rm Reg. Surg.}}_{b_0}(-\Sigma(s,t,st+1);W_{\bnu};\tau)$ are mixed mock modular forms which moreover also possess the desirable quantum modular property \eqref{eqn:leaking}, since the ratio between the additional holomorphic factors in the mixed mock shadow \eqref{prop:8-33HarmonicMaassForms} and the denominator $\eta(\tau)$ either vanishes or diverges at the cusps. 
}

\subsection{Going to the Other Side Using Appell-Lerch Sums}
In this last subsection, we take a brief detour to discuss a uniform and canonical way to associate a mock modular form to a false theta function $\tilde\theta_{m,r}$, for all pairs $(m,r)$. 
This should highlight how tricky the problem of going to the other side is, and in particular that mock modularity alone is not  sufficient to fix a unique answer. 

Given a false theta function $\tilde\theta_{m,r}(\tau)$ \eqref{eq:intro:falsetheta}, one can canonically associate to it a partner $\tilde\theta^-_{m,r}(\tau)$ as a function on $\HH$, via the Appell-Lerch sums. Moreover, one can show that 
$\tilde\theta^-_{m,r}$ is a mock modular form and is related to the characters of a vertex algebra \cite{cheng2022cone}, constructed as a cone algebra reviewed in \S\ref{sec:VA}. 
Recall that for Seifert manifolds $M_3$ with three singular fibres, those which are weakly negative have $\hat Z(M_3;W_{\bnu})$ invariants that are, up to an overall rational power of $q$ and possibly the addition of a finite $q$-polynomial, linear combinations of $\tilde\theta_{m,r}$ for some fixed $m$ \cite{plumbing}. 
From the false-mock Conjecture \ref{con:False-Mockv2} and from the above, it might seem that one can now easily compute the corresponding $\hat Z(W_{\bnu})$-invariants for the orientation-reversed manifold $-M_3$ by simply replacing each $\tilde\theta_{m,r}$ with $\tilde\theta^-_{m,r}$ (and flipping $q\leftrightarrow q^{-1}$ in the overall $q$-power and the finite polynomial). 
However, it is rather easy to see that this canonical and uniform method is incompatible with many of the concrete examples computed using various methods in this paper and in \cite{GM,Park2106}. 

To explain it, recall that by going through the relation to Fine's $q$-hypergeometric series and the universal mock theta function one arrives at the following partner of the false theta function $\tilde\theta_{m,r}(\tau)$ for any positive integer $m$ and $0<r<2m$
\cite{FolsomQuanJac}: 
\begin{gather}\begin{split}
 & \tilde \theta^-_{m,r}\left( \tau \right) = - \frac{\eta\left( 2m \tau \right)}{\eta^{2}\left( 4m\tau \right)}q^{\frac{m}{4}}\\&\times\left[ q^{-\frac{r^{2}}{4m}}A_{2}\left( \left( r-m \right)\tau,-2m\tau,4m\tau \right) - q^{\frac{\left( r-2m \right)^{2}}{4m}}A_{2}\left( \left( m-r \right)\tau,-2m\tau,4m\tau \right) \right],
\end{split}\end{gather}
where 
the Appell-Lerch sum is given by 
\begin{align}
  A_{2}\left( a\tau, -2m\tau; 4m\tau \right) 
  = q^{a}\sum_{n\in\mathbb{Z}}\frac{q^{4mn^{2}+2mn}}{1-q^{a+4mn}}, ~a\in \QQ
\end{align}
and is shown to have mock modular properties \cite{2008arXiv0807.4834Z}.  
A quick computation shows that 
\begin{equation}\label{polar_part}
  \tilde \theta^-_{m,r}\left( -\tau \right) = q^{-\frac{r^{2}}{4m}}\left( 1+O\left( q \right) \right)~ {\rm for}~ 0<r<m, 
\end{equation}
which can be extended to all $r\in \ZZ$ by the relation $  \tilde \theta^-_{m,r}=  \tilde \theta^-_{m,r+2m}=  -\tilde \theta^-_{m,-r}$. 

From \eqref{polar_part} we conclude that the $  \tilde \theta_{m,r} \mapsto   \tilde \theta^-_{m,r}$ swap does not always reproduce the $\hat Z(-M_3)$ that one believes to be correct. 
For instance, for $M_3=\Sigma(2,3,7)$, the $\hat Z(M_3)$-invariant is given by 
\begin{equation}
\hat Z(M_3;\tau)= q^{\frac{83}{168}}\sum_{0<r<42} \left(\frac{r}{21}\right) \tilde \theta_{42,r}(\tau) 
= q^{\frac{83}{168}}\left(\tilde \theta_{42,1}+\tilde \theta_{42,41}-\tilde \theta_{42,29}-\tilde \theta_{42,13}\right) (\tau) . 
\end{equation}
 The flipped invariant $\hat Z(-M_3)$, on the other hand, is expected to be given by the celebrated order 7 mock theta function of Ramanujan \cite{3d,GM,CFS1912}: 
\begin{gather}\begin{split}
q^{-\frac{1}{168}}  F_0(\tau)&
=q^{-\frac{1}{168}}  \sum_{n=0}^\infty \frac{q^{n^2}}{(q^{n+1};q)_n}\\
&= -q^{-\frac{1}{168}}\left( 1 + q + q^3 + q^4 + q^5 + 2 q^7 + q^8 + 2 q^9 + q^{10} + 2 q^{11}  + \mathcal{O}(q^{12})\right) ~.
\end{split}\end{gather}
It is easy to see that this is not the same as $\left(\tilde \theta^-_{42,1}+\tilde \theta^-_{42,41}-\tilde \theta^-_{42,29}-\tilde \theta^-_{42,13}\right)$, despite having the same mock modularity and hence essentially the same radial limit\footnote{Meaning the same radial limit after discarding the exponential singularities \eqref{eqn:leaking}.}, just from the leading behaviour near $\tau\to i\infty$ using \eqref{polar_part}.

\section{Mock Invariants and Vertex Algebras}
\label{sec:VA}

Vertex operator algebras have been shown to be related to 3d supersymmetric conformal field theories, in particular to the $\hat{Z}$-invariants, in rich and deep ways.
In the context of negative Seifert manifolds with three or four exceptional fibres the $\hat{Z}$-invariant has been shown \cite{3d,cheng20223} to be given by a linear combination of characters of a class of logarithmic vertex operator algebras.
The insertion of defects or a change of the choice of the ${\rm Spin}^c$-structure, while modifying the modules of the vertex algebra computed, does not change the underlying algebra. {This observation is compatible with the expectation that the underlying VOA is associated with the underlying three-dimensional quantum field theory, and not to a specific boundary condition. }
In these cases, this has led to a construction of a vertex algebra from a given three-dimensional manifold.

In this section we extend the relation between  $\hat{Z}$-invariants and vertex operator algebras to the ``other side" of mock invariants, by systematically associating vertex operator algebras to the mock $\hat{Z}$ invariants that have been given expressions in terms of  indefinite theta function. 
This direct connection to vertex algebras, together with the manifest mock modular property, is one of the most important advantages of having an indefinite theta function expression for $\hat{Z}$-invariants for non-weakly-negative three-manifolds.

\subsection{Cone Vertex Algebras}

It is well-known that from an integral lattice $L$ one can construct a Lie algebra and the Lie algebra module $V_L$, which has the structure of vertex algebra and is moreover a vertex operator algebra when $L$ is positive definite. 
We will now briefly describe how one can associate a vertex operator algebra to a positive-definite cone in a lattice of indefinite signature. We will skip many details and simply refer to \cite{cheng2022cone} (\S2, \S3) for a more detailed account of these cone algebras, mostly based on the treatment of \cite{duncan2017umbral}. 
For a cone $C\subset L$ with $0\in C$ that is closed under addition, the submodule $V_C$ of $V_L$ generated by elements corresponding to $\lambda\in C$
has the structure of a sub-vertex algebra of $V_L$, with the same conformal element. Furthermore, $V_C$ has the structure of vertex operator algebra when the lattice bilinear form is positive definite when restricted to $C$.   
Given $\gamma\in L\otimes_\ZZ \QQ$, for any $C'\subset L+\gamma$ such that $C'+C\subset C'$, the corresponding $V_{C'}$ can be endowed with the structure of a twisted module over $V_C$. 

Suppose that  the cone 
\begin{equation}\label{dfn:cone}
  P := \left\{ \sum_{i=1}^{\operatorname{rk}(L)} a_{i}\ssd_{i}\in L | a_{i}\ge 0,\ \forall\ i =1,\dots, \operatorname{rk}(L)\right\}
\end{equation}
for some $\ssd_{i}\in L$ is such a positive-definite subset of $L$. Then $V_P$ has a vertex operator algebra structure. 
Given a generating set $\{\boldsymbol{\varepsilon}_1,\dots, \boldsymbol{\varepsilon}_{{\operatorname{rk}(L)}}\}$  of $L$ and $\sa^+=\sum_i a_i \boldsymbol{\varepsilon}_i$, $\sa=(a_1,\dots,  a_{\operatorname{rk}(L)}) \in \QQ^{\operatorname{rk}(L)}$, define  $\sa^- =-\sa^+ + \sum_i \boldsymbol{\varepsilon}_i$ and note that $C'+P\subset C'$ when $C'=P+\sa^\pm$. 
The twisted module of $V_P$ that will be the main object of interest here is 
\be\label{dfn:voam}
V_\sa := V_{\sa^++P} \oplus V_{\sa^-+P} .
\ee
As usual, this twisted module is simply a module when $\sa^+, \sa^- \in L^\ast$ are in the dual lattice.

\subsection{Mock VOA Characters}

Next we will establish the relation between certain trace functions of the modules \eqref{dfn:voam} and the indefinite theta series of the kind encountered in the last section as proposals for $\hat{Z}$-invariants. For the sake of concreteness we will now focus on rank two lattices with signature $(1,1)$. Generalization to lattices with signature $(1,d)$ is straightforward. 
Given $\ssb = \sum_{i=1}^2 b_i \ssd_i$ with $b_i\in \QQ$,  consider the following automorphism $g_\ssb$ of the module $V_\sa:$  it acts as a multiplication by the phase ${\ex}(B(\ssb,{\boldsymbol \lambda-\sa^+}))$ on elements corresponding to ${\boldsymbol \lambda} \in P+\sa^+$, and multiplication by the phase $-{\ex}(B(\ssb,{-\boldsymbol \lambda-\sa^+}))$ on elements corresponding to ${\boldsymbol \lambda} \in P+\sa^-$.

The main object we will be interested in is the trace function
\be\label{dfn:trace}
T_{\sa,\ssb}(\tau) := {\rm Tr}_{V_\sa}\left( g_\ssb q^{L(0)-c/24}\right)
\ee
which is a well-defined function on $\tau \in \HH$ since the lattice bilinear form is positive definite when restricted to the cone $P$. 

Next we will demonstrate the relation between $T_{\sa,\ssb}$ and the indefinite theta functions seen in the previous section, generalizing Lemma 3.1 and Theorem 3.2 of  \cite{cheng2022cone}.

\begin{thm}
Consider the lattice bilinear $A= {\rm diag}(-2p\bar p, x)$ with $x,p,\bar p\in \NN$ satisfying $\bar p x >2p$, and the positive-definite cone \eqref{dfn:cone} with 
\be 
\ssd_1 =\ssc_{1,\perp} =\left( \begin{matrix}
              0 \\
              1  \end{matrix} \right), ~\ssd_2 =\ssc_{2,\perp} = \left( \begin{matrix}
              1 \\
              \bar p  \end{matrix} \right)\ . 
\ee
Consider $\tilde \sa =  \sum_{i}\tilde a_i \ssd_i $ with $~0<\tilde a_i<1$. The trace function of the corresponding twisted module $V_{\sa}$ of the cone algebra is given in terms of the indefinite theta function by 
\be 
T_{\tilde \sa, \ssb}(\tau) = -{1\over \eta^2(\tau)}  \tilde\Theta_{A,\tilde\sa,\ssb,\ssc_1 ,\ssc_2 }(\tau) 
\ee
where 
\be 
\ssc_1 =\left( \begin{matrix}
              1\\
             0  \end{matrix} \right) , ~\ssc_2 =\left( \begin{matrix}
              x\\
             2p  \end{matrix} \right) , ~ 
\ee
and $\sa$ is any vector satisfying $\sa =\tilde \sa  ~{\rm mod}~\ZZ\ssd_1+\ZZ\ssd_2$.

\end{thm}

\begin{proof}
Using the definition \eqref{dfn:trace}, one obtains 
\be\begin{split}\notag
&T_{\tilde \sa, \ssb}(\tau) = {1\over \eta^2(\tau)}\\&\times  \sum_{n_1,n_2\geq 0} \left( {\ex}(B(\ssb, \sum_i n_i \ssd_i))   q^{|\sum_i (n_i+\tilde a_i)\ssd_i|^2/2} - {\ex}(-B(\ssb, \sum_i (n_i+1) \ssd_i))  q^{|\sum_i (n_i+1-\tilde a_i)\ssd_i|^2/2}\right), \end{split}
\ee
while plugging in the vector $\ssc_i$ into \eqref{def:regularisedindefinitetheta} leads to
\begin{gather}\notag\begin{split}
&-  
\Theta_{A,\sa,\ssb,\ssc_1 ,\ssc_2 }(\tau) ={\ex}(B(\ssb, \tilde \sa)) \\& \times \left( \sum_{n_1,n_2\geq 0} {\ex}(B(\ssb, \sum_i n_i \ssd_i))  q^{|\sum_i (n_i+\tilde a_i)\ssd_i|^2/2}- \sum_{n_1,n_2 >  0}  {\ex}(-B(\ssb, \sum_i n_i \ssd_i))   q^{|\sum_i (-n_i+\tilde a_i)\ssd_i|^2/2}\right) .
\end{split}
\end{gather}

\end{proof}

The above establishes the relation between VOA characters and the mock invariants $\hat Z_{b_0}(M_3)$ for $M_3=-\Sigma(2,3,6\pm1)$, via Conjecture \ref{sonj:23torus}. A similar treatment, using the indefinite lattice in Conjecture \ref{sonj:H24}, associates  VOA characters to the mock invariants $\hat Z_{b_0}(M_3)$ for $M_3=-M\left( -1; \frac{1}{2}, \frac{1}{3}, \frac{1}{8} \right)$.

\section{Examples}
\label{sec:ex}

In this section, we will demonstrate the various phenomena discussed in the previous sections, either in the form of conjectures or proven theorems, with various examples of different origins. We will continue to work with teh equivalence \eqref{dfn:equivalence} and 
 not always keep track of pre-factors $cq^{\Delta}$ for some $c\in \CC$ and $\Delta\in \QQ$ the possible finite polynomials $p(\tau)$.  
The data provided in the paper are sufficient for the readers to compute these explicitly in all cases. 

\subsection{{False-\texorpdfstring{$\vartheta$}{theta}} Invariants from Plumbing Graphs}
In this subsection, we use the Definition \ref{conj:plumbed_defect} to compute the examples of $\hat Z$-invariants for plumbed manifolds, with and without defects. 
Using these concrete  $\hat Z$-invariants, we provide evidence for the Modularity Conjecture \ref{con:modularity}, for manifolds that are not Brieskorn spheres, for which a general proof is not available (cf. Theorem \ref{thm:modularityBrieskorn}). 

So far we have not found a way to compute the relevant $\widetilde {\rm SL}_2(\ZZ)$ representation $\Theta^{m+K}$ for generic  Seifert manifolds with three exceptional fibres without computing the $\hat Z$-invariants themselves, despite systematic results for the Brieskorn spheres. 
This said, we would like to extensively test the Modularity Conjecture \ref{con:modularity} beyond the cases in Theorem \ref{thm:modularityBrieskorn}. 
For a given negative Seifert manifold with three exceptional fibres, we will do this in three steps.  First we compute the $\hat Z$-invariants without defects, for all admissible {${\rm Spin}^c$} structures, using the original plumbing proposal  \eqref{weak_neg_def}. Second, from those we identify the relevant $\widetilde {\rm SL}_2(\ZZ)$ representation $\Theta^{m+K}$, using the notation as in \eqref{linear_weil}. 
We will also write $\sigma^{m+K}$ to denote the set of independent components of  $\Theta^{m+K}$,  in the sense detailed in \eqref{linear_weil}. 
Finally we identify the examples of Wilson line insertions that give rise to the specific components of the corresponding vector-valued quantum modular forms. In all examples in this subsection, invariants will be computed using Definition \ref{def:ZhatintegralWilson}.
 
For concreteness, in all examples in this section we represent one of the ${\rm Spin}^c$ structures as
\begin{equation}\label{dfn:b0}
  b_0 = \left\{ b_{0,v} \right\}_{v\in V},\ b_{0,v} = \text{deg}\left( v \right) -2. 
\end{equation}

\subsubsection{\texorpdfstring{$M_{3}=\Sigma\left(2,3,7\right)$, $m=42$, $K=\{1,6,14,21\}$}{}}
\label{subsubsec:237}
Although this Brieskorn sphere case is covered in Theorem \ref{thm:modularityBrieskorn}, we will still give it as an example to illustrate the theorem. This will also facilitate later comparison with other approaches to studying this particular manifold. 

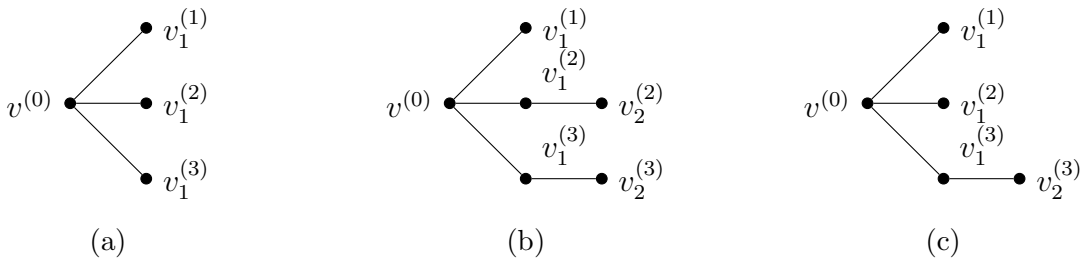
\begin{figure}
\begin{subfigure}{0.33\textwidth}
\centering
\begin{tikzpicture}
  \tikzstyle{filledcircle}=[circle, draw, fill=black, inner sep=0pt, minimum size=4pt]

  \node[filledcircle,label=left:$v^{(0)}$] (a0) at (0,0) {};
  \node[filledcircle,label=right:$v^{(1)}_1$] (a1) at (1,1) {};
  \node[filledcircle,label=right:$v^{(2)}_1$] (a2) at (1,0) {};
  \node[filledcircle,label=right:$v^{(3)}_1$] (a3) at (1,-1) {};

  \draw (a0) -- (a1);
  \draw (a0) -- (a2);
  \draw (a0) -- (a3);
\end{tikzpicture}
\caption{}\label{fig:plumbing_graph_1_v2}
\end{subfigure}
\begin{subfigure}{0.3\textwidth}
\centering
\begin{tikzpicture}
  \tikzstyle{filledcircle}=[circle, draw, fill=black, inner sep=0pt, minimum size=4pt]

  \node[filledcircle,label=left:$v^{(0)}$] (a0) at (0,0) {};
  \node[filledcircle,label=30:$v^{(3)}_{1}$] (a11) at (1,-1) {};
  \node[filledcircle,label=right:$v^{(3)}_{2}$] (a1) at (2,-1) {};
  \node[filledcircle,label=30:$v^{(2)}_{1}$] (a21) at (1,0) {};
  \node[filledcircle,label=right:$v^{(2)}_{2}$] (a2) at (2,0) {};
  \node[filledcircle,label=right:$v^{(1)}_1$] (a3) at (1,1) {};

  \draw (a0) -- (a11);
  \draw (a1) -- (a11);
  \draw (a0) -- (a21);
  \draw (a2) -- (a21);
  \draw (a0) -- (a3);
\end{tikzpicture}
\caption{}\label{fig:plumbing_graph_2_v2}
\end{subfigure}
\begin{subfigure}{0.33\textwidth}
\centering
\begin{tikzpicture}
  \tikzstyle{filledcircle}=[circle, draw, fill=black, inner sep=0pt, minimum size=4pt]

  \node[filledcircle,label=left:$v^{(0)}$] (a0) at (0,0) {};
  \node[filledcircle,label=right:$v^{(1)}_1$] (a1) at (1,1) {};
  \node[filledcircle,label=right:$v^{(2)}_1$] (a2) at (1,0) {};
  \node[filledcircle,label=right:$v^{(3)}_2$] (a3) at (2,-1) {};
  \node[filledcircle,label=30:$v^{(3)}_{1}$] (a31) at (1,-1) {};

  \draw (a0) -- (a1);
  \draw (a0) -- (a2);
  \draw (a0) -- (a31);
  \draw (a3) -- (a31);
\end{tikzpicture}
\caption{}\label{fig:plumbing_graph_3_v2}
\end{subfigure}
\caption{Plumbing diagrams}
\end{figure}
The plumbing graph for $\Sigma(2,3,7)$ is depicted in Figure \ref{fig:plumbing_graph_1_v2} with the weights $\alpha_j^{(i)}$ on the node labelled by $v_j^{(i)}$ given by $(\alpha^{(0)},\alpha^{(1)}_1,\alpha^{(2)}_1,\alpha^{(3)}_1)=(-1,-2,-3,-7)$ (cf. \eqref{eq:cont'edfrac}).

Since the three-manifold is a homological sphere, we have only one ${\rm{Spin}}^c$-structure $b_0$, and we have
\be
  \label{eq:allfalse2371}
\hat{Z}_{b_0}\left( \Sigma\left( 2,3,7 \right); \tau \right)\sim \tilde{\theta}_{1}^{m+K}. 
\ee
The Weil representation $\Theta^{m+K}$ is a three-dimensional irreducible representation with
 $\sigma^{m+K}=\{1,5,11\}$. 
 
The other two components of the vector-valued quantum modular form
\[ \begin{pmatrix}\tilde{\theta}^{m+K}_{1}\\\tilde{\theta}^{m+K}_{5}\\\tilde{\theta}^{m+K}_{11}
\end{pmatrix}
 (\tau)
\]
show up as the following defect invariants 
\begin{align}
  \label{eq:allfalse237}
  \begin{split}
  &\hat{Z}_{b_0}\left( \Sigma\left( 2,3,7 \right); W_{(0,0,1)}; \tau \right)  \sim \tilde{\theta}_{5}^{m+K}(\tau) \\
  &\hat{Z}_{b_0}\left( \Sigma\left( 2,3,7 \right); W_{(0,0,2)}; \tau \right) \sim \tilde{\theta}_{11}^{m+K}(\tau) .
  \end{split}
\end{align}

\subsubsection{\texorpdfstring{$M_{3}=M\left(-2; {1 \over 2},{2 \over 3}, {2 \over 3}\right)$, $m=6$, $K=\{1,3\}$}{}}
This example was featured in \cite{3d} (Table 13).  The corresponding Weil representation $\Theta^{6+3}$ is a two-dimensional irreducible representation  with $\sigma^{m+K}= \left\{ 1,3 \right\}$. 
Its plumbing graph is depicted in Figure \ref{fig:plumbing_graph_2_v2} with the weight assignment $$
(\alpha^{(0)}, \alpha_1^{(1)}, \alpha_1^{(2)}, \alpha_2^{(2)},\alpha_1^{(3)},\alpha_2^{(3)})=(-2,-2,-2,-2,-2,-2).$$

Apart from the trivial ${\rm{Spin}}^c$-structure $b_0$, this manifold permits a second ${\rm{Spin}}^c$-structure  that can be represented by $b_1 = (-1,1,0,1,0,3)$ in the basis defined by the node ordering above.  
This example has the special property that the two $\hat Z_b$ with no defect modifications already span the vector-valued $\widetilde {\rm SL}_2(\ZZ)$ quantum modular forms: 
\begin{align}
  \notag
  \begin{split}
  &\hat{Z}_{b_{0}}\left( M\left(-2; {1 \over 2},{2 \over 3}, {2 \over 3}\right) ; \tau \right)\sim \tilde{\theta}_{1}^{m+K}(\tau) \\
  &\hat{Z}_{b_{1}}\left(  M\left(-2; {1 \over 2},{2 \over 3}, {2 \over 3}\right) ;  \tau \right) \sim
\tilde{\theta}_{3}^{m+K}(\tau)~.
  \end{split}
\end{align}
The inclusion of defects, in this case, gives rise to the same functions up to the equivalence relation. For instance, we have
\begin{align} \notag
  \begin{split}
  &\hat{Z}_{b_{0}}\left(  M\left(-2; {1 \over 2},{2 \over 3}, {2 \over 3}\right) ; W_{(0,0,1)}; \tau \right)  \sim \tilde{\theta}_{3}^{m+K}(\tau)~\\
 &\hat{Z}_{b_{1}}\left( M\left(-2; {1 \over 2},{2 \over 3}, {2 \over 3}\right) ; W_{(0,0,1)}; \tau \right)  \sim \tilde{\theta}_{1}^{m+K} (\tau)
  \end{split}
\end{align}
 consistent with the Modularity Conjecture 
\ref{con:modularity}.

\subsubsection{\texorpdfstring{$M_{3}=M\left(-1; {1 \over 2},{1 \over 7}, {2 \over 7}\right)$, $m=14$, $K=\{1,7\}$}{}}

In \cite{3d} (Table 13) it was found that this plumbed manifold, with plumbing graph  depicted in Figure \ref{fig:plumbing_graph_3_v2} and  the weight assignment $(\alpha^{(0)},\alpha^{(1)}_1,\alpha^{(2)}_1,\alpha^{(3)}_{1},\alpha^{(3)}_2)=(-1,-2,-7,-4,-2)$, corresponds to the four-dimensional irreducible
 Weil representation $\Theta^{14+7}$  which has $\sigma^{m+K}= \left\{1,3,5,7 \right\}$.

It has two inequivalent ${\rm{Spin}}^c$-structures $b_0$ and $b_1=(-1,1,1,0,5)$. The corresponding $\hat Z$-invariants 
 correspond to 
$\{1,7\}\subset \sigma^{m+K}$: 
\begin{align}
\notag
  \begin{split}
  &\hat{Z}_{b_{0}}\left( M\left(-1; {1 \over 2},{1 \over 7}, {2 \over 7}\right) ; \tau \right) \sim \tilde{\theta}_{3}^{m+K}(\tau) \\
  &  \hat{Z}_{b_{1}}\left( M\left(-1; {1 \over 2},{1 \over 7}, {2 \over 7}\right) ; \tau \right) \sim \tilde{\theta}_{7}^{m+K}(\tau)~. 
  \end{split}
\end{align}

The rest of the components in $\sigma^{m+K}=\{ 1,3,5,7 \}$ become visible when one considers the defect invariants, 
consistent with the Modularity Conjecture 
\ref{con:modularity}. 
For instance, we see 
\begin{align}
 \notag 
  \begin{split}
  &\hat{Z}_{b_{0}}\left(  M\left(-1; {1 \over 2},{1 \over 7}, {2 \over 7}\right) ; W_{(0,0,1)}; \tau \right)\sim \tilde{\theta}_{1}^{m+K}(\tau) \\
  & \hat{Z}_{b_{0}}\left(  M\left(-1; {1 \over 2},{1 \over 7}, {2 \over 7}\right) ; W_{(0,0,4)}; \tau \right)  \sim  \tilde{\theta}_{5}^{m+K}(\tau) \\
  &\hat{Z}_{b_{0}}\left(  M\left(-1; {1 \over 2},{1 \over 7}, {2 \over 7}\right) ; W_{(0,0,5)}; \tau \right)  \sim \tilde{\theta}_{7}^{m+K}(\tau) ~,
  \end{split}
\end{align}
where we obtain matching after removing certain $q$-monomials in the latter two cases.

\subsubsection{Other Examples}

In Tables \ref{tab:the_big_one} and \ref{tab:the_smaller_one} we compile an extensive list of Seifert manifolds, for which we record the relevant $\widetilde {\rm SL}_2(\ZZ)$ representation $\Theta^{m+K}$ and its corresponding natural components labelled by the elements of the set $\sigma^{m+K}$. For all these cases, we have verified the Modularity Conjecture \ref{con:modularity}. We use the notation $\sigma^{(M_3)}$ to denote the set of components in  $\Theta^{m+K}$ that are seen in the $\hat Z$-invariants, with and without defects, of the three-manifold in question. The Modularity Conjecture \ref{con:modularity} in particular implies $\sigma^{m+K}=\sigma^{(M_3)}$.  
In Table \ref{tab:the_smaller_one} we list the cases for which the Weil representation is of the form $\Theta^{m+K}$ \eqref{dfn:projB}. 
As explained in \S\ref{sec:Modularity}, sometimes these Weil representations defined through \eqref{dfn:projB} are reducible and one needs an additional projection \eqref{irred_weil} to obtain the corresponding irreducible subrepresentation $\Theta^{m+K,{\rm{irr}}}$. 
Interestingly, these also appear as modular representations for $\hat Z$-invariants for Seifert manifolds with three singular fibres. In Table \ref{tab:the_smaller_one} we record some examples where $\Theta^{m+K,{\rm{irr}}}$ appear.

\begin{center}
\begin{footnotesize}
\begin{longtable}{ccc}
\toprule
 $M_{3}$ & $m+K$ & $\sigma^{m+K}=\sigma^{(M_3)}$\\
\midrule
\endfirsthead
\multicolumn{3}{c}{{Continued from previous page}} \\
\toprule
 $M_{3}$ & $m+K$ & $\sigma^{m+K}=\sigma^{(M_3)}$ \\
\midrule
\endhead
\multicolumn{3}{|c|}{{\tablename\ \thetable{}:  continued on next page}} \\ \hline
\endfoot
\endlastfoot
 $M\left(-2;1/2,1/2,1/2\right)$ & 2 & $\{1\}$\\
 $M\left(-2;1/2,1/2,2/3\right)$ & 3 & $\{1,2\}$\\
 $M\left(-2;1/2,1/2,3/4\right)$ & 4 & $\{1,2,3\}$\\
 $M\left(-2;1/2,1/2,4/5\right)$ & 5 & $\{1\ldots4\}$\\
 $M\left(-2;1/2,1/2,5/6\right)$ & 6 & $\{1\ldots5\}$\\
 $M\left(-2;1/2,2/3,2/3\right)$ & 6+3 & $\{1,3\}$\\
 $M\left(-2;1/2,1/2,6/7\right)$ & 7 & $\{1\ldots6\}$\\
 $M\left(-2;1/2,1/2,7/8\right)$ & 8 & $\{1\ldots7\}$\\
 $M\left(-2;1/2,1/2,8/9\right)$ & 9 & $\{1\ldots8\}$\\
 $M\left(-2;1/2,1/2,9/10\right)$ & 10 & $\{1\ldots9\}$\\
 $M\left(-1;1/2,1/5,1/5\right)$ & 10+5 & $\{1,3,5\}$\\
 $M\left(-4;1/2,1/2,1/2\right)$ & 10+5 & $\{1,3,5\}$\\
 $M\left(-2;1/2,1/2,11/12\right)$ & 12 & $\{1\ldots11\}$\\
 $M\left(-2;1/2,1/2,12/13\right)$ & 13 & $\{1\ldots12\}$\\
 $M\left(-5;1/2,1/2,1/2\right)$ & 14+7 & $\{1,3,5,7\}$\\
 $M\left(-1;1/2,1/7,2/7\right)$ & 14+7 & $\{1,3,5,7\}$\\
 $M\left(-2;1/2,1/2,15/16\right)$ & 16 & $\{1\ldots15\}$\\
 $M\left(-2;1/2,1/2,17/18\right)$ & 18 & $\{1\ldots17\}$\\
 $M\left(-6;1/2,1/2,1/2\right)$ & 18+9 & $\{1,3,5,7,9\}$\\
 $M\left(-2;1/2,1/2,24/25\right)$ & 25 & $\{1\ldots24\}$\\
 $M\left(-7;1/2,1/2,1/2\right)$ & 22+11 & $\{1,3,5,7,9,11\}$\\
 $M\left(-1;1/2,1/11,4/11\right)$ & 22+11 & $\{1,3,5,7,9,11\}$\\
$M\left(-9;1/2,1/2,1/2\right)$ & 30+15 & $\{1,3,5,7,9,11,13,15\}$\\
 $M\left(-1;1/2,2/5,1/15\right)$ & 30+15 & $\{1, 3, 5, 7, 9, 11, 13, 15\}$\\
 $M\left(-13;1/2,1/2,1/2\right)$ & 46+23 & $\{1,3,5,7,9,11,13,15,17,19,21,23\}$\\
 $M\left(-2;1/2,1/2,1/3\right)$ & 6+2 & $\{1,2,4\}$\\
 $M\left(-2;1/2,1/2,3/5\right)$ & 10+2 & $\{1\ldots4,6,8\}$\\
 $M\left(-1;1/3,1/3,1/4\right)$ & 12+3 & $\{1\ldots3,5,6,9\}$\\
 $M\left(-2;1/2,1/2,1/4\right)$ & 12+3 & $\{1\ldots3,5,6,9\}$\\
 $M\left(-1;1/3,1/5,2/5\right)$ & 15+5 & $\{1,2,4,5,7,10\}$\\
 $M\left(-3;1/2,1/2,1/3\right)$ & 15+5 & $\{1,2,4,5,7,10\}$\\
 $M\left(-2;1/2,1/2,7/9\right)$ & 18+2 & $\{1\ldots8,10,12,14,16\}$\\
 $M\left(-2;1/2,1/2,4/7\right)$ & 21+3 & $\{1\ldots6,8,9,11,12,15,18\}$\\
 $M\left(-1;1/4,1/7,4/7\right)$ & 28+7 & $\{1\ldots3,5\ldots7,9,10,13,14,17,21\}$\\
 $M\left(-1;1/4,1/7,4/7\right)$ & 28+7 & $\{1\ldots3,5\ldots7,9,10,13,14,17,21\}$\\
 $M\left(-5;1/2,1/2,1/3\right)$ & 33+11 & $\{1,2,4,5,7,8,10,11,13,16,19,22\}$\\
 $M\left(-1;1/3,1/11,6/11\right)$ & 33+11 & $\{1,2,4,5,7,8,10,11,13,16,19,22\}$\\
 \\
\caption{The Seifert manifolds and the associated $\widetilde {\rm SL}_2(\ZZ)$ representations. }\label{tab:the_big_one}
\end{longtable}
\end{footnotesize}
\end{center}

\begin{center}
\begin{footnotesize}
\begin{longtable}{ccc}
\toprule
 $M_{3}$ & $m+K$ & $\sigma^{m+K,\text{irred}}=\sigma^{(M_3)} $\\
\midrule
\endfirsthead
\multicolumn{3}{c}{{\bfseries \tablename\ \thetable{} -- continued from previous page}} \\
\toprule
 $M_{3}$ & $m+K$ & $\sigma^{m+K,\text{irred}}=\sigma^{(M_3)} $\\
\midrule
\endhead
\bottomrule
\multicolumn{3}{|c|}{{Continued on next page}} \\ \hline
\endfoot
\endlastfoot
 $M\left(-1;1/2,1/3,1/9\right)$ & 18+9 & $\{1,3, 5,7\}$\\
 $M\left(-2;1/2,1/3,2/3\right)$ & 18+9 & $\{1,3,5,7\}$\\
 $M\left(-1;1/2,1/4,1/5\right)$ & 20+4 & $\{1, 3, 4,8, 11\}$\\
 $M\left(-1;1/2,1/3,1/8\right)$ & 24+8 & $\{1,2,5,7,8,13\}$\\
 \\
\caption{The Seifert manifolds and the associated $\widetilde {\rm SL}_2(\ZZ)$ representations. }\label{tab:the_smaller_one}
\end{longtable}
\end{footnotesize}
\end{center}

\subsection{{False-\texorpdfstring{$\vartheta$}{theta}} Invariants from Knot Surgeries}\label{subsec:examples_surgeries}

In this subsection, we use   Definition \ref{conj:surgery_defect} to compute the examples of $\hat Z$-invariants, with and without defects. 
This provides evidence for  the consistency of 
Definition \ref{conj:surgery_defect}, and 
 the validity of Modularity Conjecture \ref{con:modularity}.

\subsubsection{\texorpdfstring{$M_{3}=\Sigma\left(2,3,7\right)$, $m=42$, $K=\{1,6,14,21\}$}{}}
Here we will  
consider two ways in which $\Sigma\left( 2,3,7 \right)$ can be represented as a knot surgery: either as a $(-1)$-surgery on a right-handed trefoil $\mathbf{3}^{r}_{1}$ (also denoted $T(2,3)$) or as a $(+1)$-surgery on a figure eight knot $\mathbf{4}_{1}$. 
The two constructions give rise to two distinct types of expressions that evaluate to the same functions. 
In this subsection we focus on the $\mathbf{3}^{r}_{1}$ construction. The analysis using the Figure 8 knots will be given in \S\ref{subsec:41}. 

The $F_{K}$ polynomial for the right-handed trefoil $\mathbf{3}^{r}_{1}=T(2,3)$ reads \eqref{FKTorus}
\begin{gather}\label{FKrighttref}
\begin{split}
  F_{\mathbf{3}^{r}_{1}}\left( x;q \right) 
  &=\frac{1}{2}\left[ -\left( x^{\frac12}-x^{-\frac12} \right)q + \left( x^{\frac{5}{2}} - {x^{-\frac{5}{2}}}\right)q^{2} + \left( x^{\frac{7}{2}} - {x^{-\frac{7}{2}}} \right)q^{3} \right] + \mathcal{O}\left( q^{6} \right)
\end{split}
\end{gather}
which leads to the following expression after the  Laplace transformation 
\begin{align}
  \label{eq:falselaplace237}
  \begin{split}
  \mathcal{L}^{\left( 0 \right)}_{-1}\left[ \left( x^{\frac12}-x^{-\frac12} \right)F_{\mathbf{3}^{r}_{1}}\left( x;\tau \right) \right]& \sim \hat{Z}_{b_0}\left( \Sigma\left( 2,3,7 \right) ; \tau \right) \\
  \mathcal{L}^{\left(1 \right)}_{-1}\left[ \left( x-x^{-1}\right)F_{\mathbf{3}^{r}_{1}}\left( x;\tau \right)  \right] &\sim \hat{Z}_{b_0}\left( \Sigma\left( 2,3,7 \right); W_{\left( 0,0,1 \right)} ; \tau \right)\\
  \mathcal{L}^{\left( 2 \right)}_{-1}\left[ \left( {x}^{\frac{3}{2}}-{{x^{-\frac{3}{2}}}}\right)F_{\mathbf{3}^{r}_{1}}\left( x;\tau \right) \right] &\sim \hat{Z}_{b_0}\left( \Sigma\left( 2,3,7 \right); W_{\left( 0,0,2 \right)} ; \tau \right) 
  \end{split}
\end{align}
which we checked up to  $\mathcal{O}(q^{20})$. 
For the sake of comparison, on the right hand side we have used the plumbing notation as in \S\ref{subsubsec:237}. As shown there, the above $\hat Z$-invariants are equal, up to an overall $q$-power and the addition of finite polynomials, to the false theta functions $\tilde{\theta}^{m+K}_{1}\left( \tau \right)$,  $\tilde{\theta}^{m+K}_{11}\left( \tau \right)$ and $\tilde{\theta}^{m+K}_{5}\left( \tau \right)$.

\subsubsection{\texorpdfstring{$M_{3}=M\left( -1; \frac{1}{2}, \frac{1}{3}, \frac{1}{9} \right)$, $m=18$, $K=\{1,9\}$}{}}
This manifold, also featured in Table \ref{tab:the_smaller_one},  can be constructed via a $(-3)$-surgery on the right-handed trefoil, whose $F_K(x;\tau)$ invariant is listed in \eqref{FKrighttref}. The surgery formula gives
\begin{align}
  \label{eq:falselaplace239}
  \begin{split}
  \mathcal{L}^{\left( 0 \right)}_{-3}\left[ \left( {x}^{\frac12}-{{x^{-\frac12}}}\right)F_{\mathbf{3}^{r}_{1}}\left( x;\tau  \right) \right] &\sim \hat{Z}_{b_0}\left( M\left( -1; \frac{1}{2}, \frac{1}{3}, \frac{1}{9} \right) ; \tau \right)\\
  \mathcal{L}^{\left( 2 \right)}_{-3}\left[ \left( {x}^{\frac{3}{2}}-{{x^{-\frac{3}{2}}}}\right) F_{\mathbf{3}^{r}_{1}}\left( x;\tau  \right) \right] &\sim \hat{Z}_{b_0}\left( M\left( -1; \frac{1}{2}, \frac{1}{3}, \frac{1}{9} \right); W_{2} ; \tau \right)\\
  \mathcal{L}^{\left( 3 \right)}_{-3}\left[ \left( {x}^{2}-{{x^{-2}}}\right) F_{\mathbf{3}^{r}_{1}}\left( x;\tau  \right) \right] &\sim \hat{Z}_{b_{0}}\left( M\left( -1; \frac{1}{2}, \frac{1}{3}, \frac{1}{9} \right); W_{3} ; \tau \right)\\
 \mathcal{L}^{\left( 4 \right)}_{-3}\left[ ( {x}^{\frac{5}{2}}-{x^{-\frac{5}{2}}})F_{\mathbf{3}^{r}_{1}}\left( x;\tau  \right) \right] &\sim\hat{Z}_{b_0}\left( M\left( -1; \frac{1}{2}, \frac{1}{3}, \frac{1}{9} \right); W_{4} ; \tau \right) ~,
  \end{split}
\end{align}
which we checked up to  $\mathcal{O}(q^{20})$. 
For the sake of comparison, on the right hand side we have used the plumbing notation for the defect operators from equation \eqref{eq:surgeryS2_defect}. 
The plumbing description is in terms of the plumbing graph as in Figure \ref{fig:plumbing_graph_1_v2}, with the weight vector given by $(-1,-2,-3,-9)$. 
The defect invariants are equivalent to the components of the vector-valued quantum modular form 
\[
 \begin{pmatrix}\tilde\theta^{18+9,{\rm irr}}_1\\\tilde{\theta}^{18+9,{\rm irr}}_{3}\\\tilde{\theta}^{18+9,{\rm irr}}_{5}\\\tilde{\theta}^{18+9,{\rm irr}}_{7}
\end{pmatrix}
 (\tau), 
\]
corresponding to the four-dimensional irreducible Weil representation
$\Theta^{18+9,{\rm irr}}$.

{
\subsubsection{\texorpdfstring{$M_{3}=M\left( -1; \frac{1}{2}, \frac{1}{3}, \frac{1}{8} \right)$, $m=24$, $K=\{1,8\}$}{}}
\label{sec:24+8-False-Surgery}
This manifold, which just like previous example is also featured in Table \ref{tab:the_smaller_one}, can be constructed via a $(-2)$-surgery on the right-handed trefoil $\mathbf{3}^{r}_{1}$. The surgery formula now gives
\be
\left(\begin{matrix} \hat{Z}_{b_0}\left({M_3;}W_{0};\tau\right)\\ \hat{Z}_{b_0}\left({M_3;}W_{1};\tau\right)\\\hat{Z}_{b_0}\left({M_3;}W_{2};\tau\right)\\ \hat{Z}_{b_0}\left({M_3;}W_{3};\tau\right)\\
\hat{Z}_{b_1}\left({M_3;}W_{0};\tau\right)\\ \hat{Z}_{b_1}\left({M_3;}W_{2};\tau\right)
\end{matrix} 
\right) \sim\left(\begin{matrix} \tilde\theta^{24+8,{\rm irr}}_1\\  \tilde\theta^{24+8,{\rm irr}}_2 \\ \tilde\theta^{24+8,{\rm irr}}_5\\  \tilde\theta^{24+8,{\rm irr}}_8 \\ \tilde\theta^{24+8,{\rm irr}}_7\\  \tilde\theta^{24+8,{\rm irr}}_{13}  
\end{matrix} 
\right) (\tau) .
\ee
}

\subsection{{Mock and False-\texorpdfstring{$\vartheta$}{theta} Invariants with the Inverted Habiro Series}}
\label{subsec:habiroexamples}

In this subsection we will compute defects $\hat Z$-invariants using Conjecture \ref{con:habiro}, generalizing the conjecture in \cite{Park2106} to include defect operators. 
We present the examples of four Seifert manifolds, displaying different properties while all confirming the Modularity Conjecture \ref{con:modularity}. 
Note that the mock invariants computed here are also provided with indefinite theta expressions in \S\ref{subsec:indef_zhat}. Moreover, we notice that they are all equivalent to some of the optimal Jacobi forms studied in \cite{CD1605}, which are distinguished by their slowese possible growth of coefficients. 

\subsubsection{\texorpdfstring{$M_{3}=\pm\Sigma\left( 2,3,5 \right)$, $m=30$, $K=\{1,6,10,15\}$}{}}

The manifold $M_3=\Sigma(2,3,5)$ can be represented as a plumbed manifold with plumbing diagram given  in Figure \ref{fig:plumbing_graph_1_v2} and weights $(\alpha^{(0)},\alpha^{(1)}_1,\alpha^{(2)}_1,\alpha^{(3)}_1)=(-1,-2,-3,-5)$, and also as the $(-1)$-surgery of the left-handed trefoil knot $\mathbf{3}^l_1$ (also denoted $T(2,-3)$). From the latter point of view, we can make use of Conjecture  \ref{con:habiro} with $p=\pm 1$ in order to compute the defect $\hat Z$-invariants for $ \Sigma(2,3,5)$ as well as $-\Sigma(2,3,5)$. 

For the latter, the results obtained here can be compared to the results obtained as indefinite theta functions using its description as a plumbed manifold, detailed in \S\ref{subsec:indef_zhat}. 

From \begin{equation}
    \Sigma(2,3,5) = S^3_{-1} (\mathbf{3}^l_1) ~,
\end{equation}
and the 
 Habiro coefficients
\begin{equation}
\label{eq:a_m left trefoil}
    a_{-m}(\mathbf{3}^l_1)=(-1)^m q^{\frac{m(m-3)}{2}} ~, \qquad m\in\mathbb{Z} ~, 
\end{equation}
we arrive at 
\begin{align}\label{eq:falseZhat235}
  \begin{split}
  \hat{Z}_{b_0}\left(S_{-1}^{3}(\mathbf{3}^l_1);\tau   \right)
  \sim \tilde{\theta}^{m+K}_1 (\tau)\\
   \hat{Z}_{b_0}\left(S_{-1}^{3}(\mathbf{3}^l_1); W_{1};\tau \right) \sim\tilde{\theta}^{m+K}_7 (\tau) .
  \end{split}
\end{align}
This is exactly what Theorem \ref{thm:modularityBrieskorn} dictates in this case. 
Similarly, from \begin{equation}
    -\Sigma(2,3,5) = S^3_{+1} (\mathbf{3}^r_1) ~,~~ a_{-m}(\mathbf{3}^r_1)=a_{-m}(\mathbf{3}^l_1)\lvert_{q\to q^{-1}}, 
\end{equation}
we get 
\be\label{235mock}
\left(\begin{matrix} \hat{Z}_{b_0}(-\Sigma(2,3,5))\\ \hat{Z}_{b_0}(-\Sigma(2,3,5);W_1)
\end{matrix} 
\right) \sim\left(\begin{matrix} H^{30+6,10,15}_1\\  H^{30+6,10,15}_5
\end{matrix} 
\right)  \sim\left(\begin{matrix} \chi_0 \\  \chi_1
\end{matrix} 
\right) ,
\ee
given by the two order $5$ mock theta functions of Ramanujan $\chi_{0}$ and $\chi_{1}$, which are also the components of the optimal mock Jacobi theta function with the corresponding $m+K$ \cite{CD1605}. 

Assuming  Conjecture \ref{con:habiro} is correct, this example also serves to showcase Conjecture \ref{con:modularity}.2. Expressions of the above invariants in terms of indefinite lattice theta functions have been given in Conjecture \ref{sonj:23torus}.

\subsubsection{\texorpdfstring{$M_{3}=-M\left( -1; \frac{1}{2}, \frac{1}{3}, \frac{1}{8} \right)$, $m=24$, $K=\{1,8\}$}{}}\label{sec:mockfalseIHSM24p8}
Now we consider
\be
M_{3}=-M\left( -1; \frac{1}{2}, \frac{1}{3}, \frac{1}{8} \right) = S^3_{+2}(\mathbf{3}^l_1), 
\ee
corresponding to the 6-dimensional irreducible Weil representation $\Theta^{24+8,{\rm irr}}$, with \[\sigma^{24+8,{\rm irr}}=\{1,2,5,7,8,13\} . \]
The plumbing description has plumbing graph as in Figure \ref{fig:plumbing_graph_1_v2}, with weight assignment $$\left(\alpha^{(0)},\alpha^{(1)}_1,\alpha^{(2)}_1,\alpha^{(3)}_1\right)=(-1,-2,-3,-8).$$ The two inequivalent Spin$^c$ structures are represented by $b_0$ and $b_1=(-1,1,1,3)$. 

The prescription in Conjecture \ref{con:habiro} leads to 
$q$-series 
which can be observed to coincide, up to a pre-factor and a finite polynomial, with the components of the optimal mock Jacobi theta function with the corresponding $m+K$, at least to the order $\mathcal{O}(q^{15})$ we have computed. Explicitly, we have 
\be
\left(\begin{matrix} \hat{Z}_{b_0}({M_3;}W_0)\\ \hat{Z}_{b_0}({M_3;}W_1)\\\hat{Z}_{b_0}({M_3;}W_2)\\ \hat{Z}_{b_0}({M_3;}W_3)\\
\hat{Z}_{b_1}({M_3;}W_0)\\ \hat{Z}_{b_1}({M_3;}W_2)
\end{matrix} 
\right) \sim\left(\begin{matrix} H^{24+8}_1\\  H^{24+8}_2 \\H^{24+8}_5\\  H^{24+8}_8 \\H^{24+8}_7\\  H^{24+8}_{13}  
\end{matrix} 
\right) .
\ee
We collect in Appendix \ref{app:Pnpb} the $P^{p,b+\nu}_{n}$  polynomials that were used to compute the $q$-series for this example. 
The above coincides with the expressions given in terms of indefinite lattice theta functions in Conjecture \ref{sonj:H24}.

\subsubsection{\texorpdfstring{$M_{3}=\pm\Sigma\left( 2,3,7 \right)$, $m=42$, $K=\{1,6,14,21\}$}{}}
From 
\begin{equation}\label{def:mfd237}
        \Sigma(2,3,7) = S^3_{-1} (\mathbf{3}^r_1) ~,~~  a_{-m}(\mathbf{3}^r_1)=(-1)^m q^{-\frac{m(m-3)}{2}} 
\end{equation}
we get using {Conjecture \ref{con:habiro}} the following expressions  
\begin{gather}\label{eq:falsesigma237}
  \begin{split}
  \hat{Z}_{b_0} (S^3_{-1} (\mathbf{3}^r_1);\tau) & = q^{\frac{1}{2}}\left(1+\sum_{n= 1}^{\infty} (-1)^n  \frac{q^{\frac{n\left( n+1 \right)}{2}}}{(q^{n+1};q)_n}  \right) = q^{\frac{83}{168}}\tilde{\theta}^{m+K}_1 (\tau)  ~,\\
  \hat{Z}_{b_0} (S^3_{-1} (\mathbf{3}^r_1);W_{1};\tau) & = q^{\frac{3}{4}}\left(2+\sum_{n=1}^{\infty} (-1)^n  \frac{q^{\frac{n\left( n-1 \right)}{2}}}{(q^{n};q)_n}  \right) =  -q^{\frac{101}{168}}\tilde{\theta}^{m+K}_{5} (\tau) - p_{1}\left( \tau \right) ~,\\
  \hat{Z}_{b_0} (S^3_{-1} (\mathbf{3}^r_1);W_{2};\tau) & = q^{\frac{1}{2}}\left(1+2q+\sum_{n = 1}^{\infty} (-1)^n  \left(\frac{q^{1+\frac{n\left( n-3 \right)}{2}}}{(q^{n};q)_{n-1}}  +
 \left( q^{n}+1+q^{-n} \right)  \frac{q^{\frac{n\left( n+1 \right)}{2}}}{(q^{n+1};q)_n} \right)  \right) \\
  &= -q^{\frac{131}{168}}\tilde{\theta}^{m+K}_{11}(\tau)- p_{2}\left( \tau \right) ~,
  \end{split}
\end{gather}
where $p_{1}\left( \tau \right) = -2q^{3/4}$ and  $p_{2}\left( \tau \right) = -2q^{3/2}$, consistent with the plumbing result in \S\ref{subsubsec:237}. 
Inverting the orientation by inverting the sign of the surgery and taking the mirror knot, namely
\[
   - \Sigma(2,3,7) = S^3_{+1} (\mathbf{3}^l_1) ,
\]
 we get
\begin{align}\label{eq:mocksigma237}
  \begin{split}
   \hat{Z}_{b_0} (S^3_{1} (\mathbf{3}^l_1);\tau) 
   &= q^{-\frac12}\left(1+\sum_{n=1}^{\infty} \frac{q^{n^{2}}}{(q^{n+1};q)_n}\right) = q^{-\frac12}\left(1 + q + q^3 + q^4 + q^5 + 2q^7 + \mathcal{O}\left( q^{8} \right)\right)\\
   \hat{Z}_{b_0} (S^3_{1} (\mathbf{3}^l_1);W_{1};\tau) 
   &= q^{-\frac34}\left(2+\sum_{n=1}^{\infty} \frac{q^{n^{2}}}{(q^{n};q)_n}  \right) =q^{-\frac34}\left(2+ q + q^2 + q^3 + q^4 + q^5 + 2q^6 +\mathcal{O}\left( q^{7} \right)\right)\\
  \hat{Z}_{b_0} (S^3_{1} (\mathbf{3}^l_1); W_{2};\tau) 
&= q^{-\frac12}\left(1 + \frac{2}{q} + \sum_{n=1}^{\infty} \left(-\frac{q^{n^{2}-n}}{(q^{n};q)_{n-1}} + \left( q^{n}+1+q^{-n} \right) \frac{q^{n^{2}}}{(q^{n+1};q)_n} \right) \right) \\
&= q^{-\frac{3}{2}}\left(2 + q + q^2 + 2q^3 + q^4 + 2q^5 + \mathcal{O}\left( q^{6} \right)\right)~,
  \end{split}
\end{align}
which are the components of the optimal mock Jacobi theta function with the corresponding $m+K$. In this case, they are proportional to Ramanujan's order 7 mock theta functions $F_0$, $F_1$, and $F_2$. 
In other words, we have
 \be
\left(\begin{matrix} \hat{Z}_{b_0}\\ \hat{Z}_{b_0}(W_1)\\\hat{Z}_{b_0}(W_2)
\end{matrix} 
\right) \sim\left(\begin{matrix} H^{m+K}_1\\ H^{m+K}_{5} \\H^{m+K}_{11}
\end{matrix} 
\right)  \sim\left(\begin{matrix}F_0 \\ F_1\\F_2
\end{matrix} 
\right)  .
\ee

\subsection{{Mock-\texorpdfstring{$\vartheta$}{theta}} Functions from a Hyperbolic Knot} 
\label{subsec:41}

\noindent 
Contrary to the trefoil, the figure 8 knot $\mathbf{4}_{1}$ is not a plumbed knot. As a result, Conjecture \ref{con:habiro} does not apply and we do not expect it to give the right answer.

That said, we know that 
\[
S^3_{-1} (\mathbf{3}^r_1) = S^3_{+1}(\mathbf{4}_{1}) = \Sigma(2,3,7) , ~
S^3_{+1} (\mathbf{3}^l_1) = S^3_{-1}(\mathbf{4}_{1}) = -\Sigma(2,3,7) , 
\]
and the construction using the trefoil knots does give surgery expressions for the (defect) $\hat Z$-invariants using Conjecture \ref{con:habiro}. We hence wonder whether a similar expression could be obtained for surgeries along the $\mathbf{4}_{1}$ knot.

In the notation of \S\ref{subsec:knot}, we can re-express $\hat Z_0( \Sigma(2,3,7); W_{\nu},\tau)$ in terms of an infinite sum involving  ${\mathcal L}^{(\delta_\nu)}_{+1}\left((x^\ha+x^{-\ha})^{\delta_\nu}/{  D_n}\right)$ as
\begin{equation}\label{eqn:zhat411}
\hat Z_0( \Sigma(2,3,7); W_{\nu},\tau) =\hat Z_0( S^3_{+1}(\mathbf{4}_{1}); W_{\nu},\tau)
\sim \sum_{n\geq 0} 
{\mathcal L}^{(\delta_\nu)}_{+1}\left({(x^\ha+x^{-\ha})^{\delta_\nu}\over{  D_n}}\right) \tilde a^{(\nu)}_{-n-1}(q^{-1})
\end{equation}
where 
{\begin{gather}\begin{split}
& \tilde a^{(\nu)}_{-n-1}(q^{-1}) = \\ & a^{}_{-n-1}(\mathbf{3}^r_1;q) \sum_{\ell=0}^{{\rm min}(n-1,\lfloor {\nu\over 2}\rfloor)}\frac{ {\mathcal L}^{({\delta_\nu})}_{-1}\left((x^\ha+x^{-\ha})^{\delta_\nu}/{ D_{n-\ell}}\right) }{{\mathcal L}^{({\delta_\nu})}_{+1}\left((x^\ha+x^{-\ha})^{\delta_\nu}/{  D_n}\right) }  \sum_{j=0}^{\lfloor {\nu\over 2}\rfloor - \ell}  c_j^{(\nu/2)} S_{ { \lfloor\frac{\nu}{2}\rfloor} -\ell-j}(Q_{ n},\dots,Q_{{ n}-\ell}) 
\end{split}\end{gather}}
in the notation of Conjecture \ref{con:habiro}. The key point here is that $
\tilde a^{(\nu)}_{-n-1}(q)$ is a finite polynomial for any positive integers $n$ and $\nu$. To see this, note that 
{
$$\frac{ {\mathcal L}^{({\delta_\nu})}_{-1}\left((x^\ha+x^{-\ha})^{\delta_\nu}/{ D_{n-\ell}}\right) }{{\mathcal L}^{({\delta_\nu})}_{+1}\left((x^\ha+x^{-\ha})^{\delta_\nu}/{  D_n}\right) }
=(-1)^n \frac{\left(q^{n+1-\delta_\nu};q\right)_n}{\left(q^{ { n-\ell} +1-\delta_\nu};q\right)_{{ n-\ell}}}\,
q^{{ (n-\ell)}^2-{n(n+1)\over2}+\delta_\nu(\ha {- n+\ell}) }
$$
}
is a finite polynomial in $q$ for all $0\leq  n-\ell\leq  n$, as can be seen from the cyclotomic polynomial expression for the q-Pochhammer symbols, using $1-q^n = \prod_{m|n}\Phi_m(q)$. Finally,  the remaining factor is a finite sum over finite polynomials $c_j^{(\nu/2)} S_{m-\ell-j}(Q_m,\dots,Q_{m-\ell})$.

Similarly, considering the $(+1)$-surgery of the left-handed trefoil,  we write the invariants for $- \Sigma(2,3,7)$ in the following way
\begin{gather}\label{eqn:zhat412}\begin{split}
\hat Z(- \Sigma(2,3,7); W_{\nu},\tau)  =\hat Z_0( S^3_{-1}(\mathbf{4}_{1}); W_{\nu},\tau)&\sim \sum_{n\geq 0} 
{\mathcal L}^{(\delta_\nu)}_{-1}\left({(x^\ha+x^{-\ha})^{\delta_\nu}\over{  D_n}}\right) \tilde a^{(\nu)}_{-n-1}(q). 
\end{split}
\end{gather}

Next we exploit the similarity of \eqref{eqn:zhat411}-\eqref{eqn:zhat412} to the conjectural surgery formula for plumbed knots \eqref{eq:habiro_defect}, ignoring the fact that the figure eight knot is not a plumbed knot, and wonder out loud what one obtains if one simply generalise \eqref{eqn:zhat411} and \eqref{eqn:zhat412} to  $S^3_{\pm p}(\mathbf{4}_{1})$ in a way analogous to the discussion in \S\ref{subsec:knot}: 
\begin{gather}\label{eqn:crazy41}
\begin{split}
&\hat Z(S^3_{\pm p}(\mathbf{4}_{1}); W_{\nu},\tau) \\&\stackrel{?}{\sim} \sum_{n\geq 0}  
{\mathcal L}^{(\delta_\nu)}_{\pm p}\left({(x^\ha+x^{-\ha})^{\delta_\nu}\over{  D_n}}\right) \tilde a^{(\nu)}_{-n-1}(q^{ \mp 1}) =\sum_{n\geq 0} 
f_{m}^{p,b}(q^{\pm 1}) \, P_{m}^{p,b}(q^{\pm 1})
\tilde a^{(\nu)}_{-n-1}(q^{ \mp 1})
\end{split}
\end{gather}
where the ``?" indicates the fact that we do not have the relation to plumbed manifolds to justify the above expression. However, we find it remarkable that 
 the above formula gives rise to very interesting functions, which are consistent with the False-Mock Conjecture \ref{con:False-Mockv2} and the Modularity Conjecture \ref{con:modularity}.  For us, this justifies recording the highly conjectural equations here for the benefit of the interested readers.

Using the expression \eqref{eqn:crazy41}, we obtain the results in Table \ref{tab:figure8habiro}, which states that the right-hand side of \eqref{eqn:crazy41} using  the indicated values of $b$ and $\nu$, is,  up to a pre-factor and possibly a finite polynomial, the same as the false or mock theta function listed in the column ``$X^{m+K}_r$'', where $H^{m+K}_r$ denotes the component of the optimal mock Jacobi theta function in \cite{CD1605}\footnote{We've checked it up to $\mathcal{O}(q^{100})$ in the false cases and $\mathcal{O}(q^{15})$ in the mock cases.}. In the table we also list the Seifert representation of the resulting three-manifold $S_{\pm p}(\mathbf{4}_{1})$. In all cases, the plumbing graph is given in Figure \ref{fig:plumbing_graph_1_v2}, for which we write the weight assignment in the fourth column, labeled by $(\alpha^{(0)},\alpha^{(1)}_1,\alpha^{(2)}_1,\alpha^{(3)}_1)$. On the right side of the table, we list the comparison to the results obtained from the plumbing prescription
using Definition \ref{conj:plumbed_defect}, where we indeed see a non-trivial matching.

\begin{rmk}
Note that some of the components, corresponding to the elements of the relevant $\sigma^{m+K}$ or  $\sigma^{m+K,{\rm irr}}$, are missing in the Table. This is because some of the components only appear when more general defect lines, not just those corresponding to the figure 8 knot, are included. For instance, $\tilde \theta^{20+4,{\rm irr}}_8$  appears as $\hat Z_{b_0}(W_{(0,1,2)})$ in the notation of the Table  \ref{tab:figure8habiro} \footnote{$\tilde \theta^{20+4,{\rm irr}}_4$  appears as $\hat Z_{b_0}(W_{(0,1,0)})$ in the notation of the Table  \ref{tab:figure8habiro}.}, but one doesn't see this false theta function when only Wilson lines corresponding to the distinguished knot is considered.  
Similar comments also hold for the case $m=12$, $K=\{1,3\}$. 
\end{rmk}

\begin{table}[tp]
\begin{center}
\scalebox{0.85}{
\begin{tabular}{c|ccc|ccc|cc}
\toprule
$p$ & $M_3$ & $m+K$ & $(\alpha^{(0)},\alpha^{(1)}_1,\alpha^{(2)}_1,\alpha^{(3)}_1)$ & $b$ & $\nu$ & $X^{m+K}_r$ & $\left(\nu_1,\nu_2,\nu_3\right)$ & Spin$^c$ \\\toprule
\multirow{3}{*}{1} & \multirow{3}{*}{$\Sigma(2,3,7)$} & \multirow{3}{*}{$42+6,14,21$} & \multirow{3}{*}{$(-1,-2,-3,-7)$} & 0 & 0 &  $\theta^{m+K}_1$ & $(0,0,0)$ & (-1,1,1,1)\\
 &  &  & & 0 & 1 & $\theta^{m+K}_{5}$ & $(0,0,1)$ &(-1,1,1,1)\\
 &  &  & & 0 & 2 & $\theta^{m+K}_{11}$ & $(0,0,2)$ &(-1,1,1,1)\\\midrule
\multirow{3}{*}{-1} & \multirow{3}{*}{$-\Sigma(2,3,7)$} & \multirow{3}{*}{$42+6,14,21$} &\multirow{3}{*}{$(1,2,3,7)$} & 0 & 0 & $H_1^{42+6,14,21}$ &  & \\
  & & &  & 0 & 1 & $H^{42+6,14,21}_{5}$ &  & \\
 &  & & & 0 & 2 & $H^{42+6,14,21}_{11}$ &  & \\\midrule
\multirow{5}{*}{2} & \multirow{5}{*}{$M\left(-1;\frac12,\frac14,\frac15\right)$} & \multirow{5}{*}{$20+4$} &\multirow{5}{*}{$(-1,-2,-4,-5)$} & 0 & 0 & $\tilde\theta^{m+K,{\rm{irr}}}_1$ & $(0,0,0)$ &(-1,1,1,1) \\
& &  &  & 2 & 0 & $\tilde\theta^{m+K,{\rm{irr}}}_{11}$ & $(0,0,0)$ &(-1,1,-1,1) \\
 & &  &  & 0 & 1 & $\tilde\theta^{m+K,{\rm{irr}}}_{{3}}$ & $(0,1,3)$ &(-1,1,1,1) \\
 & & &  & 0 & 2 & $\tilde\theta^{m+K,{\rm{irr}}}_{{7}}$ & $(0,2,6)$ & (-1,1,1,1)\\
 \midrule
\multirow{5}{*}{-2} & \multirow{5}{*}{$-M\left(-1;\frac12,\frac14,\frac15\right)$} & \multirow{5}{*}{$20+4$} &\multirow{5}{*}{$(1,2,4,5)$} & 0 & 0 & $H^{20+4}_1$ &  & \\
 &  & & & 2 & 0 & $H^{20+4}_{11}$ &  & \\
 &  & & & 0 & 1 & $H^{20+4}_{3}
 $ &  &\\
 &  & & & 0 & 2 & $H^{20+4}_{7}
 $ &  & \\
 \midrule
\multirow{4}{*}{3} & \multirow{4}{*}{$M\left(-1;\frac13,\frac13,\frac14\right)$} & \multirow{4}{*}{$12+3$} &\multirow{4}{*}{$(-1,-3,-3,-4)$} & 0 & 0 & $\tilde\theta^{m+K}_1$ & $(0,0,0)$ &(-1,1,1,1) \\
 &  & & & 2 & 0 & $\tilde\theta^{m+K}_{9}$ & $(0,0,0)$ &(-1,1,-1,1) \\
 &  & & & 0 & 1 & $\tilde\theta^{m+K}_3$& $(1,0,0)$ & (-1,1,1,1)\\
 &  & & & 2 & 1 & $\tilde\theta^{m+K}_5$& $(1,0,0)$ & (-1,1,-1,1)\\\midrule
\multirow{4}{*}{-3} & \multirow{4}{*}{$-M\left(-1;\frac13,\frac13,\frac14\right)$} & \multirow{4}{*}{$12+3$} &\multirow{4}{*}{$(1,3,3,4)$} & 0 & 0 & $H^{12+3}_1$&  & \\
 &  & & & 2 & 0 & $H^{12+3}_{9}$&  & \\
 &  & & & 0 & 1 & $H^{12+3}_3$&  & \\
 &  & & & 2 & 1 & $H^{12+3}_5$&  & \\
\bottomrule
\end{tabular}}
\caption{Quantum modular forms arising from \eqref{eqn:crazy41}.  }\label{tab:figure8habiro}
\end{center}
\end{table}

\section{Conclusion and Discussion}
In this paper, we shed light on the three important questions regarding the modular and algebraic aspects of of the $\hat Z$-invariants: the $\widetilde{\rm SL}_2(\ZZ)$ action on the invariants, the construction of mock invariants, and the construction of vertex operator algebras associated with the mock invariants. 
This work also leads to many more interesting research questions, which we will briefly list in this final section.

\begin{itemize}
\item All the proposals, definitions, conjectures, theorems and discussions in this paper should have counterparts for other ADE gauge groups different from $SU(2)$, and three manifolds other than Seifert manifolds with three exceptional fibres.  It will be illuminating to develop those explicitly. 
\item Equipped with the defect invariants introduced in this paper, the resurgence analysis of $\hat Z$-invariants should reveal more complete structures of complex Chern-Simons theory. In particular, the analysis of the relation to non-Abelian flat connections, analogous to those discussed in \cite{3d}, should now become more complete with these defect invariants. 
\item In \S\ref{sec:regindefTheta} we compared the indefinite theta function formulation to various other possible approaches to the mock invariants, including the conjectural positive surgery formula of \cite{Park2106} and the Appell-Lerch sum continuation. A very important exercise is to further compare with the interesting approach proposed in \cite{costin2023going} using resurgence techniques. 
\item In \cite{3d}, the authors have proposed to understand the associated vertex operator algebras from the point of view of Kazhdan-Lusztig correspondence to quantum groups. In this picture, the VOAs of the false and mock sides should correspond to the negative and positive zones, respectively. It would be interesting to study the cone VOAs we proposed from this perspective. 
More generally, it would be very illuminating to have a map going directly from the VOA relevant for $M_3$ to  the VOA relevant for $-M_3$.

\item At least two main mysteries remain surrounding the mock $\hat Z$-invariants: the exponential singularities near certain cusps that need to be subtracted (cf. \eqref{eqn:leaking}), the role of the positive-definite lattice \eqref{positive_lattice} and the chosen cone. Understanding their physical interpretation will hold the key to unlocking many key aspects of three-dimensional topology and the corresponding 3d SCFTs. 
\end{itemize}

\subsection*{Acknowledgements}
The authors would like to thank Sergei Gukov, Mrunmay Jagadale, Ciprian Manolescu, and Sunghyuk Park for very helpful discussions. We also thank Peter Paule for sharing his implementation of the qZeil algorithm. MC and DP would also like to thank the American Institute of Mathematics for hosting a collaboration meeting on related topics. 
This work has been supported by the~Vidi grant (number 016.Vidi.189.182) from the Dutch Research Council (NWO). The work of I.C. has been supported in part by the European Union’s Horizon 2024 Research and Innovation Programme under the Marie Skłodowska-Curie Grant Agreement No. 101204790 (QFT2VOA-DuStRel). The work of P.K. has been supported by the Sonata grant no. 2022/47/D/ST2/02058 funded by the Polish National Science Centre.

\appendix

\section{Mock Modular Properties Analysis}\label{app:mock-proof}

In this section we will study the mixed mock modular properties of the regularized indefinite theta series $\Theta_{A,\sa_{\hat\epsilon},\ssb,\ssc_1,\ssc_2}(\tau)$ that appeared in \S\ref{subsec:relation}, as a part of the surgery result $ \hat Z_0^{{\rm Reg. Surg.}}$ for three-manifolds $-\Sigma(s,t,str+1)$. 
First we will analyse the modular completion using \eqref{prop:8-33HarmonicMaassFormsV2}, with the data specifying  the theta function given by \eqref{23r:parametersR}. 

From \eqref{prop:8-33HarmonicMaassFormsV2}, we define 
\begin{equation}\label{secproof:genericshadowcomponent}
    S_{j,\ssl,\sa} :=   
{\ex}\left(\frac{d_j B(\ssc_j,\ssb)B(\ssc_j,\ssl+\sa)}{2n_jm}\right)\,
\theta^1_{n_j m,-d_jB(\ssc_j,\ssl+\sa)} \left(\frac{\tau}{d_j}\right) ~ \overline{\sum_{\nu\in(\ssl+\sa)^\perp_j + \ZZ \ssc_{j, \perp} } e(B(\nu,\ssb)) q^{\frac{\normsq{\nu}}{2}}}
\end{equation}
for $j\in \{1,2\}$, $\ssl\in P_j$ (see the text after \eqref{prop:8-33HarmonicMaassForms} for the definition of $P_j$), so that
the function 
\be\label{def:F}
F_{(p_1,p_2,p_3),{\bnu},\chi}:= 
 \sum_{\substack{\hat\epsilon=(\epsilon_1,\epsilon_2,\epsilon_3)\\\in (\ZZ/2)^{3}}}
(-1)^{\hat\epsilon} \Theta_{A,\asub,\ssb,\ssc_1,\ssc_2}(\tau), 
\ee
where $\chi$ denotes the unfixed parameter in \eqref{23r:parametersR}, 
has completion $\hat F_{(p_1,p_2,p_3),{\bnu},\chi}$ that satisfies 
\be\label{appendix_eqn_shadow}
 i {\sqrt{2m}} \sqrt{2\Im(\tau)}\, \overline{{\partial\over\partial \bar{\tau}}\hat{F}_{(p_1,p_2,p_3),{\bnu},\chi}(\tau)}  =
\sum_{j=1,2} (-1)^j  {\sqrt{\frac{d_j}{n_j}}}  \sum_{\substack{\ssl\in P_j\\ \hat\epsilon\in(\ZZ/2)^{3}} }  (-1)^{\sum_{i=1}^3\epsilon_i}
     \,S_{j,\ssl,\asub}  . 
\ee
To analyse the above, we will compute the summand with $j=1$ and $j=2$ separately. 

\begin{lem}\label{lem:c1_contribution}
\be\label{lemma:c1_contribution}
-{{\ex}\left(-{\chi\over 4x}\right)\over2\overline{f_{x,\chi}(\tau)}}  { \sqrt{\frac{d_1}{n_1}}} \sum_{\substack{\ssl\in P_1\\ \hat\epsilon\in(\ZZ/2)^{3}} }   (-1)^{\sum_{i=1}^3\epsilon_i}
     \,S_{1,\ssl,\asub}  = \theta^{1,m+K}_{r_{\bnu}}
\ee
where $m$, $K$, $r_{\bnu}$ are as in Theorem \ref{thm:modularityBrieskorn}. 
\end{lem}

\begin{proof}
For $\ssc_{1} =(1,0)^T$, we have $n_1=d_1=1$ and $P_1\cong  \ZZ^{2}/(\ssc_j\ZZ\oplus\ssc_j^{\perp}\ZZ) \cong \{0\}$. 
As a result, the left-hand side of \eqref{lemma:c1_contribution} reduces to 
\begin{gather}\notag
\begin{split}
&-{{\ex}\left(-{\chi\over 4x}\right)\over2\overline{f_{x,\chi}(\tau)}}     
 \sum_{ \hat\epsilon\in(\ZZ/2)^{3} }   (-1)^{\sum_{i=1}^3\epsilon_i} S_{1,0,\asub} \\
&= -{{\ex}\left(-{\chi\over 4x}\right)\over2\overline{f_{x,\chi}(\tau)}}    \sum_{ \hat\epsilon\in(\ZZ/2)^{3} }   (-1)^{\sum_{i=1}^3\epsilon_i}  \theta^1_{m,-r_{\bnu,\hat \epsilon}} 
\overline{\sum_{k\in \ZZ} (-1)^k {\ex}\left(-{\chi\over 4x}\right) q^{{x\over 2}(k-{\chi\over 2x})^2 }}\\ 
&= 
\sum_{\substack{(\epsilon_{1},\epsilon_{2},\epsilon_{3})\in (\ZZ/2)^{3} \\ \sum_{j}\epsilon_{j}\equiv 0\Mod{2}}} 
  \theta^1_{m,r_{\bnu,\hat \epsilon}} = \theta^1_{m,r_{\bnu}}
\end{split}\end{gather}
where we have   used the expression for the theta function $f_{x,\chi}$ \eqref{eq:theta-function} 
in the first equality, and used \eqref{eqn:folding} in the last equality. 
\end{proof}

Next we study the contribution, proportional to  
\be
\sum_{\substack{\ssl\in P_2\\ \hat\epsilon\in(\ZZ/2)^{3}} }   (-1)^{\sum_{i=1}^3\epsilon_i}
     \,S_{2,\ssl,\asub}
\ee
from the $j=2$ summand in \eqref{appendix_eqn_shadow}. For this purpose, note that we can take 
\[
\ssc_{2,\perp}=\left( \begin{matrix}
              1 \\
               \bar p_3 \end{matrix} \right)
\]
from which we see 
\[
P_2 \cong  \left\{\left( \begin{matrix}
              0 \\
               \ell \end{matrix} \right) , \ell\in \ZZ/L\ZZ \right\}, ~L=\bar p_3 x -2p_3  = -{ |\ssc_2|^2\over 2p_3 x } = {{ |\ssc_{2,\perp}|^2\over \bar p_3  }}. 
\]
Write 
\begin{equation}\label{secproof:genericshadowcomponent2}
    S_{2,\ssl,\sa} :=   
{\ex}\left(\frac{d_2 B(\ssc_2,\ssb) }{2n_2m}\Xsub \right)\,
\theta^1_{n_2 m,-d_2 \Xsub} \left(\frac{\tau}{d_2}\right) ~ \overline{ {\ex}\left({ \Asub\over 2L}\right)\sum_{ k\in \ZZ } q^{{|\ssc_{2,\perp}|^2\over 2} (k+{\Asub\over L\bar p_3})^2 }}
\end{equation}
where we have assumed $\bar p_3/2\in \ZZ$ (which is the case for Lemma \ref{lem:c2_contribution}), and defined for a given $\bnu$
\be\label{condition:trans}
 \Xsub: =B(\ssc_2,\ssl+\asub), ~~ \Asub := B(\ssc_{2,\perp},\ssl+\asub), ~~{\rm{with}}~\ssl=\left( \begin{matrix}
              0 \\
               \ell \end{matrix} \right).
\ee

As a result, the $j=2$ piece vanishes if there exists a bijection $\Gamma$ on $(\ZZ/2)^{3}\times P_2$ such that 
\[
\Gamma (\hat \epsilon, \ell) = (\Gamma_1(\hat\epsilon), \Gamma_2(\ell))
\]
has the property that it preserves $\Xsub$, and maps ${1\over L\bar p_3}\Asub \mapsto \pm {1\over L\bar p_3}\Asub \Mod{\ZZ}$, and is moreover odd on
\(
(-1)^{\hat\epsilon}  {\ex}\left({ \Asub\over 2L}\right)  
\).

It turns out that, among the surgery results $ \hat Z_0^{{\rm Reg. Surg.}}$ for all three-manifolds of the form $-\Sigma(s,t,str+1)$, the family related to the torus knot $T(2,3)$, the trefoil knot behaves differently from the rest. It stems from the vanishing of the contribution by the $\ssc_2$ to the modular completion to the theta function
$\Theta_{A,\sa_{\hat\epsilon},\ssb,\ssc_1,\ssc_2}(\tau)$, captured by the following lemma. 

\begin{lem}\label{lem:c2_contribution}
\be\label{lemma:c2_contribution}
\sum_{\substack{\ssl\in P_2\\ \hat\epsilon\in(\ZZ/2)^{3}} }   (-1)^{\hat \epsilon}
     \,S_{2,\ssl,\asub}  = 0 
\ee
where $m$, $K$, $r_{\bnu}$ are as in Theorem \ref{thm:modularityBrieskorn}, when $p_2=2$, $p_3=xp_1-2$, and $x$ is an odd positive integer number. 
\end{lem}

\begin{proof}
In this case, we have $L=4$ and the above expression in  \ref{lemma:c2_contribution} is equal to 
\[
\sum_{\substack{\ell\in \ZZ/4\ZZ \\ \hat\epsilon\in(\ZZ/2)^{3}} }   (-1)^{\hat \epsilon}{\ex}\left(\frac{d_2 B(\ssc_j,\ssb) }{2n_2m}\Xsub \right)\,
\theta^1_{n_2 m,-d_j \Xsub} \left(\frac{\tau}{d_2}\right) ~ \overline{ {\ex}\left({ \Asub\over 2L}\right)\sum_{ k\in \ZZ } q^{{|\ssc_{2,\perp}|^2\over 2} (k+{\Asub\over L\bar p_3})^2 }}= 0 .
\]
From the above discussion, let us first study the symmetries of $\Xsub$. 
From 
\[
\Xsub=-p_3\chi-mx +x\left(2p_3\ell+\sum_i (-1)^{\epsilon_i} A_i\bar p_i\right), ~~ A_i=1+\nu_i
\]
\noindent

We see that  $\Xsub$ is invariant under 
\[
\epsilon_j \mapsto \epsilon_j +\delta_{i,j}, ~\ell \mapsto \ell+\Delta\ell_i
\]
for $i=1,2$, with 
\[
\Delta\ell_i= (-1)^{\epsilon_i} A_i {m\over p_i p_3}.
\]

\noindent
{\bf{Case 1: odd $A_2$}}\\
From 
\[\Asub =-m-\sum_i (-1)^{\epsilon_i} A_i\bar p_i + 2p_1 x\ell -p_1\chi, 
\]
it is easy to check that 
\begin{gather}\notag\begin{split}
\Asub+A_{\ell+\Delta\ell_1,(\epsilon_1+1,\epsilon_2,\epsilon_3)} 
&= 4p_1\left(-{\chi\over 2}+ (-1)^{\epsilon_2}{A_2p_3\over 2} -p_3 +\ell x +(-1)^{\epsilon_1}xA_1+(-1)^{\epsilon_3}A_3 \right) \\
& = 4p_1\times\\&\left(-{\chi\over 2}+ (-1)^{\epsilon_2}{xA_2p_1\over 2} - (-1)^{\epsilon_2}A_2 -p_3 +\ell x +(-1)^{\epsilon_1}xA_1+(-1)^{\epsilon_3}A_3 \right)  \\
\Asub+A_{\ell+\Delta\ell_1+\Delta\ell_2,(\epsilon_1+1,\epsilon_2+1,\epsilon_3)} 
&= 4p_1\left(-{\chi\over 2}+ (-1)^{\epsilon_2}{xA_2p_1\over 2} -p_3 +\ell x +(-1)^{\epsilon_1}xA_1+(-1)^{\epsilon_3}A_3 \right) 
\end{split}\end{gather}
which shows that, 
when $A_2\equiv 1 \Mod{2}$, we have 
\[
\Asub+A_{\ell+\Delta\ell_1,(\epsilon_1+1,\epsilon_2,\epsilon_3)}  \equiv \Asub+A_{\ell+\Delta\ell_1+\Delta\ell_2,(\epsilon_1+1,\epsilon_2+1,\epsilon_3)} \equiv 0\Mod{4p_1}
\]
and 
\[
\frac{\Asub+A_{\ell+\Delta\ell_1,(\epsilon_1+1,\epsilon_2,\epsilon_3)}  }{4p_1}+1  \equiv \frac{\Asub+A_{\ell+\Delta\ell_1+\Delta\ell_2,(\epsilon_1+1,\epsilon_2+1,\epsilon_3)}}{4p_1} \Mod{2}. 
\]
Moreover, from 
\[
\Asub = {1\over 2}\left(\Asub+A_{\ell+\Delta\ell_1,(\epsilon_1+1,\epsilon_2,\epsilon_3)}\right)-4(-1)^{\epsilon_1}A_1 
\]
we see that 
 half of $\ell\in \ZZ/4$ satisfies $\Asub+A_{\ell+\Delta\ell_1,(\epsilon_1+1,\epsilon_2,\epsilon_3)} \equiv 0\Mod{L\bar p_3 = 8p_1}$, and simultaneously $\Asub\equiv 0\Mod{4}$, while the other satisfies $\Asub+A_{\ell+\Delta\ell_1+\Delta\ell_2,(\epsilon_1+1,\epsilon_2+1,\epsilon_3)} \equiv 0\Mod{L\bar p_3 = 8p_1}$ and simultaneously $\Asub\equiv 2\Mod{4}$. 
Hence, we see that in the former case, the 
transformation $\Gamma (\hat \epsilon, \ell) = ((\epsilon_1+1,\epsilon_2,\epsilon_3),  \ell+\Delta\ell_1)$ satisfies the condition below \eqref{condition:trans} and in the latter case 
 $\Gamma (\hat \epsilon, \ell) = ((\epsilon_1+1,\epsilon_2+1,\epsilon_3),  \ell+\Delta\ell_1+\Delta\ell_2)$ does the job.\\

\noindent
{\bf{Case 2: even $A_2$}}\\
Note 
\begin{gather}\begin{split}
A_{\ell+\Delta\ell_2,(\epsilon_1,\epsilon_2+1,\epsilon_3)} - \Asub = 4p_1 (-1)^{\epsilon_2}A_2(xp_1-1) 
\end{split}\end{gather}
so $A_{\ell+\Delta\ell_2,(\epsilon_1,\epsilon_2+1,\epsilon_3)} \equiv \Asub \Mod{\bar p_3 L=8p_1}$ when $A_2$ is even. 
We see that the above transformation satisfies the condition below \eqref{condition:trans}.

Combining the above, we have proven the statement of the Lemma.  
\end{proof}

From Lemma  \ref{lem:c2_contribution}-\ref{lem:c2_contribution}, we see from \eqref{appendix_eqn_shadow} that the completion of the combined indefinite lattice theta function
\be\label{appendix_eqn_shadow2}
 i {\sqrt{2m}} \sqrt{2\Im(\tau)}\, \overline{{\partial\over\partial \bar{\tau}}\hat{F}_{(p_1,p_2,p_3),{\bnu}}(\tau)}  \sim 
\overline{f_{x,\chi}} \theta^{1,m+K}_{r_{\bnu}}  
\ee
takes a very simple form. 

For the rest of the cases where $T(s,t)\neq T(2,3)$,  this will be replaced by a sum of different terms, each a product of a holomophic and an anti-holomorphic function.

{
For completeness, we include here a calculation of
the S-transform of the linear combination of the regularized indefinite theta function $F_{(p_1,p_2,p_3),{\bnu},\chi}$ for $(p_1,p_2,p_3)=(2,3,6r+1)$, as defined in \eqref{def:F}, via an application of \eqref{SL2_IND}. {For simplicity we will restrict our analysis to the case when $6r+1$ is prime. }
Parametrize elements  ${\boldsymbol \mu}$ in $L^\ast/L$, which has size $|L^\ast/L|=2m(2r+1)$, in \eqref{SL2_IND} as 
\be
{\boldsymbol \mu}=\left( 
              \frac{\alpha_1}{12(6r+1)}~~,~~
              \frac{\alpha_2}{2r+1}  \right)^T =
              \left( 
              \frac{\alpha_1}{2m}~~,~~
              \frac{\alpha_2}{x}  \right)^T
              ~,
\ee
with
\be 
\alpha_1\in {\mathbb Z}/2m~, \qquad \alpha_2\in{\mathbb Z}/(2r+1)  ~,
\ee
}
it is then a straightforward exercise to show 
{
\begin{gather}
    \begin{split}\label{eq:intermediateS3}
      &  F_{(p_1,p_2,p_3),{\bnu},\chi} (\tau)= -  \, 
\frac{8i\tilde\tau}{\sqrt{|{\rm{det}}(A)|}} \\
   & \qquad \qquad 
   \sum_{{\boldsymbol \mu \in {L^\ast/L}}}  \sum_{\sn\in\ZZ^2} 
   (-1)^{n_2}e^{\frac{\pi i \chi\alpha_2}{(2r+1)}} \, s(\alpha_1) \,
\hat{\rho}^{\ssc_1,\ssc_2}({\boldsymbol \mu}+\ssb+\sn;\tilde\tau)
 \tilde{q}^{\tfrac{|{\boldsymbol \mu}+\ssb+\sn|^2}{2}} 
 ~, 
    \end{split}
\end{gather}
where the numerical factor in the sum is given by
\begin{gather}\label{sinfactorsa1} 
\begin{split}
s(\alpha_1)&:= \prod_{i=1}^3 \sin\left(\alpha_1\, {\pi (1+\nu_i)\over p_i} \right)= \sin\left( \frac{\pi\alpha_1(1+\nu_1)}{2} \right)
 \sin\left( \frac{\pi\alpha_1(1+\nu_2)}{3} \right)
 \sin\left( \frac{\pi\alpha_1(1+\nu_3)}{6r+1} \right)  \, \\
 & \in\left\{ 0 \right\}\cup \left\{ \pm \frac{\sqrt{3}}{2} \sin\left(\frac{\pi k}{6r+1}\right) \,\, | \,\, k=1,\ldots ,3r \right\} ~.
\end{split} 
\end{gather}

From the fact that $p_i$ are all prime, and assuming that $(1+\nu_i,p_i)=1$ for all $i$, 
this factor is nonvanishing for $$
2\prod_{i}(p_i-1)=
24r
$$
number of elements inside 
${\mathbb Z}/2m$, which can be parametrized by 
$$(\srho, \hat \epsilon)  ,\srho=(0,0,\rho_3)~{\rm with }~
 \rho_3\in\{0,1,\dots,3r-1\}, \hat\epsilon=(\epsilon_1,\epsilon_2,\epsilon_3)\in (\ZZ/2)^{ 3}
$$ as
\be \label{generalal1form} 
\alpha_1({\boldsymbol{\rho}},\hat\epsilon)= m - \sum_{i=1}^3
(-1)^{\epsilon_i} (1+\rho_i)\bar{p}_i \,\,\,\mathrm{mod}\,\,\,2m ~~~ .
\ee
Denote by $s_{\srho} = s(\alpha_1({\boldsymbol{\rho}},(0,0,0)))$, which gives a injective map from $(0,1,\dots,3r)$ to
$$
\left\{ \pm \frac{\sqrt{3}}{2} \sin\left(\frac{\pi k}{6r+1}\right) \,\, | \,\, k=1,\ldots ,3r \right\}
$$
we arrive at
\begin{gather}\begin{split}
     F_{(p_1,p_2,p_3),{\bnu},\chi} (\tau)
     \sim {\tilde\tau}  \sum_{s=1}^r \cos\left( \frac{\pi (2s-1) \chi}{2x} \right) \sum_{\substack{\srho=(0,0,\rho_3) \\
     \rho_3 =0,1,\dots,3r-1}} s_{\srho}\, F_{(p_1,p_2,p_3),{\srho},2s-1} (\tilde\tau)
\end{split}
\end{gather}
where we have combined the sum over $\hat \epsilon\in ({\mathbb Z}/2)^3$ into $F_{(p_1,p_2,p_3),{\srho},2s-1} $ using the definition of \eqref{def:F}.

\section{Polynomials \texorpdfstring{$P_{n}^{p,b}$}{Pnpb}}
\label{app:Pnpb}

{
Recall the definition of the polynomials $P_{n}^{p,b}$ \eqref{dfn:Ppoly}
: 
\begin{equation}
 P_{n}^{p,b}(q^{-1}):= {1\over f_{n}^{p,b}(q^{-1}) }\mathcal{L}_{-p}^{(b)}\left(
\frac{(x^{\ha}+x^{-\ha})^{\delta_b}}{ D_n} \right)
\end{equation}
with $f$ given in \eqref{dfn:f}
$$f_{n}^{p,b}(q^{-1}) = q^{{\frac{\delta_b}{2}}-\frac{b(2p-b)}{4p}} \frac{q^{n(n-\delta_b)}}{(q^{n+1-\delta_b};q)_n}.$$

As a result, we have 
\begin{equation}
 P_{n}^{p,b}(q^{-1}):= q^{{- \frac{\delta_b}{2}}+\frac{b(2p-b)}{4p}-n(n-\delta_b)}\frac{(q^{n+1-\delta_b};q)_n}{(q;q)_{2n}}  \left( (q;q)_{2n} \,\mathcal{L}_{-p}^{(b)}\left(
\frac{(x^{\ha}+x^{-\ha})^{\delta_b}}{ D_n} \right)\right)
\end{equation}

Following \cite{park-thesis} (Prop. 4.1.14) we obtain 
\begin{gather}\label{exp:p-V2}
\begin{split}
&(q;q)_{2n} \,\mathcal{L}_{-p}^{(b)}\left(
\frac{(x^{\ha}+x^{-\ha})^{\delta_b}}{ D_n} \right) =  
\sum_{k=0}^\infty  q^{-nk}\left(q^{k+1};q\right)_{2n} 
\\
&   
\left( q^{\left(n+k - \frac{\delta_{b}}{2}\right)^2\over p}
\delta_{n+k-b/2 - \frac{\delta_{b}}{2} \Mod{p}} - 
q^{\left(n+k+1+ \frac{\delta_{b}}{2}\right)^2\over p}\delta_{n+k+1-b/2 + \frac{\delta_{b}}{2} \Mod{p}} \right), 
\end{split}\end{gather}
which reduces to a finite sum due to telescoping. 
To see this  explicitly, one expands the $q$-factorial factor $(q^{k+1};q)_{2n}$ in the 1$^{st}$ line into a sum of $2^{2n}$ terms. Introducing the notation the sets
\begin{equation}
    \mathcal{S} :=\{\, 1,\, 2,\, \ldots, \, 2n \, \}~, \quad \mathcal{I}\subseteq\mathcal{S}
\end{equation}
and for a given $\mathcal{I}$ an element in $\{0,1\}^{2n}$
\begin{gather}\label{eq:prodexpansionqfactorial}
    \begin{split}
    \mathsf{s}_{\mathcal{I}}&:=(s_1,s_2,\ldots , s_{2n})_{\mathcal{I}}~, ~~ s_{i\in\mathcal{I}}=1 , \, s_{i\in \mathcal{S}\backslash \mathcal{I}}=0 ~,
    \end{split}
\end{gather}
a generic term in the expansion of $(q^{k+1};q)_{2n}$ can be labeled by its corresponding tuple $s_\mathcal{I}$, which prescribes selecting the factor
\begin{equation}
    \begin{cases} 
    -q^{k+i} ~ \mathrm{if} ~ s_i=1 \\
    1 \qquad ~ \mathrm{if} ~ s_i=0 ~ 
    \end{cases}
\end{equation}
 from the $i^{th}$ factor $(1-q^{k+i})$ inside the product $(q^{k+1};q)_{2n}$.
Moreover, each such term has an associated partner in the expansion which is labeled by the tuple
\begin{equation}
    \tilde{\mathsf{s}}_{\mathcal{I}}:= 
    (1-s_{2n},1-s_{2n-1},\ldots , 1-s_{1})_{\mathcal{I}}~.
\end{equation}
Pairing terms labeled by a tuple $s_{\mathcal{I}}$ from 
$$q^{-nk}\left(q^{k+1};q\right)_{2n} \,
q^{\left(n+k - \frac{\delta_{b}}{2}\right)^2\over p}
\delta_{n+k- \frac{b+\delta_{b}}{2} \Mod{p}}~,$$
with terms labeled by the tuple $\tilde{s}_{\mathcal{I}}$ from 
$$q^{-nk}\left(q^{k+1};q\right)_{2n} \,
q^{\left(n+k+1+ \frac{\delta_{b}}{2}\right)^2\over p}\delta_{n+k+1 - \frac{b-\delta_{b}}{2} \Mod{p}}~,$$
the sum over $k$ on the right hand side of equation \eqref{exp:p-V2} then telescopes and reduces to a finite interval for each set $\mathcal{I}$. Therefore by summing over all such sets $\mathcal{I}$, one arrives at 
\begin{gather}\label{eq-Ppolyeven5A}
   \begin{split}
         &
          (q;q)_{2n} \,\, \mathcal{L}^{(b)}_{-p} \left(\frac{1}{D_n}\right) = (1-\delta_{b}) \sum_{ \mathcal{I} }
          (-1)^{|\mathcal{I}|} \, q^{ C_{\mathcal{I}}} \,\mathsf{S}_{\mathcal{I}} 
          \\
         &
         \\
         &  \mathsf{S}_{\mathcal{I}} = \begin{cases} 
          \sum_{\ell=L}^{ L+\delta^{(n,b)}-A_{\mathcal{I}} - 1 }  
           q^{\frac{1}{p}( p\ell  + \frac{b}{2} )( p (\ell+A_{\mathcal{I}}) + \frac{b}{2} )  }
           \qquad , ~~ |\mathcal{I}| \leq \delta^{(n,b)} + n - 1
          \\
          - 
          \sum_{\ell=L+\delta^{(n,b)}}^{ L+A_{\mathcal{I}} - 1 }  
           q^{\frac{1}{p}( p(\ell - A_{\mathcal{I}}) + \frac{b}{2} )( p \ell + \frac{b}{2} ) } 
           \qquad ~~, ~~ |\mathcal{I}| \geq \delta^{(n,b)}+n+1
           \end{cases}~~, 
    \end{split}
\end{gather}
where we have introduced
\begin{equation}
    A_{\mathcal{I}}:=|\mathcal{I}| -n ~, \quad 
    C_{\mathcal{I}}:= \sum_{i\in\mathcal{I}} i - n A_{\mathcal{I}}~,
    \quad
    L:= \left\lceil \frac{n-b/2}{p} \right\rceil ~, \quad \delta^{(n,b)}:=\delta_{n-b/2,0 \, (\mathrm{mod}~p)}~,
\end{equation}
and 
\begin{gather}\label{eq-Ppolyeven5B}
   \begin{split}
         &
          (q;q)_{2n} \,\, \mathcal{L}^{(b)}_{-p} \left(\frac{x^{1/2}+x^{-1/2}}{D_n}\right) = \delta_{b} \sum_{ \mathcal{I} }
          (-1)^{|\mathcal{I}|} \, q^{ \widetilde{C}_{\mathcal{I}}} \,\widetilde{\mathsf{S}}_{\mathcal{I}} 
          \\
         &
         \\
         &  \widetilde{\mathsf{S}}_{\mathcal{I}} = \begin{cases} 
          \sum_{\ell=\widetilde{L}}^{ \widetilde{L}+\widetilde{\delta}^{(n,b)}-A_{\mathcal{I}} - 1 }  
           q^{\frac{1}{p}( p\ell  + \frac{b}{2} )( p (\ell+A_{\mathcal{I}}) + \frac{b}{2} )  }
           \qquad , ~~ |\mathcal{I}| \leq \widetilde{\delta}^{(n,b)} + n - 1
          \\
          - 
          \sum_{\ell=\widetilde{L}+\widetilde{\delta}^{(n,b)}}^{ \widetilde{L}+A_{\mathcal{I}} - 1 }  
           q^{\frac{1}{p}( p(\ell - A_{\mathcal{I}}) + \frac{b}{2} )( p \ell + \frac{b}{2} ) } 
           \qquad ~~, ~~ |\mathcal{I}| \geq \widetilde{\delta}^{(n,b)}+n+1
           \end{cases}~~, 
    \end{split}
\end{gather}
with
\begin{equation}
    \widetilde{C}_{\mathcal{I}}:= \sum_{i\in\mathcal{I}} i - \left(n-\tfrac{1}{2} \right) A_{\mathcal{I}}~,
    \quad L:= \left\lceil \frac{n-(b+1)/2}{p} \right\rceil ~, \quad 
    \widetilde{\delta}^{(n,b)}:=\begin{cases}
       1& p=2\\
       \delta_{n-(b+1)/2~(p)} & p\geq 3
    \end{cases} ~~.
\end{equation}

\bigskip

We record the coefficient lists of the first few $P^{p,b}_n$ in variable $q$ and starting with $q^0$. Note that  $P^{p,b}_1=1$ for all $p$ and $b$ and it is hence not listed in the tables \ref{tab:P-polyCoeff0} and \ref{tab:P-polyCoeff1}.

\begin{table}[ht]
  \centering
\caption{The coefficient lists of $P^{2,b}_n$. }
  \begin{tabular}{ l l l}
$b$ & $n$ & coefficients \\ \toprule
\multirow{5}{*}{$0$} 
& 2&$1,0,1$\\ 
& 3&$1,0,1,1,1$\\ 
& 4&$1,0,1,1,2,1,1,0,1$\\
& 5&$1, 0, 1, 1, 2, 2, 2, 1, 2, 1, 1, 1, 1$\\
& 6&$1, 0, 1, 1, 2, 2, 3, 2, 3, 2, 3, 2, 3, 2, 2, 1, 1, 0, 1
$\\\midrule
\multirow{5}{*}{$1 $} 
& 2&$1,1$\\
& 3&$ 1,1,1,1$\\
& 4&$ 1,1,1,2,1,1,1$\\
& 5&$ 1,1,1,2,2,2,2,2,1,1,1$\\
& 6&$1,1,1,2,2,3,3,3,3,3,3,2,2,1,1,1$\\\midrule
\multirow{5}{*}{$2 $} 
& 2&$1,1$\\ 
& 3&$1, 1, 1, 0, 1$\\ 
& 4&$1, 1, 1, 1, 1, 1, 1, 1$\\ 
& 5&$1, 1, 1, 1, 2, 1, 2, 2, 2, 1, 1, 0, 1$\\
& 6&$1, 1, 1, 1, 2, 2, 2, 3, 3, 3, 3, 2, 2, 2, 1, 1, 1, 1$\\\midrule
\multirow{5}{*}{$3 $} 
&2&$1,1$\\
&3&$1,1,1,1$\\
&4&$1,1,1,2,1,1,1$\\
&5&$1,1,1,2,2,2,2,2,1,1,1$\\
&6&$1,1,1,2,2,3,3,3,3,3,3,2,2,1,1,1$\\
\bottomrule
\end{tabular}
\label{tab:P-polyCoeff0}
\end{table}

\begin{table}[ht]
  \centering
\caption{The coefficient lists of $P^{3,b}_n$. }
  \begin{tabular}{ l l l}
$b$ & $n$ & coefficients \\ \toprule
\multirow{5}{*}{$0$} 
&2&$1,0,1,1$\\&3&$1,0,1,2,2,1,2$\\&4&$1,0,1,2,3,3,4,3,3,3,2,1,1$\\&5&$1,0,1,2,3,4,6,5,7,7,8,7,8,5,6,4,3,2,2$\\&6&$1,0,1,2,3,4,7,7,9,11,13,14,17,16,18,18,17,15,16,13,11,9,7,5,5,2,1,1$\\
\midrule
\multirow{5}{*}{$1 $} 
&2&$1,1,1$\\&3&$1,1,2,2,2,1$\\&4&$1,1,2,3,4,4,4,3,3,1,1$\\&5&$1,1,2,3,5,6,7,8,9,8,8,7,6,4,3,2,1$\\&6&$1,1,2,3,5,7,9,11,14,15,18,18,20,19,19,17,16,13,11,8,7,4,3,1,1$\\
\midrule
\multirow{5}{*}{$2 $} 
& 2 & $1,1,1$ \\
& 3 & $1,1,2,1,2,1,1$ \\
& 4 & $1,1,2,2,3,3,4,3,3,2,2,1$ \\
& 5 & $1,1,2,2,4,4,6,6,8,7,8,7,7,5,5,3,3,1,1$ \\
& 6 & $1,1,2,2,4,5,7,8,11,12,15,15,17,17,18,17,17,15,14,11,10,8,6,4,3,2,1$ \\
\midrule
\multirow{5}{*}{$3 $} 
& 2 & $1,2$ \\&
 3 & $1,2,2,2,1,1$ \\&
 4 & $1,2,2,4,3,4,4,3,2,2$ \\&
 5 & $1,2,2,4,5,6,7,9,8,9,8,6,5,5,2,1,1$ \\&
 6 & $1,2,2,4,5,8,9,12,14,17,18,19,19,20,18,17,14,13,10,8,5,4,2,2$ \\
\bottomrule
\end{tabular}
\label{tab:P-polyCoeff1}
\end{table}

\section{{Proof of Lemma \ref{lem:LambdaPoly}}}
\label{app:prooflem2}

Write $X=x+x^{-1}$, $Q_j=q^j+q^{-j}$, and $D_m = \prod_{j=1}^m (X-Q_j)$. 
From the recursion relation 
\begin{equation}
    \frac{X^m}{D_n}=X^{m-1}\left(\frac{Q_n}{D_n}+\frac{1}{D_{n-1}} \right)
\end{equation}
 we obtain
\begin{lem}\label{simple_recursion}
\begin{equation}\label{binomial_recursion_nsmall}
\frac{X^n}{D_m} = 
S_{n-m}(X,Q_m,Q_{m-1},\dots, Q_1) +  \sum_{\ell=0}^{{\rm min}(n,m-1)} \frac{S_{n-\ell}(Q_m,\dots,Q_{m-\ell})}{D_{m-\ell}},~~ {\rm{for }}~n,  m \in {\mathbb N}. 
\end{equation}

In the above, $S_n$ denotes the complete homogeneous symmetric polynomial of degree $n$: $$S_n(x_1,\dots,x_k) = \sum_{\substack{i_\ell \geq 0\\ \sum_{\ell=1}^k i_\ell =n}}  x_1^{i_1}\dots x_k^{i_k},  ~{\rm for}~ n\geq 0~, ~~S_n(x_1,\dots,x_k) =0 ~{\rm for}~n<0. $$

\end{lem}

\begin{proof}

We first show that 
\begin{equation}
\frac{X^n}{D_m} ={\rm quot}_n(X,q,q^{-1}) +  \sum_{\ell=0}^{{\rm min}(n,m-1)} \frac{X^{n}[Q_m,\dots,Q_{m-\ell}]}{D_{m-\ell}}
\end{equation}
with some  polynomial ${\rm quot}_n(X,q)$ of $X$, $q$, and $q^{-1}$ and
where $f[x_1,\dots,x_k]$ denotes the divided differences which are defined recursively by
\[
f[x_0,\dots,x_k] := \frac{f[x_0,\dots,x_{k-1}] -f[x_1,\dots,x_{k}]}{x_0-x_k} ~.
\]

First, write 
\begin{gather}\label{eqn:coeff1}
\begin{split}
{X^n} &={\rm quot}_n(x,q) D_m +  \sum_{\ell=0}^{{\rm min}(n,m-1)} \lambda_{m-\ell}(q) \frac{D_m}{D_{m-\ell}}\\
&= {\rm quot}_n(x,q) D_m + \lambda_{m}(q) +\lambda_{m-1}(q) (X-Q_m)+ 
\lambda_{m-2}(q) (X-Q_m)(X-Q_{m-1})+\dots
\end{split}\end{gather}
We obtain $\lambda_m = Q_m^n=X^n[Q_m]$ by plugging in $X=Q_m$ and using the fact that 
$$\frac{D_m}{D_{m-\ell}}(X,q,q^{-1})\Big\lvert_{X=Q_m} =0,~~ \ell>0. 
$$
Next, using $\lambda_m=X^n[Q_m]$ and plugging in $X=Q_{m-1}$ we obtain 
$\lambda_{m-1}=X^n[Q_m,Q_{m-1}]$. 

Next, we assume that $\lambda_{m-\ell}=X^n[Q_m,Q_{m-1},\dots,Q_{m-\ell}]$ for all $\ell=0,1,\dots, L-1$ for some $ 0 <\ell < {\rm min}(n,m-1)$. 
Then using the recursive definition of the divided differences $f[x_0,\dots,x_k]$ 
we obtain that $\lambda_{m-L}=X^n[Q_m,Q_{m-1},\dots,Q_{m-L}]. $
From this we obtain the proof by induction that 
$\lambda_{m-\ell} =  X^n[Q_m,Q_{m-1},\dots,Q_{m-\ell}]$ for all $\ell\in\{0,1,\dots, {\rm min}(n,m-1)\}$ in \eqref{eqn:coeff1}. Putting these expressions for $\lambda_{m-\ell}$ back to \eqref{eqn:coeff1}, we obtain 
\[{\rm quot}_n(x,q) =X^n[X,Q_1,\dots,Q_m].\]
The Lemma \ref{simple_recursion} then follows from
 the equality $X^n[x_0,x_1,\dots,x_k] =S_{n-k}(x_0,x_1,\dots,x_k)$.  
\end{proof}

Write, for $n\in \NN$, \(y_n := \chi_{2n}(x^\ha)=\sum_{k=0}^{n} x^{\frac{n}{2}-k}\). From the recursion
\[
y_n = X y_{n-1}-y_{n-2}
\]
and 
the boundary condition $y_0=1$, $y_1= X+1$. We obtain
\begin{equation}\label{character_recursion}
y_n = \sum_{j=0}^{n}  c^{(n)}_j X^{n-j}, \quad c^{(n)}_{j} \;=\; (-1)^{\lfloor {j\over 2}\rfloor}\,
\binom{\,n-\lceil {j\over 2}\rceil\,}{\,\lfloor {j\over 2}\rfloor\,}
\end{equation}
which can be proven by induction. 
Combining \eqref{character_recursion} and \eqref{binomial_recursion_nsmall}, we obtain that for any positive integers $n$, $m$ we have 
\begin{equation}\label{character_nsmall_even}
{\chi_{2n}(x^\ha)\over D_m} =\sum_{j=0}^{n} c_j^{(n)} S_{n-m-1-j}(X,Q_m,\dots,Q_{1})
 + \sum_{\ell=0}^{{\rm min}(n,m-1)} {1\over D_{m-\ell} } \sum_{j=0}^{ {n-\ell}} c_j^{(n)} S_{n-\ell-j}(Q_m,\dots,Q_{m-\ell}).
\end{equation}
Similarly, write $$y_{n+\ha} := {{\chi_{2n+1}(x^\ha)}\over  x^{\ha} +x^{-\ha}}.$$ From the recursion
\[
y_{n+\ha} = X y_{n-\ha}-y_{n-{3\over 2}}
\]
and 
the boundary condition $y_{1/2}=1$, $y_{3/2}= X$. We obtain
\begin{equation}\label{character_recursion2}
y_{n+\ha} = \sum_{j=0}^{n}  c^{(n+\ha)}_{j} X^{n-j}, \quad  c^{(n+\ha)}_{k} \;=\; (-1)^{k}\,
\binom{\,n-k\,}{\,k\,}, ~c^{(n+\ha)}_{2k+1}=0. 
\end{equation}
which can be proven by induction. Combining \eqref{character_recursion2} and \eqref{binomial_recursion_nsmall}, we obtain that for any positive integers $n$, $m$ we have 
\begin{gather}\label{character_nsmall_even2}
\begin{split}
&(  x^{\ha} +x^{-\ha})^{-1} {\chi_{2n+1}(x^\ha)\over D_m}  \\
&= \sum_{j=0}^{n} c_j^{(n+\ha)} S_{n-m-j}(X,Q_m,\dots,Q_1)
 + \sum_{\ell=0}^{{\rm min}(n,m-1)} {1\over D_{m-\ell} } \sum_{j=0}^{ {n-\ell}} c_j^{(n+\ha)} S_{n-\ell-j}(Q_m,\dots,Q_{m-\ell})
\end{split}
\end{gather}

\bibliographystyle{JHEP}
\bibliography{3dre_bib}

\end{document}